\documentclass[]{article}

\usepackage[
	a4paper,
	margin=1.2in,
	tmargin=0.8in,
	bmargin=0.8in,
	includeheadfoot,
	headheight=15pt]{geometry}

\usepackage{graphicx}%
\usepackage{multirow}%
\usepackage{amsmath,amssymb,amsfonts}%
\usepackage{amsthm}%
\usepackage{mathrsfs}%
\usepackage[title]{appendix}%
\usepackage[dvipsnames]{xcolor}%
\usepackage{textcomp}%
\usepackage{manyfoot}%
\usepackage{booktabs}%
\usepackage{algorithm}%
\usepackage{algorithmicx}%
\usepackage{algpseudocode}%
\usepackage{listings}%
\usepackage{hyperref}

\hypersetup{
    colorlinks=true,
    linkcolor=blue,
	citecolor=blue
    }

\usepackage{tikz}
\usepackage{mathtools}
\usepackage{todonotes}
\usepackage{pifont}
\usepackage{graphicx}

\usetikzlibrary { decorations.pathmorphing, decorations.pathreplacing, decorations.shapes, }
\usetikzlibrary{shapes.geometric}
\usetikzlibrary{fit}
\usetikzlibrary{calc}

\theoremstyle{thmstyleone}%
\newtheorem{theorem}{Theorem}%
\newtheorem{definition}{Definition}%
\newtheorem{lemma}{Lemma}

\newtheorem{observation}{Observation}
\newtheorem{corollary}{Corollary}

\DeclareMathOperator{\troot}{root}
\DeclareMathOperator{\bdt}{BDT}
\DeclareMathOperator{\obdt}{OBDT}
\DeclareMathOperator{\bd}{BD}
\DeclareMathOperator{\edges}{edges}

\newcommand{\cupdot}{\mathbin{\mathaccent\cdot\cup}}
\newcommand{\mathbff}[1]{\boldsymbol{\mathbf{#1}}}

\begin{document}

\title{Categorizing Merge Tree Edit Distances by Stability using Minimal Vertex Perturbation}

\author{Florian Wetzels and Christoph Garth,\\University of Kaiserslautern-Landau}

\maketitle

\begin{abstract}
\noindent
This paper introduces a novel stability measure for edit distances between merge trees of piecewise linear scalar fields.
We apply the new measure to various metrics introduced recently in the field of scalar field comparison in scientific visualization.
While previous stability measures are unable to capture the fine-grained hierarchy of the considered distances, we obtain a classification of stability that fits the efficiency of current implementations and quality of practical results.
Our results induce several open questions regarding the lacking theoretical analysis of such practical distances.
\end{abstract}

\section{Introduction}

In scientific visualizations and data analysis, comparative analysis of scalar fields through distance measures on topological descriptors has received increasing interest in recent years~\cite{surveyComparison2021,DBLP:journals/corr/abs-2110-05631}.
Many distances for various abstractions have been defined, reaching from persistence diagrams~\cite{Cohen-Steiner2007,edelsbrunner09,vidal_vis19,interleaving_distance} to contour trees~\cite{DBLP:journals/cgf/LohfinkWLWG20,DBLP:journals/tvcg/ThomasN11} and Reeb graphs~\cite{DBLP:conf/3dor/BauerFL16,DBLP:journals/entcs/FabioL12,DBLP:journals/dcg/FabioL16}, or even geometry-incorporating abstractions such as extremum graphs~\cite{DBLP:conf/apvis/NarayananTN15,ThomasN13}.
The merge tree abstraction~\cite{DBLP:journals/tvcg/BollenTL23,morozov14,intrinsicMTdistance,Yan_geometry_aware,oesterling_time_varying_mergetrees,BeketayevYMWH14,DBLP:journals/corr/abs-2404-05879} has received particular interest, specifically so-called edit distances~\cite{DBLP:journals/cgf/SaikiaSW14,DBLP:journals/tvcg/SridharamurthyM20,wetzels2022branch,DBLP:journals/tvcg/PontVDT22,taming} between them.
Such merge tree edit distances have been applied successfully for deriving novel visualization and analysis techniques, e.g.\ for interpolation~\cite{DBLP:journals/tvcg/PontVDT22,DBLP:journals/tvcg/WetzelsPTG24,DBLP:journals/tvcg/PontVT23}, embedding~\cite{DBLP:journals/tvcg/PontT24} or summarization~\cite{DBLP:journals/tvcg/PontVDT22,wetzels2022path,DBLP:journals/cgf/LohfinkWLWG20,DBLP:journals/tvcg/WetzelsPTG24,10.2312:evs.20241069} of ensembles and time series.
Furthermore, edit distances come in a large variety of constrained forms, i.e.\ by restricting the allowed operations.
Those usually trade computational complexity for expressiveness of the distance, making them a powerful tool in practice.
While there is a plethora of different merge tree edit distances (often 
including 
publicly available 
implementations~\cite{DBLP:journals/tvcg/TiernyFLGM18,DBLP:journals/tvcg/PontVDT22,wetzels2022path,taming}) with vastly diverging complexity and expressiveness, theoretic results on their quality are sparse.
One very important property of such distances defined on abstractions is stability: do small changes in the domain only induce small changes in the abstraction, and specifically small distances in the applied metric?
Though some stability criteria have been proposed for merge tree edit distances~\cite{DBLP:journals/cgf/SaikiaSW14,taming} and tested experimentally~\cite{taming,wetzels2022path,wetzels2022branch}, formal theoretical analyses of stability remain underexplored.
This stands in contrast to other types of metrics such as interleaving distances~\cite{morozov14,interleaving_distance} or related ones~\cite{DBLP:journals/tvcg/BollenTL23}, as well as other abstractions~\cite{DBLP:conf/3dor/BauerFL16,DBLP:journals/entcs/FabioL12,DBLP:journals/focm/BauerLM21}, for which a formal study of stability has seen more interest.
To fill in this gap, we propose a novel stability measure that overcomes limitations of established ones when working with merge tree edit distances.

\begin{figure}
    \centering
    \resizebox{0.6\linewidth}{!}{
    \Large
    \begin{tikzpicture}[xscale=1,yscale=1,
    box1/.style = {draw,red,ultra thick,inner sep=0.2pt,rounded corners=15pt},
    box2/.style = {draw,red,ultra thick,inner sep=0.2pt,rounded corners=15pt},
    box3/.style = {draw,red,ultra thick,inner sep=5pt,rounded corners=16pt}]
    
    \node[circle] at (0, 0) (sc) {\bfseries SC};
    \node[circle] at (-2, 2) (vs) {\bfseries VS};
    \node[circle] at (2, 2) (es) {\bfseries ES};
    \node[circle] at (0, 4) (hs) {\bfseries HS};
    
    \node[blue] at (1, 1) (bdt) {\small $\mathbff{\theta(n^2)}$};
    \node[green!70!black] at (-1, 1) (clas) {\small $\mathbff{\theta(n^2)}$};
    \node[red] at (0, 2.1) (bdi) {\small $\mathbff{\mathcal{O}(n^4)}$};
    \node[black] at (0, 3.2) (bdi) {\small NP-c};

    \node[fill=blue!80,label={[right,yshift=-4pt,xshift=3pt]BDT-based}] at (3.5, 4.1-1.2) (bdt_l) {};
    \node[fill=green!80!black,label={[right,yshift=-4pt,xshift=3pt]Classic constrained}] at (3.5, 3.5-1.2) (clas_l) {};
    \node[fill=red,label={[right,yshift=-4pt,xshift=3pt]Branch/Path Mapping}] at (3.5, 2.9-1.2) (bdi_l) {};
\node[fill=black,label={[right,yshift=-4pt,xshift=3pt]Deformation-based}] at (3.5, 2.3-1.2) (bdt_l) {};

    \draw [red,ultra thick] ($ (sc) + (-0.42,-0.42) $) arc(225:315:0.59) --($ (es) + (0.42,-0.42) $) arc(-45:90:0.59) -- ($ (vs) + (0,0.59) $) arc(90:225:0.59) -- cycle;
    \node[box1,green!80!black,rotate fit=45,fit=(sc)(vs)] {};
    \node[box2,blue!80,rotate fit=45,fit=(sc)(es)] {};
    \node[box3,black,rotate fit=45,fit=(sc)(es)(hs)] {};
    
    \end{tikzpicture}
    }
    \caption{Illustrative visualization of the stability-based hierarchy of the different edit distances with their complexity annotated as well. A class of edit distance encloses a type of perturbation if it is stable against them.}
    \label{fig:hierarchy}
\end{figure}
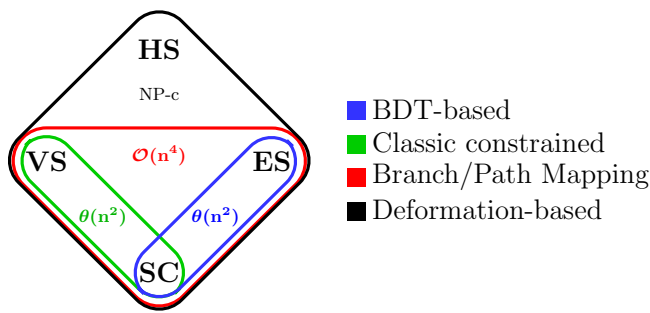

\smallskip
\noindent
\textbf{Fine-grained Stability.}
The varying complexity and expressiveness of merge tree edit distances appear to be closely linked to their stability properties.
A fine-grained understanding of stability is essential — not only to assess their practical limitations but also to inform the design of new methods that strike a balance between robustness and efficiency.
In many cases, enforcing stronger stability comes at the cost of increased computational complexity, making certain distances impractical for large-scale applications.
Understanding these trade-offs in detail is crucial for developing methods that remain both robust and practically viable.

For example, Wetzels et al.\ experimentally observed that a constrained variant of the so-called deformation-based edit distance exhibits instabilities in the presence of saddle swaps, while its unconstrained counterpart does not~\cite{taming}.
This insight directly led to the development of a heuristic algorithm~\cite{DBLP:journals/corr/abs-2501-05529} capable of handling saddle swaps in a controlled manner, improving result quality while maintaining feasible runtime performance.
More broadly, the development of deformation-based edit distances~\cite{wetzels2022branch,wetzels2022path,taming} has also been driven by such insights into the behavior in specific situations.

However, these findings remain largely experimental, lacking a formal theoretical foundation.
To address this, we introduce a new stability measure that examines the behavior of edit distances under minimal perturbations of single vertices.
This finer level of analysis reveals the specific conditions under which a given distance remains stable, offering deeper insight into its suitability for different types of data.
By making stability an integral part of method development, we aim to provide a principled foundation for designing more effective and efficient distance measures.

\smallskip
\noindent
\textbf{Max vs Sum.}
Another limitations of established stability measures that the new model overcomes is their inherent incompatibility with edit distances:
The $L_\infty$-norm measures the \emph{maximum} deviation between two data points, whereas edit distances typically measure the \emph{sum} of all deviations.
This makes formal stability for edit distances unrealistic\footnote{Excluding distances such as reeb graph edit distances that break the typical edit distance framework by taking the maximum deviation, though in a hidden manner.}.
By considering perturbations of single vertices of the scalar field, we overcome this issue without changing the basic paradigm of the merge tree distance.
This is important for applications: distances based on the sum have been proven to be powerful tools, in particular for visualization purposes.
This can be seen in the popularity of the Wasserstein distance for persistence diagrams or the aforementioned visualization methods derived from edit distances.
In bottleneck distances, practically relevant changes can often be overshadowed by the bottleneck, limiting them in certain applications.

\smallskip
\noindent
\textbf{Contribution.}
To be more precise, we define minimal vertex perturbations on piecewise linear scalar fields (Section~\ref{sec:transformations}) and provide a classification of those into four types: simple changes (SC), edge splits (ES), vertical branch swaps (VS) and horizontal branch swaps (HS).
We then analyze which distances are stable against which types of perturbations (Section~\ref{sec:perturbation_stability}).
As our main result, we obtain a hierarchy of stability (see Figure~\ref{fig:hierarchy}) on the distances which nicely fits their complexity and the running times of existing implementations: algorithms able to handle \emph{either} vertical swaps \emph{or} edge splits have quadratic running time, handling both currently requires quartic time, whereas handling horizontal instabilities leads to NP-hardness.

Our findings on minimal vertex perturbations also transfer to more general data perturbations, as discussed in Section~\ref{sec:sequence_stability}.
The stability properties derived here have been called \emph{local} or \emph{finite} stability by Pegoraro in previous works~\cite{pegoraro2024finitelystableeditdistance}.

Our results raise several open questions, discussed in Section~\ref{sec:conclusion} along with a refined hierarchy and future work.
Overall, this paper provides a foundational framework for a more fine-grained view on stability of topological descriptors necessary to capture the nuances of practically applicable distance measures.

\smallskip
\noindent
\textbf{Note on practical relevance.}
While stability against minimal vertex perturbation appears as a weak property at first glance, the provided hierarchy indeed matches the performance in practical settings.
In scientific visualization, persistence-based topological simplification~\cite{DBLP:conf/focs/EdelsbrunnerLZ00} is a well-established preprocessing step for noisy data.
In such a scenario, stability against insertion of small actual features (as opposed to noise) is practically more relevant and instabilities have been repeatedly shown to inhibit interpretability of results~\cite{wetzels2022branch,wetzels2022path,taming,DBLP:journals/tvcg/WetzelsPTG24,wetzels24evolution}, even without the presence of noise.
Figure~\ref{fig:tosca} shows the impact of stability on result quality on an established benchmark dataset.

\begin{figure}[t!]
    \centering
    \includegraphics[width=\linewidth]{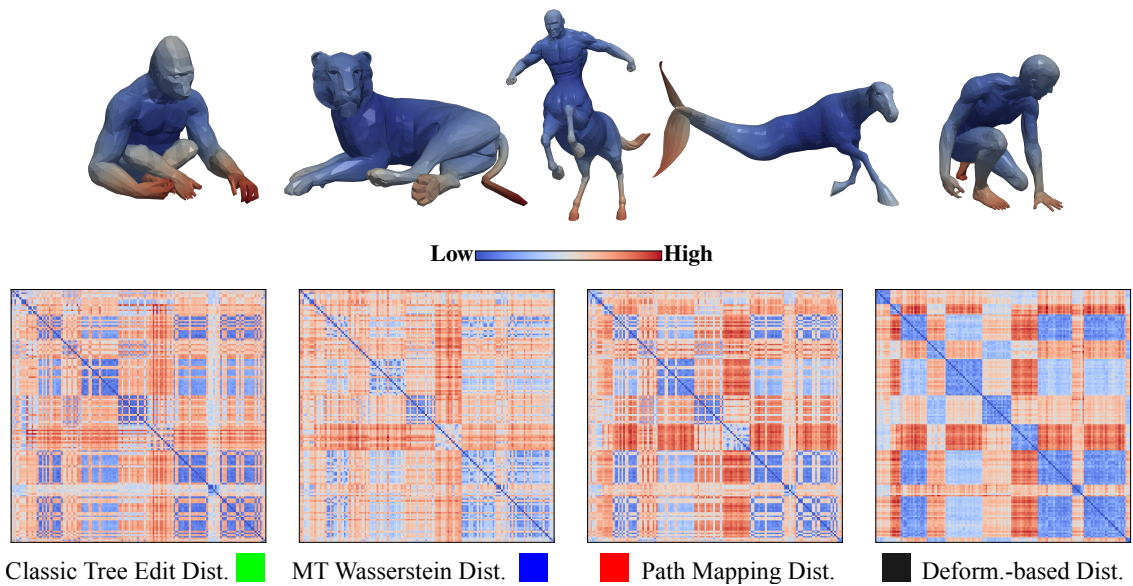}
    \caption{Illustration of the considered stability properties in practice: The top row shows members of the TOSCA shape matching ensemble~\cite{DBLP:series/mcs/BronsteinBK09}, an established benchmark dataset for distance measures on topological descriptors~\cite{DBLP:journals/tvcg/SridharamurthyM20,DBLP:journals/tvcg/WetzelsPTG24}. The 3D meshes describing various animal shapes in differing poses have an average geodesic field~\cite{HilagaSKK01} attached, visualized through the cool-warm color map. The bottom row shows heatmaps (same color map) for distance matrices based on various metrics. Entry $(i,j)$ represents the distance between the $i$-th and $j$-th member of the ensemble, which is ordered lexicographically by shape and pose. We observe that more stable metrics yield ``cleaner'' matrices, i.e.\ clearly separated clusters representing a single shape each. Note how off-diagonal areas of low distance show high similarity between e.g.\ different human shapes and the gorilla. Colors next to distances indicate the types from Figure~\ref{fig:hierarchy}. All matrices were computed using the existing implementations provided with the corresponding papers~\cite{taming,wetzels2022path,DBLP:journals/tvcg/PontVDT22}.}
    \label{fig:tosca}
\end{figure}

\section{Preliminaries}
\label{sec:preliminaries}

In this section, we recap definitions of basic notations and topological concepts needed in this work.
For a proper introduction into computational topology, see~\cite{edelsbrunner09}.
A piecewise linear scalar field is a simplicial complex $\mathbb{X}$ (see~\cite{edelsbrunner09} for a definition) together with an injective function defined on its vertices, $f: V(\mathbb{X}) \rightarrow \mathbb{R}$
(we use $V(G),E(G)$ for the vertices and edges of any graph or complex).
Next, we define merge trees  of scalar fields.
For ease of notation, we define merge trees as split trees, but they can be easily adapted for join trees.
First, we introduce notations for trees in general.
We consider trees using parent edges, i.e.\ an unordered tree $T$ is a connected, directed graph without any undirected cycles and a unique sink, denoted $\troot(T)$.
We denote the children of a node $v$ by $C_T(v)$, its degree by $\deg_T(v) = |C_T(v)|$ and the maximum degree of all nodes in $T$ by $\deg(T)$.
An \emph{ordered} tree is an unordered tree together with a partial ordering for the children of each node.
We write $T_1 \subseteq T_2$ if $T_1$ is a subgraph of $T_2$, i.e.\ $V(T_1) \subseteq V(T_2)$ and $E(T_1) \subseteq E(T_2)$.
We write $T_1 \subset T_2$ if $T_1 \subseteq T_2$ and one of the two inclusions is strict, i.e.\ either $V(T_1) \subset V(T_2)$ or $E(T_1) \subset E(T_2)$.
On ordered trees, the ordering needs to be consistent for the subtree property to hold.

Given a piecewise linear scalar field $\mathbb{X},f$, we define its augmented merge tree $\mathcal{T}_f$ as follows. 
The nodes of $\mathcal{T}_f$ are the vertices of $\mathbb{X}$ (or the preimage of $f$), i.e.\ $V(\mathcal{T}_f) = V(\mathbb{X})$.
There is an edge $(y,x) \in E(\mathcal{T}_f)$ if and only if $y$ is the lowest vertex in some connected component $Y$ of the superlevel set $f^{-1}((f(x),\infty))$ and $\{x,y'\} \in \mathbb{X}$ for some $y' \in Y$.

The root of $\mathcal{T}_f$ is called a minimum node.
Any other vertex $v \in \mathbb{X}$ with $\deg_{\mathcal{T}_f}(v) = 1$ is called a regular vertex.
If $\deg_{\mathcal{T}_f}(v) = 0$, it is called a maximum vertex, and a saddle if $\deg_{\mathcal{T}_f}(v) > 1$.
Note that we treat local minima of the scalar field as regular vertices (with the global minimum as an exception), as they behave equivalently in the context of augmented split trees.

The augmented merge tree $\mathcal{T}_f$ can be reduced to the unaugmented merge tree $T_f$ by pruning all regular vertices.
Here, pruning a regular vertex $y$ means replacing edges $(x,y),(y,z)$ by the edge $(x,z)$ and then removing $y$.
If the underlying space $|\mathbb{X}|$ of the simplicial complex is a $d$-manifold with $d \geq 2$, $T_f$ is an unordered, node-labeled tree with $\deg_{T_f}(\troot(T_f))=1$ and  $\deg_{T_f}(v)\neq1$ for all $v \in V(T_f)$ with $v \neq \troot(T_f)$.
In~\cite{wetzels2022branch,wetzels2022path}, such trees are called \emph{abstract merge trees}.
For the remainder of this paper, we assume such scalar fields.

In practical applications, merge trees are often considered as their \emph{branch decompositions}.
A \emph{path} of length $k \geq 2$ in a tree $T$ is a sequence of vertices $p=v_1 ... v_k \in V(T)^k$ with $(v_{i},v_{i-1}) \in E(T)$ for all $2 \leq i \leq k$. 
We denote the edges within a path $p$ by $\edges(p) \coloneqq \{(v_{i},v_{i-1}) \bigm| 2 \leq i \leq k\}$.
A \emph{branch} of $T$ is a path that ends in a leaf.
A branch $b=b_1 ... b_k$ is called a parent branch of another branch $a=a_1 ... a_l$ (and $a$ a child branch of $b$) if $a_1 = b_i$ for some $1 < i < k$.
A set of branches $B=\{B_1,...,B_k\}$ of a merge tree $T$ is called a branch decomposition of $T$ if $\{\edges(B_1),...,\edges(B_k)\}$ is a partition of $E(T)$.
Every branch decomposition $B$ of an abstract merge tree $T$ contains exactly one branch $b \in B$ with $\troot(T) \in b$: the \emph{main branch} of $B$.
The persistence-based branch decomposition or simply \emph{the} branch decomposition of $T_f$, denoted $\bd(T_f)$, is derived by the elder rule~\cite{edelsbrunner09}: we track branches starting from maxima; when merging multiple branches at a saddle, we always continue the longest branch and stop the others.
For a branch $b=b_1\dots b_k \in \bd(T_f)$, we define $\mathcal{B}_f(v) = b$ if $v=b_i$, $i>1$.
The values $f(b_1)$ and $f(b_k)$ are called death and birth values of $b$.
The \emph{branch decomposition tree} of a merge tree $T_f$, $\bdt(T_f)$, is the unordered tree structure defined by the persistence-based branch decomposition $\bd(T_f)$ as vertex set and the parent-child relation of branches.
The starting nodes of child branches in a branch decomposition induce a partial ordering $<_{\text{B}}$ on the children in $\bdt(T_f)$ as follows.
Let $a = a_1\dots a_k$, $b = b_1\dots b_l$ be child branches of $c = c_1 \dots c_m$, s.t.\ we have $a_1 = c_i$ and $b_1 = c_j$.
We have $a <_{\text{B}} b$ if $i < j$.
The \emph{ordered} branch decomposition tree of $T_f$, denoted $\obdt(T_f)$ is the ordered tree defined by $\bdt(T_f)$ and $<_{\text{B}}$.

\subsection*{Edit Distances}
Given a set of edit operations transforming a labeled tree into another one, tree edit distances define the distance between two trees $T_1,T_2$ as the total costs of a cost-optimal edit sequence transforming $T_1$ into $T_2$.
They
have been introduced by Zhang for both ordered~\cite{DBLP:journals/siamcomp/ZhangS89} and unordered rooted trees~\cite{DBLP:journals/ipl/ZhangSS92}, using node insertion, deletion and relabel as edit operations.
Due to high computational complexity, they often come in constrained forms, allowing deletions and insertions only in specific cases, see~\cite{treeEditSurvey} for an overview.
Many applications to merge trees exist, using specific tree structures, labels of nodes and cost functions on edit operations.
Most of them use the same edit operation set as the original distance by Zhang, to which we refer as \emph{classic edit distances}.
In contrast, Wetzels and Garth formulated a new set of edit operations tailored to merge trees specifically, the \emph{deformation-based edit distance}~\cite{wetzels2022path,taming}.
An almost identical set of edit operations has also been proposed by Pegoraro and Secchi~\cite{pegoraro2024functionaldatarepresentationmerge,pegoraro2024finitelystableeditdistance}.
Note that while we ignore methods that are called edit distance but break the typical framework, like the Reeb graph edit distance or the merge tree matching distance, the deformation-based edit operations are closely related to those.

The \textbf{classic set of edit operations} (see Figure~\ref{fig:edit_operations_classic}) on node-labeled trees is defined as follows.
A node deletion removes a node $v$ as a child of $p$ and replaces each edge $(c,v)$ by $(c,p)$ for all $c \in C_T(v)$.
A node insertion in the inverse operation.
A node relabel changes the label $\ell(v)$ of a single node $v$.
On ordered trees, all edit operations preserve the ordering of the children.
The cost of a relabel operation changing the label of a node or edge from $\ell_1$ to $\ell_2$ is denoted by $c(\ell_1,\ell_2)$.
The cost of an insertion of a node or edge with label $\ell_2$ is denoted by the same cost function, where the first label is replaced by an empty symbol representing a non-existing node or edge, $c(\bot,\ell_2)$, a deletion symmetrically, $c(\ell_1,\bot)$.
The label and cost function and the empty symbol can be chosen freely depending on the application.

\begin{figure}
 \centering 
 \resizebox{0.9\linewidth}{!}{
 \begin{tikzpicture}[yscale=0.8]
 \definecolor{lightred}{rgb}{1,0.6,0.6}
 
 \begin{scope}[shift={(0,0)}]
 
 \node[draw,circle,fill=gray] at (0, 4) (root) {\bfseries A};
 \node[draw,circle,fill=cyan] at (-2, 2) (n1) {\bfseries B};
 \node[draw,circle,fill=green!90!black] at (0, 2) (n2) {\bfseries C};
 \node[draw,circle,fill=orange] at (2, 2) (n3) {\bfseries D};
 \node[draw,circle,fill=gray] at (-1, 0) (c1) {\bfseries E};
 \node[draw,circle,fill=gray] at (1, 0) (c2) {\bfseries F};
 
 \draw[gray,ultra thick] (root) -- (n1);
 \draw[gray,ultra thick] (root) -- (n2);
 \draw[gray,ultra thick] (root) -- (n3);
 \draw[gray,ultra thick] (n2) -- (c1);
 \draw[gray,ultra thick] (n2) -- (c2);
     
 \end{scope}
 
 \begin{scope}[shift={(6,0)}]
 
 \node[draw,circle,fill=gray] at (0, 4) (root) {\bfseries A};
 \node[draw,circle,fill=cyan] at (-2, 2) (n1) {\bfseries B};
 \node[draw,circle,fill=orange] at (2, 2) (n3) {\bfseries D};
 \node[draw,circle,fill=gray] at (-0.75, 2) (c1) {\bfseries E};
 \node[draw,circle,fill=gray] at (0.75, 2) (c2) {\bfseries F};
 
 \draw[gray,ultra thick] (root) -- (n1);
 \draw[gray,ultra thick] (root) -- (n3);
 \draw[gray,ultra thick] (root) -- (c1);
 \draw[gray,ultra thick] (root) -- (c2);
     
 \end{scope}
 
 \begin{scope}[shift={(12,0)}]
 
 \node[draw,circle,fill=gray] at (0, 4) (root) {\bfseries A};
 \node[draw,circle,fill=cyan] at (-2, 2) (n1) {\bfseries B};
 \node[draw,circle,fill=orange] at (2, 2) (n3) {\bfseries D};
 \node[draw,circle,fill=gray] at (-0.75, 2) (c1) {\bfseries E};
 \node[draw,circle,fill=cyan] at (0.75, 2) (c2) {\bfseries F};
 
 \draw[gray,ultra thick] (root) -- (n1);
 \draw[gray,ultra thick] (root) -- (n3);
 \draw[gray,ultra thick] (root) -- (c1);
 \draw[gray,ultra thick] (root) -- (c2);
     
 \end{scope}
 
 \begin{scope}[shift={(18,0)}]
 
 \node[draw,circle,fill=gray] at (0, 4) (root) {\bfseries A};
 \node[draw,circle,fill=cyan] at (-2, 2) (n1) {\bfseries B};
 \node[draw,circle,fill=orange] at (2, 2) (n3) {\bfseries D};
 \node[draw,circle,fill=gray] at (-0.75, 2) (c1) {\bfseries E};
 \node[draw,circle,fill=cyan] at (0.75, 2) (c2) {\bfseries F};
 \node[draw,circle,fill=gray] at (0.75, 0) (c3) {\bfseries G};
 
 \draw[gray,ultra thick] (root) -- (n1);
 \draw[gray,ultra thick] (root) -- (n3);
 \draw[gray,ultra thick] (root) -- (c1);
 \draw[gray,ultra thick] (root) -- (c2);
 \draw[gray,ultra thick] (c2) -- (c3);
     
 \end{scope}
 
 
 \draw[->,ultra thick] (1.5,3.7) -- (4.5,3.7) node [midway,above] {\bfseries \large Delete C};
 
 \draw[->,ultra thick] (7.5,3.7) -- (10.5,3.7) node [midway,above] {\bfseries \large Relabel F};
 
 \draw[->,ultra thick] (13.5,3.7) -- (16.5,3.7) node [midway,above] {\bfseries \large Insert G};

 \end{tikzpicture}
 }
 \caption{Illustration of classic edit operations. Labels shown as colors.}
 \label{fig:edit_operations_classic}
\end{figure}
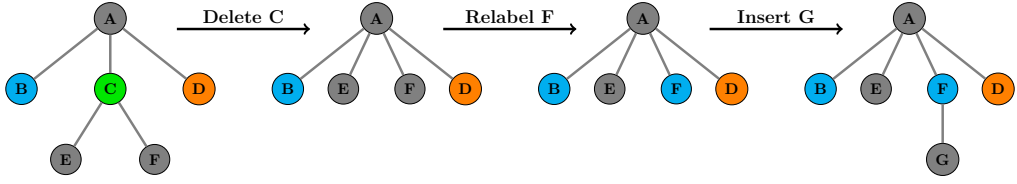

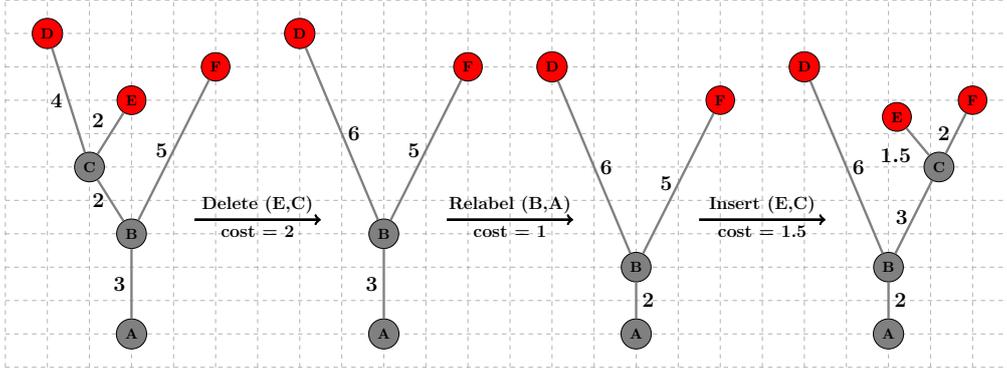
\begin{figure}
 \centering 
 \resizebox{0.9\linewidth}{!}{
 \begin{tikzpicture}[yscale=0.8]
 \definecolor{lightred}{rgb}{1,0.6,0.6}
 \draw[help lines, color=gray!60, dashed] (-3,0) grid (21,11);
 
 
 \node[draw,circle,fill=gray] at (0, 1) (root) {\bfseries A};
 \node[draw,circle,fill=gray] at (0, 4) (s1) {\bfseries B};
 \node[draw,circle,fill=gray] at (-1, 6) (s2) {\bfseries C};
 
 \node[draw,circle,fill=red] at (-2, 10) (m1) {\bfseries D};
 \node[draw,circle,fill=red] at (0, 8) (m2) {\bfseries E};
 \node[draw,circle,fill=red] at (2, 9) (m3) {\bfseries F};
 
 \draw[gray,ultra thick] (root) -- (s1) node [midway,left,black] {\Large$\mathbff{3}$};
 \draw[gray,ultra thick] (s1) -- (s2) node [midway,left,black] {\Large$\mathbff{2}$};
 \draw[gray,ultra thick] (s2) -- (m1) node [midway,left,black] {\Large$\mathbff{4}$};
 \draw[gray,ultra thick] (s2) -- (m2) node [midway,above left=0.5pt,black] {\Large$\mathbff{2}$};
 \draw[gray,ultra thick] (s1) -- (m3) node [midway,left,black] {\Large$\mathbff{5}$};
 
 
 \node[draw,circle,fill=gray] at (6+0, 1) (root') {\bfseries A};
 \node[draw,circle,fill=gray] at (6+0, 4) (s1') {\bfseries B};
 
 \node[draw,circle,fill=red] at (6-2, 10) (m1') {\bfseries D};
 \node[draw,circle,fill=red] at (6+2, 9) (m3') {\bfseries F};
 
 \draw[gray,ultra thick] (root') -- (s1') node [midway,left,black] {\Large$\mathbff{3}$};
 \draw[gray,ultra thick] (s1') -- (m1') node [midway,right,black] {\Large$\mathbff{6}$};
 \draw[gray,ultra thick] (s1') -- (m3') node [midway,left,black] {\Large$\mathbff{5}$};
 
 
 \node[draw,circle,fill=gray] at (12+0, 1) (root'') {\bfseries A};
 \node[draw,circle,fill=gray] at (12+0, 3) (s1'') {\bfseries B};
 
 \node[draw,circle,fill=red] at (12-2, 9) (m1'') {\bfseries D};
 \node[draw,circle,fill=red] at (12+2, 8) (m3'') {\bfseries F};
 
 \draw[gray,ultra thick] (root'') -- (s1'') node [midway,right,black] {\Large$\mathbff{2}$};
 \draw[gray,ultra thick] (s1'') -- (m1'') node [midway,right,black] {\Large$\mathbff{6}$};
 \draw[gray,ultra thick] (s1'') -- (m3'') node [midway,left,black] {\Large$\mathbff{5}$};
 
 
 \node[draw,circle,fill=gray] at (18+0, 1) (root'') {\bfseries A};
 \node[draw,circle,fill=gray] at (18+0, 3) (s1'') {\bfseries B};
 \node[draw,circle,fill=gray] at (18+1.2, 6) (s2'') {\bfseries C};
 
 \node[draw,circle,fill=red] at (18-2, 9) (m1'') {\bfseries D};
 \node[draw,circle,fill=red] at (18+0.2, 7.5) (m2'') {\bfseries E};
 \node[draw,circle,fill=red] at (18+2, 8) (m3'') {\bfseries F};
 
 \draw[gray,ultra thick] (root'') -- (s1'') node [midway,right,black] {\Large$\mathbff{2}$};
 \draw[gray,ultra thick] (s1'') -- (s2'') node [midway,left,black] {\Large$\mathbff{3}$};
 \draw[gray,ultra thick] (s1'') -- (m1'') node [midway,right,black] {\Large$\mathbff{6}$};
 \draw[gray,ultra thick] (s2'') -- (m2'') node [midway,below left=0.5pt,black] {\Large$\mathbff{1.5}$};
 \draw[gray,ultra thick] (s2'') -- (m3'') node [midway,left,black] {\Large$\mathbff{2}$};
 
 
 \draw[->,ultra thick] (1.5,4.42) -- (4.5,4.42) node [midway,above] {\bfseries \large Delete (E,C)};
 \draw[->,ultra thick] (1.5,4.42) -- (4.5,4.42) node [midway,below] {\bfseries \large cost = 2};
 
 \draw[->,ultra thick] (7.5,4.42) -- (10.5,4.42) node [midway,above] {\bfseries \large Relabel (B,A)};
 \draw[->,ultra thick] (7.5,4.42) -- (10.5,4.42) node [midway,below] {\bfseries \large cost = 1};
 
 \draw[->,ultra thick] (13.5,4.42) -- (16.5,4.42) node [midway,above] {\bfseries \large Insert (E,C)};
 \draw[->,ultra thick] (13.5,4.42) -- (16.5,4.42) node [midway,below] {\bfseries \large cost = 1.5};

 \end{tikzpicture}
 }
 \caption{Illustration of deformation-based operations. The coordinate grid indicates edge lengths.}
 \label{fig:edit_operations_deformation}
\end{figure}

\textbf{Deformation-based edit operations} (see Figure~\ref{fig:edit_operations_deformation}) are defined on edge-labeled trees (requiring numerical labels and s specific cost function, usually representing edge length/persistence) are defined as follows.
An edge deletion contracts an edge $(v,p)$, meaning that the edge and $v$ are removed from the tree and each edge $(c,v)$ is replaced by $(c,p)$.
If the remaining node $p$ has degree one afterwards, it is pruned, i.e.\ we replace the edges $(p,p'),(c,p)$ by $(c,p')$.
The new label of $(c,p')$ is the sum of the labels of $(p,p'),(c,p)$.
Edge insertions are the inverse of deletions and relabels again simply change a single label of an edge.

Given a sequence of edit operations $s = s_1 \dots s_k$, we write $T_1 \xrightarrow{\scriptscriptstyle s} T_2$ if $s$ transforms $T_1$ into $T_2$.
The cost of an edit sequence is the sum of all operation costs: $c(s) = \sum_{i=1}^k c(s_i)$.
An edit distance $\delta$ is defined as the cost of an optimal valid edit sequence:
\[ \delta(T_1,T_2) = \min \{ c(s) | T_1 \xrightarrow{\scriptscriptstyle s} T_2, s \text{ is a valid edit sequence} \}. \]
Validity is typically derived from constraints on the edit operations, e.g.\ the one-degree edit distance~\cite{DBLP:journals/ipl/Selkow77,treeEditSurvey} allows insertions/deletions only on leaves, or equivalent edit \emph{mappings}.
We now define the various merge tree edit distances from previous works, organized into categories, in the model described above:

\begin{itemize}
    \item \textbf{BDT-based methods}: three methods based on the usual branch decomposition. The Merge Tree Wasserstein distance~\cite{DBLP:journals/tvcg/PontVDT22}, denoted $\delta_W$, is the classic degree-one edit distance applied to \emph{unordered} BDTs. Nodes in the BDT are labeled by their birth and death values (Figure~\ref{fig:merge_tree_labels}d). Empty branches are labeled with diagonal points and the cost function is the typical Wasserstein metric~\cite{edelsbrunner09}, meaning the Euclidean distance in the birth-death space. Similarly, the extended branch decomposition graph method by Saikia, Seidel and Weinkauf~\cite{DBLP:journals/cgf/SaikiaSW14}, denoted $\delta_X$, uses the same constrained distance, but \emph{ordered} BDTs.
    The Sridharamurthy distance~\cite{DBLP:journals/tvcg/SridharamurthyM20}, denoted $\delta_S$, is the classic constrained edit distance on merge trees directly using branch labels on merge tree nodes (Figure~\ref{fig:merge_tree_labels}e). The latter two will only be discussed in Appendix~\ref{sec:other_dists} and~\ref{sec:perturbation_stability_app} but behave very similar to $\delta_W$.
    \item \textbf{Classic tree edit distances}: two methods adapted from the contour tree alignment distance by Lohfink et al.~\cite{DBLP:journals/cgf/LohfinkWLWG20}. We consider the classic one-degree distance applied to merge trees (for easier comparison with other methods, the original work uses tree alignments), denoted $\delta_L$, where each node is labeled by the scalar distance to its parent (Figure~\ref{fig:merge_tree_labels}c). The root node is labeled with some blank label enforcing the two roots always being mapped. The empty symbol is~$0$ and we use the cost function $c(\ell_1,\ell_2) = |\ell_1 - \ell_2|$. The unconstrained variant, denoted $\delta_G$, uses the classic general tree edit distance on the same trees and the same cost function.
    \item \textbf{Deformation-based edit distances}: The deformation-based edit distance~\cite{taming}, denoted $\delta_E$, is the unconstrained edit distance using the deformation-based edit operations. Edges are labeled by their scalar range (Figure~\ref{fig:merge_tree_labels}b), i.e.\ $\ell((x,y)) = |f(x)-f(y)|$, the empty symbol is~$0$. The cost function is defined as $c(\ell_1,\ell_2) = |\ell_1-\ell_2|$. The path mapping distance\cite{wetzels2022path}, denoted $\delta_P$, is the one-degree variant of $\delta_E$. 
    \item \textbf{The branch mapping distance}~\cite{wetzels2022branch}, denoted $\delta_B$, is a 
    related merge tree-tailored edit distance defined through mappings between branches of two merge trees, see Appendix~\ref{sec:other_dists}. 
\end{itemize}

\begin{figure}[t]
 \centering 
 \resizebox{\linewidth}{!}{
 \begin{tikzpicture}[yscale=0.8]
 \definecolor{lightred}{rgb}{1,0.6,0.6}
 
 \begin{scope}[shift={(0,0)}]
 \draw[help lines, color=gray!70, dashed] (-3,-1) grid (3,11);
 
 \node[draw,circle,fill=gray,label={[right=6pt]\large$\mathbff{0}$}] at (0, 0) (root) {\bfseries A};
 \node[draw,circle,fill=gray,label={[above]\large$\mathbff{2}$}] at (0, 2) (s1) {\bfseries B};
 \node[draw,circle,fill=gray,label=left:{\large$\mathbff{4}$}] at (-1, 4) (s2) {\bfseries C};
 \node[draw,circle,fill=gray,label=right:{\large$\mathbff{6}$}] at (0.5, 6) (s3) {\bfseries D};
 
 \node[draw,circle,fill=red,label={[right=6pt]\large$\mathbff{10}$}] at (-2.5, 10) (m1) {\bfseries E};
 \node[draw,circle,fill=red,label={[left=6pt]\large$\mathbff{5}$}] at (2.5, 5) (m2) {\bfseries F};
 \node[draw,circle,fill=red,label={[above]\large$\mathbff{8}$}] at (-1, 8) (m3) {\bfseries G};
 \node[draw,circle,fill=red,label={[above]\large$\mathbff{8.5}$}] at (0.5, 8.5) (m4) {\bfseries H};
 \node[draw,circle,fill=red,label={[above]\large$\mathbff{9}$}] at (2, 9) (m5) {\bfseries I};
 
 \draw[gray,ultra thick] (root) -- (s1);
 \draw[gray,ultra thick] (s1) -- (s2);
 \draw[gray,ultra thick] (s2) -- (m1);
 \draw[gray,ultra thick] (s2) -- (s3);
 \draw[gray,ultra thick] (s3) -- (m3);
 \draw[gray,ultra thick] (s3) -- (m4);
 \draw[gray,ultra thick] (s3) -- (m5);
 \draw[gray,ultra thick] (s1) -- (m2);
     
 \end{scope}
 
 \begin{scope}[shift={(8,0)}]
 \draw[help lines, color=gray!70, dashed] (-3,-1) grid (3,11);
 
 \node[draw,circle,fill=gray] at (0, 0) (root) {\bfseries A};
 \node[draw,circle,fill=gray] at (0, 2) (s1) {\bfseries B};
 \node[draw,circle,fill=gray] at (-1, 4) (s2) {\bfseries C};
 \node[draw,circle,fill=gray] at (0.5, 6) (s3) {\bfseries D};
 
 \node[draw,circle,fill=red] at (-2.5, 10) (m1) {\bfseries E};
 \node[draw,circle,fill=red] at (2.5, 5) (m2) {\bfseries F};
 \node[draw,circle,fill=red] at (-1, 8) (m3) {\bfseries G};
 \node[draw,circle,fill=red] at (0.5, 8.5) (m4) {\bfseries H};
 \node[draw,circle,fill=red] at (2, 9) (m5) {\bfseries I};
 
 \draw[gray,ultra thick] (root) -- (s1) node [midway,right,black] {\large$\mathbff{2}$};
 \draw[gray,ultra thick] (s1) -- (s2) node [midway,above,sloped,black] {\large$\mathbff{2}$};
 \draw[gray,ultra thick] (s2) -- (m1) node [midway,left,black] {\large$\mathbff{6}$};
 \draw[gray,ultra thick] (s2) -- (s3) node [midway,below,sloped,black] {\large$\mathbff{2}$};
 \draw[gray,ultra thick] (s3) -- (m3) node [midway,below,sloped,black] {\large$\mathbff{2}$};
 \draw[gray,ultra thick] (s3) -- (m4) node [midway,above,sloped,black] {\large$\mathbff{2.5}$};
 \draw[gray,ultra thick] (s3) -- (m5) node [midway,below,sloped,black] {\large$\mathbff{3}$};
 \draw[gray,ultra thick] (s1) -- (m2) node [midway,above,sloped,black] {\large$\mathbff{2}$};
     
 \end{scope}
 
 \begin{scope}[shift={(16,0)}]
 \draw[help lines, color=gray!70, dashed] (-3,-1) grid (3,11);
 
 \node[draw,circle,fill=gray,label={[right=6pt]\large$\mathbff{\infty}$}] at (0, 0) (root) {\bfseries A};
 \node[draw,circle,fill=gray,label={[above]\large$\mathbff{2}$}] at (0, 2) (s1) {\bfseries B};
 \node[draw,circle,fill=gray,label=left:{\large$\mathbff{2}$}] at (-1, 4) (s2) {\bfseries C};
 \node[draw,circle,fill=gray,label=right:{\large$\mathbff{2}$}] at (0.5, 6) (s3) {\bfseries D};
 
 \node[draw,circle,fill=red,label={[right=6pt]\large$\mathbff{6}$}] at (-2.5, 10) (m1) {\bfseries E};
 \node[draw,circle,fill=red,label={[left=6pt]\large$\mathbff{3}$}] at (2.5, 5) (m2) {\bfseries F};
 \node[draw,circle,fill=red,label={[above]\large$\mathbff{2}$}] at (-1, 8) (m3) {\bfseries G};
 \node[draw,circle,fill=red,label={[above]\large$\mathbff{2.5}$}] at (0.5, 8.5) (m4) {\bfseries H};
 \node[draw,circle,fill=red,label={[above]\large$\mathbff{3}$}] at (2, 9) (m5) {\bfseries I};
 
 \draw[gray,ultra thick] (root) -- (s1);
 \draw[gray,ultra thick] (s1) -- (s2);
 \draw[gray,ultra thick] (s2) -- (m1);
 \draw[gray,ultra thick] (s2) -- (s3);
 \draw[gray,ultra thick] (s3) -- (m3);
 \draw[gray,ultra thick] (s3) -- (m4);
 \draw[gray,ultra thick] (s3) -- (m5);
 \draw[gray,ultra thick] (s1) -- (m2);
     
 \end{scope}
 
 \begin{scope}[shift={(3,-14+3)}]
 
 \node[draw,very thick,circle,label=above:{\large$\mathbff{(10,0)}$}] at (0, 4) (m1) {\bfseries E-A};
 \node[draw,very thick,circle,label=below:{\large$\mathbff{(5,2)}$}] at (-2, 2) (m2) {\bfseries F-B};
 \node[draw,very thick,circle,label=below:{\large$\mathbff{(8,6)}$}] at (1, 0) (m3) {\bfseries G-D};
 \node[draw,very thick,circle,label=below:{\large$\mathbff{(8.5,6)}$}] at (3, 0) (m4) {\bfseries H-D};
 \node[draw,very thick,circle,label=above:{\large$\mathbff{(9,4)}$}] at (2, 2) (m5) {\bfseries I-C};
 
 \draw[ultra thick] (m1) -- (m2);
 \draw[ultra thick] (m1) -- (m5);
 \draw[ultra thick] (m5) -- (m3);
 \draw[ultra thick] (m5) -- (m4);
     
 \end{scope}
 
 \begin{scope}[shift={(11,-14)}]
 \draw[help lines, color=gray!70, dashed] (-3,-1) grid (3,11);
 
 \node[draw,circle,fill=gray,label={[right=6pt]\large$\mathbff{(10,0)}$}] at (0, 0) (root) {\bfseries A};
 \node[draw,circle,fill=gray,label=right:{\large$\mathbff{(5,2)}$}] at (0, 2) (s1) {\bfseries B};
 \node[draw,circle,fill=gray,label=left:{\large$\mathbff{(9,4)}$}] at (-1, 4) (s2) {\bfseries C};
 \node[draw,circle,fill=gray,label=left:{\large$\mathbff{(8.5,6)}$}] at (0.5, 6) (s3) {\bfseries D};
 
 \node[draw,circle,fill=red,label={[right=6pt]\large$\mathbff{(10,0)}$}] at (-2.5, 10) (m1) {\bfseries E};
 \node[draw,circle,fill=red,label=above:{\large$\mathbff{(5,2)}$}] at (2.5, 5) (m2) {\bfseries F};
 \node[draw,circle,fill=red,label={[above]\large$\mathbff{(8,6)}$}] at (-1, 8) (m3) {\bfseries G};
 \node[draw,circle,fill=red,label={[above]\large$\mathbff{(8.5,6)}$}] at (0.5, 8.5) (m4) {\bfseries H};
 \node[draw,circle,fill=red,label={[above]\large$\mathbff{(9,4)}$}] at (2, 9) (m5) {\bfseries I};
 
 \draw[gray,ultra thick] (root) -- (s1);
 \draw[gray,ultra thick] (s1) -- (s2);
 \draw[gray,ultra thick] (s2) -- (m1);
 \draw[gray,ultra thick] (s2) -- (s3);
 \draw[gray,ultra thick] (s3) -- (m3);
 \draw[gray,ultra thick] (s3) -- (m4);
 \draw[gray,ultra thick] (s3) -- (m5);
 \draw[gray,ultra thick] (s1) -- (m2);
     
 \end{scope}

 \node[] at (0,-1.5) (la) {\Large\bfseries (a) Original Merge Tree};
 \node[] at (8,-1.5) (lb) {\Large\bfseries (b) Labeled Tree for $\mathbff{\delta_E}$\&$\mathbff{\delta_P}$};
 \node[] at (16,-1.5) (lc) {\Large\bfseries (c) Labeled Tree for $\mathbff{\delta_G}$\&$\mathbff{\delta_L}$};
 \node[] at (3,-14-1.5) (ld) {\Large\bfseries (d) Labeled BDT for $\mathbff{\delta_W}$\&$\mathbff{\delta_X}$};
 \node[] at (11,-14-1.5) (le) {\Large\bfseries (e) Labeled Tree for $\mathbff{\delta_S}$};

 \end{tikzpicture}
 }
 \caption{Examples for label functions used by different distances.}
 \label{fig:merge_tree_labels}
\end{figure}
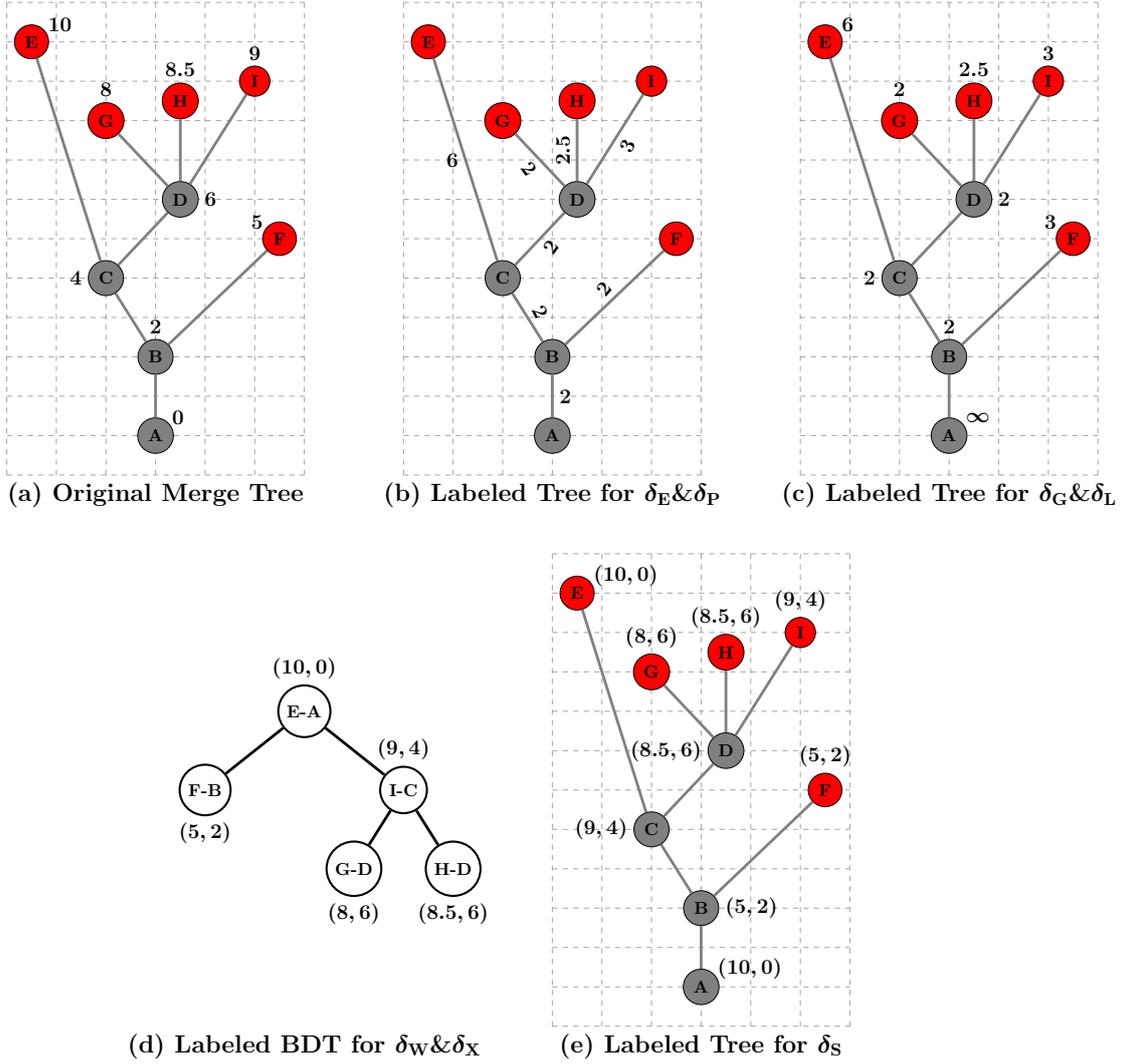

\section{Minimal Vertex Perturbations}
\label{sec:transformations}

Consider a piecewise linear scalar field $\mathbb{X},f$.
Let $\mathbb{X},f'$ be a perturbation of $f$.
Let $F,F' : V(\mathbb{X}) \rightarrow \{1,\dots,n\}$ be the corresponding vertex orderings by ascending scalar value.
We say that the two fields differ by a \emph{minimal vertex perturbation}, if exactly one vertex is perturbed in a minimal way: it swaps in the ordering with at most one other vertex.

\begin{definition}
    Two scalar fields $X,f$ and $X,f'$ differ by a minimal vertex perturbation if
    \begin{itemize}
        \item there is exactly one vertex $v$ such that $f(v) \neq f'(v)$ and $f(u) = f'(u)$ for all $u \neq v$,
        \item we either have $F = F'$ or there is exactly one vertex $u$ with $u \neq v$ such that $|F(v) - F(u)| = 1$, $F(v) = F'(u)$, and $F(u) = F'(v)$, and $F(w) = F'(w)$ for all $w$ with $w \neq u$ and $w \neq v$.
    \end{itemize}
\end{definition}

\noindent
To study the effects of such minimal vertex perturbation on the merge tree $T_f$ and stability of edit distances, we proceed as follows.
First, we study in which ways the \emph{augmented} merge tree $\mathcal{T}_f$ can be affected.
Then, we consider all possible induced changes on the \emph{unaugmented} merge tree $T_f$ and its branch decomposition $\obdt(T_f)$, which we categorize into one of four types.
Our case distinction on the augmented merge tree is reminiscent of the one presented in the works of Oesterling et al.~\cite{oesterling_time_varying_mergetrees}.
For simplicity reasons, we only consider the case where the scalar value of a vertex is increased.
The case where a scalar value is decreased is completely symmetrical.
In Appendix~\ref{sec:sequence_stability}, we then consider the more general case of arbitrary perturbations.
Before considering specific effects of a perturbation, we begin with introducing the four types of merge tree changes.

\medskip
\noindent
\textbf{Classification of merge tree changes.}
Two kinds of changes can happen between the two unaugmented merge trees $T_f$ and $T_{f'}$.
First, it is possible that the size or the connectivity differs, meaning $E(T_f) \neq E(T_{f'})$.
Second, it is possible that the corresponding branch decomposition changes (together or without structural changes), meaning $\obdt(T_f) \neq \obdt(T_{f'})$.
Both types of changes can appear in an easy fashion or a more difficult one.
By this, we mean whether the two trees are in a subgraph relation or not.
If $T_f \subseteq T_{f'}$ or $\obdt(T_f) \subseteq \obdt(T_{f'})$, then we can derive one from the other by simply appending a leaf.
Such changes are usually easy to handle by edit distances, even constrained ones.
Thus, our classification is based on exactly this case distinction.

The first case we consider is the one where no significant changes happen for both structures.
As described above, by this we mean $T_f \subseteq T_{f'}$ or $T_{f'} \subseteq T_f$, and $\obdt(T_f) \subseteq \obdt(T_{f'})$ or $\obdt(T_{f'}) \subseteq \obdt(T_f)$.
We say that $T_f$ and $T_{f'}$ differ by simple changes.

The first type of actual change we consider is a \emph{vertical branch swap}, the name being derived from the well-established term \emph{vertical instability}.
Intuitively, a vertical swap happens when the main branch of the merge tree switches from one maximum to another, but nothing else.
Of course, we also want to capture the case where this behavior happens for a specific subtree: any changes in the priority of branches should be considered.
In such cases, only the branch decomposition changes, but the tree structure remains unchanged.
Formally, we describe this behavior as follows:
two merge trees $T_f$ and $T_{f'}$ differ by a vertical branch swap, if $\obdt(T_f) \not \subseteq \obdt(T_{f'})$ and $\obdt(T_{f'}) \not \subseteq \obdt(T_f)$, but $T_f \subseteq T_{f'}$ or $T_{f'} \subseteq T_f$.

The next type we call \emph{edge splits}.
Edge splits describe a behavior that is captured well by distances based on branch decomposition, but not by classic tree edit distances.
As the name suggests, they intuitively describe changes where an edge in a merge tree is split by a new saddle vertex and a new subtree is added branching from the new saddle.
Formally, two merge trees $T_f$ and $T_{f'}$ differ by an edge split if $T_f \not \subseteq T_{f'}$ and $T_{f'} \not \subseteq T_f$, but $\obdt(T_f) \subseteq \obdt(T_{f'})$ or $\obdt(T_{f'}) \subseteq \obdt(T_f)$.

The last type is a \emph{horizontal branch swap}, the name again being derived from the well-established term \emph{horizontal instability}.
Intuitively, a horizontal swap happens, when the priorities of branches stay equal, but their nesting changes.
They are also often called saddle swaps, as this happens if saddles pass each other in the tree structure.
Usually, two types are considered: either only the \emph{ordered} branch decomposition tree changes or the \emph{unordered} branch decomposition tree does so too.
Only in the latter case the nesting of branches truly changes.
Formally, two merge trees $T_f$ and $T_{f'}$ differ by an ordered horizontal swap if  $T_f \not \subseteq T_{f'}$ and $T_{f'} \not \subseteq T_f$, as well as $\obdt(T_f) \not \subseteq \obdt(T_{f'})$ and $\obdt(T_{f'}) \not \subseteq \obdt(T_f)$.
They differ by an unordered horizontal swap, if $T_f \not \subseteq T_{f'}$ and $T_{f'} \not \subseteq T_f$, as well as $\bdt(T_f) \not \subseteq \bdt(T_{f'})$ and $\bdt(T_{f'}) \not \subseteq \bdt(T_f)$.
For our classification, we focus on ordered horizontal swaps.
When considering unordered horizontal swaps, the situation gets a bit more complicated, which we discuss in Section~\ref{sec:conclusion}.

We usually consider the classification above \emph{up to isomorphism}: the subgraph relations only need to hold under renaming of the vertices.
Note that the classification is indeed complete: any minimal vertex perturbation is of one of the four types.

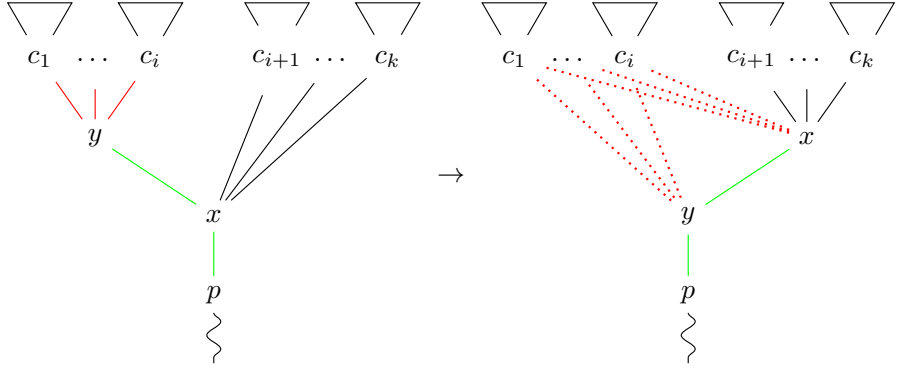
\begin{figure}
 \centering 
 \resizebox{0.8\linewidth}{!}{
 \begin{tikzpicture}
 \tikzset{decoration={snake}}

 \begin{scope}[shift={(-3,0)}]
 
 \node[] at (0, 0) (p) {$p$};
 \node[] at (0, -1) (p_) {};
 \node[] at (0, 1) (x) {$x$};
 \node[] at (-1.5, 2) (y) {$y$};
 
 \node[draw,isosceles triangle,isosceles triangle apex angle=60,rotate=270,anchor=apex,minimum size=20pt]
         at (-2.2, 3) (c1t) {};
 \node[circle, fill=white]
         at (-2.2, 3) (c1) {$c_1$};
 \node[draw,isosceles triangle,isosceles triangle apex angle=60,rotate=270,anchor=apex,minimum size=20pt]
         at (-0.8, 3) (cit) {};
 \node[circle, fill=white]
         at (-0.8, 3) (ci) {$c_{i}$};
 \node[circle]
         at (-1.5, 3) (cd) {$\dots$};
         
 \node[draw,isosceles triangle,isosceles triangle apex angle=60,rotate=270,anchor=apex,minimum size=20pt]
         at (0.8, 3) (ciit) {};
 \node[circle, fill=white]
         at (0.8, 3) (cii) {$c_{i+1}$};
 \node[draw,isosceles triangle,isosceles triangle apex angle=60,rotate=270,anchor=apex,minimum size=20pt]
         at (2.2, 3) (c2t) {};
 \node[circle, fill=white]
         at (2.2, 3) (c2) {$c_{k}$};
 \node[circle]
         at (1.5, 3) (cdd) {$\dots$};
 
 \draw[decorate] (p) -- (p_);
 \draw[green] (p) -- (x);
 \draw[green] (x) -- (y);
 \draw[red] (y) -- (c1);
 \draw[red] (y) -- (ci);
 \draw[red] (y) -- (cd);
 \draw[black] (x) -- (cii);
 \draw[black] (x) -- (c2);
 \draw[black] (x) -- (cdd);

 \end{scope}

 \node[] at (0,1.5) (arr) {$\rightarrow$};

 \begin{scope}[shift={(3,0)}]
 
 \node[] at (0, 0) (p) {$p$};
 \node[] at (0, -1) (p_) {};
 \node[] at (1.5, 2) (x) {$x$};
 \node[] at (0, 1) (y) {$y$};
 
 \node[draw,isosceles triangle,isosceles triangle apex angle=60,rotate=270,anchor=apex,minimum size=20pt]
         at (-2.2, 3) (c1t) {};
 \node[circle, fill=white]
         at (-2.2, 3) (c1) {$c_1$};
 \node[draw,isosceles triangle,isosceles triangle apex angle=60,rotate=270,anchor=apex,minimum size=20pt]
         at (-0.8, 3) (cit) {};
 \node[circle, fill=white]
         at (-0.8, 3) (ci) {$c_{i}$};
 \node[circle]
         at (-1.5, 3) (cd) {$\dots$};
         
 \node[draw,isosceles triangle,isosceles triangle apex angle=60,rotate=270,anchor=apex,minimum size=20pt]
         at (0.8, 3) (ciit) {};
 \node[circle, fill=white]
         at (0.8, 3) (cii) {$c_{i+1}$};
 \node[draw,isosceles triangle,isosceles triangle apex angle=60,rotate=270,anchor=apex,minimum size=20pt]
         at (2.2, 3) (c2t) {};
 \node[circle, fill=white]
         at (2.2, 3) (c2) {$c_{k}$};
 \node[circle]
         at (1.5, 3) (cdd) {$\dots$};
 
 \draw[decorate] (p) -- (p_);
 \draw[green] (p) -- (y);
 \draw[green] (x) -- (y);
 \draw[red,dotted,thick] (x) -- (c1);
 \draw[red,dotted,thick] (x) -- (ci);
 \draw[red,dotted,thick] (x) -- (cd);
 \draw[red,dotted,thick] (y) -- (c1);
 \draw[red,dotted,thick] (y) -- (ci);
 \draw[red,dotted,thick] (y) -- (cd);
 \draw[black] (x) -- (cii);
 \draw[black] (x) -- (c2);
 \draw[black] (x) -- (cdd);

 \end{scope}

 \end{tikzpicture}
 }
 \caption{Illustration of a minimal vertex perturbation in an augmented merge tree according to Observation~\ref{obs:no_edge}-\ref{obs:children_y}: all green edges change in a fixed manner, all black edges and subtrees remain unchanged, and all red edges need to be considered in each case specifically.}
 \label{fig:perturbation_edge_changes}
\end{figure}

\subsection*{Structural Changes in the Augmented Merge Tree}

We now study which changes in augmented merge trees $\mathcal{T}_f,\mathcal{T}_{f'}$ can be induced by minimal vertex perturbation between $f$ and $f'$ where $x$ passes $y$.
In~\cite{oesterling_time_varying_mergetrees} it was observed that if $x$ and $y$ are not neighbors in $\mathcal{T}_f$, then $\mathcal{T}_f$ and $\mathcal{T}_{f'}$ are structurally equal, and that the transposition only affects edges incident to $x$ or $y$: for all $v,u$ with $x \neq v \neq y$, $x \neq u \neq y$, and $u$ lowest in $f^{-1}((f(v),\infty)) = f'^{-1}((f'(v),\infty))$, $u$ stays the lowest vertex.
\begin{observation}
\label{obs:no_edge}
    If $(y,x) \notin E(\mathcal{T}_f)$ then $E(\mathcal{T}_f) = E(\mathcal{T}_{f'})$.
\end{observation}
\begin{observation}
\label{obs:remaining_edges}
    For any $(u,v) \in E(\mathcal{T}_f)$ with $\{u,v\} \cap \{x,y\} = \emptyset$, we have $(u,v) \in E(\mathcal{T}_{f'})$.
\end{observation}
Next, we consider the edges immediately involved in the transposition in the case where $x$ and $y$ are neighbors in $\mathcal{T}_f$.
First, we note that the edge between $x$ and $y$ is preserved (but flipped) and that $y$ inherits the parent $p$ of $x$~\cite{oesterling_time_varying_mergetrees}: we have $x$ lowest (w.r.t.\ $f$) in its component $C$ of $f^{-1}((f(p),\infty))$ and thus $y$ lowest (w.r.t.\ $f'$) in $C$, also a component of $f'^{-1}((f'(p),\infty))$.
Second, we consider the other children of $x$ in $\mathcal{T}_f$, meaning the vertices $z \neq y$ with $(z,x) \in E(\mathcal{T}_f)$.
Since $x$ only passes $y$, the connectivity of its superlevel set components not containing $y$ stays untouched.
Therefore, all distinct subtrees of $x$ (not containing~$y$) stay distinct subtrees and their roots stay children of $x$~\cite{oesterling_time_varying_mergetrees}.
\begin{observation}
\label{obs:flipped_edge}
    We have $(x,y) \in E(\mathcal{T}_{f'})$, and for $(x,p) \in E(\mathcal{T}_{f})$: $(y,p) \in E(\mathcal{T}_{f'})$.
\end{observation}
\begin{observation}
\label{obs:children_x}
    For each $z \neq y$ with $(z,x) \in E(\mathcal{T}_{f})$, we have $(z,x) \in E(\mathcal{T}_{f'})$.
\end{observation}
The only remaining edges to consider are the children of $y$ in $\mathcal{T}_f$.
They can stay a child of $y$ or become a child of $x$, depending on the connectivity of $x$ in $\mathbb{X}$.
Let $z$ be a child of $y$, i.e.\ $(z,y) \in E(\mathcal{T}_{f})$ with $C$ being its component 
in the superlevel set $f^{-1}((f(y),\infty))$ (and also a distinct subtree of $y$ in $\mathcal{T}_{f}$).
Either there is an edge directly connecting $x$ to $C$ in $\mathbb{X}$ or not.
If there is one, $C$ becomes a distinct subtree of $x$, otherwise it stays a distinct subtree of $y$.
\begin{observation}
\label{obs:children_y}
    We either have $(z,x) \in E(\mathcal{T}_{f'})$ or $(z,y) \in E(\mathcal{T}_{f'})$ $\forall z$ with $(z,y) \in E(\mathcal{T}_{f})$.
\end{observation}
Figure~\ref{fig:perturbation_edge_changes} summarizes the conclusions above: we have to consider which children of $y$ become children of $x$ to derive all possible structural changes between $\mathcal{T}_{f}$ and $\mathcal{T}_{f'}$.

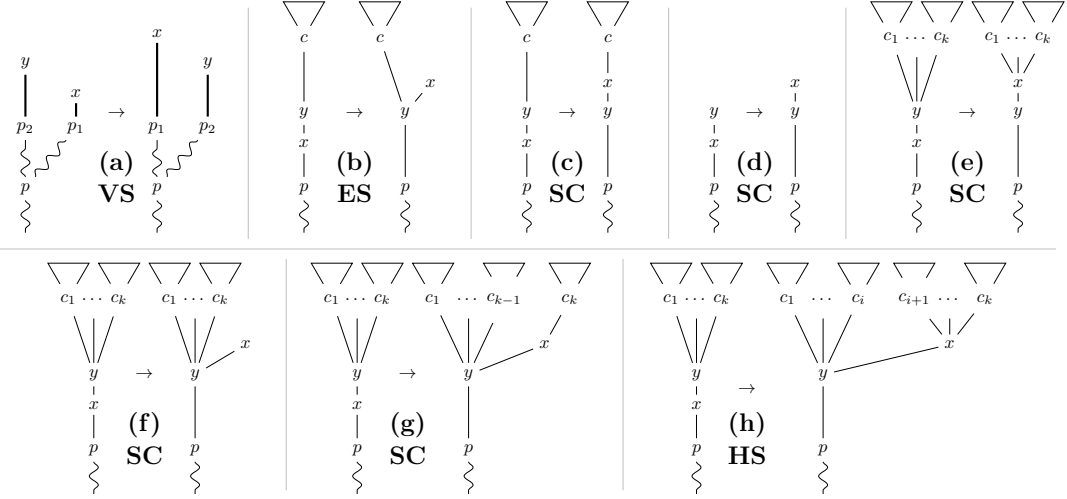
\begin{figure}
 \centering 
 \resizebox{\linewidth}{!}{
 \begin{tikzpicture}
 \tikzset{decoration={snake}}

 \begin{scope}[shift={(0,5.2)}]
 \node[draw=none,fill=none,circle] at (-2.5, -1) (dummy1) {};
 \node[draw=none,fill=none,circle] at (2.5, -1) (dummy2) {};
 \node[draw=none,fill=none,circle] at (-2.5, 4) (dummy3) {};
 \node[draw=none,fill=none,circle] at (2.5, 4) (dummy4) {};

 \begin{scope}[shift={(-0.8,0)}]
 
 \node[] at (0, 1.2) (p1) {$p_1$};
 \node[] at (-1, 1.2) (p2) {$p_2$};
 \node[] at (-1, 0) (p) {$p$};
 \node[] at (-1, -1) (p_) {};
 \node[] at (0, 1.9) (x) {$x$};
 \node[] at (-1, 2.5) (y) {$y$};
 
 \draw[very thick] (p2) -- (y);
 \draw[very thick] (p1) -- (x);
 \draw[decorate] (p) -- (p1);
 \draw[decorate] (p) -- (p2);
 \draw[decorate] (p) -- (p_);

 \end{scope}

 \node[] at (0,1.5) (arr) {$\rightarrow$};
 \node[] at (0,0.5) (l) {\Large\textbf{(a)}};
 \node[] at (0,-0.1) (l) {\Large\textbf{VS}};

 \begin{scope}[shift={(0.8,0)}]
 
 \node[] at (0, 1.2) (p1) {$p_1$};
 \node[] at (1, 1.2) (p2) {$p_2$};
 \node[] at (0, 0) (p) {$p$};
 \node[] at (0, -1) (p_) {};
 \node[] at (0, 3.1) (x) {$x$};
 \node[] at (1, 2.5) (y) {$y$};
 
 \draw[very thick] (p2) -- (y);
 \draw[very thick] (p1) -- (x);
 \draw[decorate] (p) -- (p1);
 \draw[decorate] (p) -- (p2);
 \draw[decorate] (p) -- (p_);

 \end{scope}
 \end{scope}

 \begin{scope}[shift={(4.7,5.2)}]
 \node[draw=none,fill=none,circle] at (-2, -1) (dummy1) {};
 \node[draw=none,fill=none,circle] at (2, -1) (dummy2) {};
 \node[draw=none,fill=none,circle] at (-2, 4) (dummy3) {};
 \node[draw=none,fill=none,circle] at (2, 4) (dummy4) {};

 \begin{scope}[shift={(-1,0)}]
 
 \node[] at (0, 0) (p) {$p$};
 \node[] at (0, -1) (p_) {};
 \node[] at (0, 0.9) (x) {$x$};
 \node[] at (0, 1.5) (y) {$y$};
 \node[draw,isosceles triangle,isosceles triangle apex angle=60,rotate=270,anchor=apex,minimum size=20pt]
         at (0, 3) (ct) {};
 \node[circle, fill=white]
         at (0, 3) (c) {$c$};
 
 \draw[decorate] (p) -- (p_);
 \draw[] (p) -- (x);
 \draw[] (x) -- (y);
 \draw[] (y) -- (c);

 \end{scope}

 \node[] at (0,1.5) (arr) {$\rightarrow$};
 \node[] at (0,0.5) (l) {\Large\textbf{(b)}};
 \node[] at (0,-0.1) (l) {\Large\textbf{ES}};

 \begin{scope}[shift={(1,0)}]
 
 \node[] at (0, 0) (p) {$p$};
 \node[] at (0, -1) (p_) {};
 \node[] at (0.5, 2.1) (x) {$x$};
 \node[] at (0, 1.5) (y) {$y$};
 \node[draw,isosceles triangle,isosceles triangle apex angle=60,rotate=270,anchor=apex,minimum size=20pt]
         at (-0.5, 3) (ct) {};
 \node[circle, fill=white]
         at (-0.5, 3) (c) {$c$};
 
 \draw[decorate] (p) -- (p_);
 \draw[] (p) -- (y);
 \draw[] (x) -- (y);
 \draw[] (y) -- (c);

 \end{scope}
 \end{scope}
 \label{fig:reg_reg_with_change}

 \begin{scope}[shift={(8.9,5.2)}]
 \node[draw=none,fill=none,circle] at (-2, -1) (dummy1) {};
 \node[draw=none,fill=none,circle] at (2, -1) (dummy2) {};
 \node[draw=none,fill=none,circle] at (-2, 4) (dummy3) {};
 \node[draw=none,fill=none,circle] at (2, 4) (dummy4) {};

 \begin{scope}[shift={(-0.8,0)}]
 
 \node[] at (0, 0) (p) {$p$};
 \node[] at (0, -1) (p_) {};
 \node[] at (0, 0.9) (x) {$x$};
 \node[] at (0, 1.5) (y) {$y$};
 \node[draw,isosceles triangle,isosceles triangle apex angle=60,rotate=270,anchor=apex,minimum size=20pt]
         at (0, 3) (ct) {};
 \node[circle, fill=white]
         at (0, 3) (c) {$c$};
 
 \draw[decorate] (p) -- (p_);
 \draw[] (p) -- (x);
 \draw[] (x) -- (y);
 \draw[] (y) -- (c);

 \end{scope}

 \node[] at (0,1.5) (arr) {$\rightarrow$};
 \node[] at (0,0.5) (l) {\Large\textbf{(c)}};
 \node[] at (0,-0.1) (l) {\Large\textbf{SC}};

 \begin{scope}[shift={(0.8,0)}]
 
 \node[] at (0, 0) (p) {$p$};
 \node[] at (0, -1) (p_) {};
 \node[] at (0, 2.1) (x) {$x$};
 \node[] at (0, 1.5) (y) {$y$};
 \node[draw,isosceles triangle,isosceles triangle apex angle=60,rotate=270,anchor=apex,minimum size=20pt]
         at (0, 3) (ct) {};
 \node[circle, fill=white]
         at (0, 3) (c) {$c$};
 
 \draw[decorate] (p) -- (p_);
 \draw[] (p) -- (y);
 \draw[] (x) -- (y);
 \draw[] (x) -- (c);

 \end{scope}
 \end{scope}
 \label{fig:reg_reg_without_change}

 \begin{scope}[shift={(12.6,5.2)}]
 \node[draw=none,fill=none,circle] at (-1.5, -1) (dummy1) {};
 \node[draw=none,fill=none,circle] at (1.5, -1) (dummy2) {};
 \node[draw=none,fill=none,circle] at (-1.5, 4) (dummy3) {};
 \node[draw=none,fill=none,circle] at (1.5, 4) (dummy4) {};

 \begin{scope}[shift={(-0.8,0)}]
 
 \node[] at (0, 0) (p) {$p$};
 \node[] at (0, -1) (p_) {};
 \node[] at (0, 0.9) (x) {$x$};
 \node[] at (0, 1.5) (y) {$y$};
 
 \draw[decorate] (p) -- (p_);
 \draw[] (p) -- (x);
 \draw[] (x) -- (y);

 \end{scope}

 \node[] at (0,1.5) (arr) {$\rightarrow$};
 \node[] at (0,0.5) (l) {\Large\textbf{(d)}};
 \node[] at (0,-0.1) (l) {\Large\textbf{SC}};

 \begin{scope}[shift={(0.8,0)}]
 
 \node[] at (0, 0) (p) {$p$};
 \node[] at (0, -1) (p_) {};
 \node[] at (0, 2.1) (x) {$x$};
 \node[] at (0, 1.5) (y) {$y$};
 
 \draw[decorate] (p) -- (p_);
 \draw[] (p) -- (y);
 \draw[] (x) -- (y);

 \end{scope}
 \end{scope}
 \label{fig:reg_max}

 \begin{scope}[shift={(16.8,5.2)}]
 \node[draw=none,fill=none,circle] at (-2.5, -1) (dummy1) {};
 \node[draw=none,fill=none,circle] at (2.5, -1) (dummy2) {};
 \node[draw=none,fill=none,circle] at (-2.5, 4) (dummy3) {};
 \node[draw=none,fill=none,circle] at (2.5, 4) (dummy4) {};

 \begin{scope}[shift={(-1,0)}]
 
 \node[] at (0, 0) (p) {$p$};
 \node[] at (0, -1) (p_) {};
 \node[] at (0, 0.9) (x) {$x$};
 \node[] at (0, 1.5) (y) {$y$};
 \node[draw,isosceles triangle,isosceles triangle apex angle=60,rotate=270,anchor=apex,minimum size=20pt]
         at (-0.5, 3) (c1t) {};
 \node[circle, fill=white]
         at (-0.5, 3) (c1) {$c_1$};
 \node[draw,isosceles triangle,isosceles triangle apex angle=60,rotate=270,anchor=apex,minimum size=20pt]
         at (0.5, 3) (c2t) {};
 \node[circle, fill=white]
         at (0.5, 3) (c2) {$c_k$};
 \node[circle]
         at (0, 3) (cd) {$\dots$};
 
 \draw[decorate] (p) -- (p_);
 \draw[] (p) -- (x);
 \draw[] (x) -- (y);
 \draw[] (y) -- (c1);
 \draw[] (y) -- (c2);
 \draw[] (y) -- (cd);

 \end{scope}

 \node[] at (0,1.5) (arr) {$\rightarrow$};
 \node[] at (0,0.5) (l) {\Large\textbf{(e)}};
 \node[] at (0,-0.1) (l) {\Large\textbf{SC}};

 \begin{scope}[shift={(1,0)}]
 
 \node[] at (0, 0) (p) {$p$};
 \node[] at (0, -1) (p_) {};
 \node[] at (0, 2.1) (x) {$x$};
 \node[] at (0, 1.5) (y) {$y$};
 \node[draw,isosceles triangle,isosceles triangle apex angle=60,rotate=270,anchor=apex,minimum size=20pt]
         at (-0.5, 3) (c1t) {};
 \node[circle, fill=white]
         at (-0.5, 3) (c1) {$c_1$};
 \node[draw,isosceles triangle,isosceles triangle apex angle=60,rotate=270,anchor=apex,minimum size=20pt]
         at (0.5, 3) (c2t) {};
 \node[circle, fill=white]
         at (0.5, 3) (c2) {$c_k$};
 \node[circle]
         at (0, 3) (cd) {$\dots$};
 
 \draw[decorate] (p) -- (p_);
 \draw[] (p) -- (y);
 \draw[] (x) -- (y);
 \draw[] (x) -- (c1);
 \draw[] (x) -- (c2);
 \draw[] (x) -- (cd);

 \end{scope}
 \end{scope}
 \label{fig:reg_sad_all}

 \begin{scope}[shift={(0.55,0)}]
 \node[draw=none,fill=none,circle] at (-2.5, -1) (dummy1) {};
 \node[draw=none,fill=none,circle] at (2.5, -1) (dummy2) {};
 \node[draw=none,fill=none,circle] at (-2.5, 4) (dummy3) {};
 \node[draw=none,fill=none,circle] at (2.5, 4) (dummy4) {};

 \begin{scope}[shift={(-1,0)}]
 
 \node[] at (0, 0) (p) {$p$};
 \node[] at (0, -1) (p_) {};
 \node[] at (0, 0.9) (x) {$x$};
 \node[] at (0, 1.5) (y) {$y$};
 \node[draw,isosceles triangle,isosceles triangle apex angle=60,rotate=270,anchor=apex,minimum size=20pt]
         at (-0.5, 3) (c1t) {};
 \node[circle, fill=white]
         at (-0.5, 3) (c1) {$c_1$};
 \node[draw,isosceles triangle,isosceles triangle apex angle=60,rotate=270,anchor=apex,minimum size=20pt]
         at (0.5, 3) (c2t) {};
 \node[circle, fill=white]
         at (0.5, 3) (c2) {$c_k$};
 \node[circle]
         at (0, 3) (cd) {$\dots$};
 
 \draw[decorate] (p) -- (p_);
 \draw[] (p) -- (x);
 \draw[] (x) -- (y);
 \draw[] (y) -- (c1);
 \draw[] (y) -- (c2);
 \draw[] (y) -- (cd);

 \end{scope}

 \node[] at (0,1.5) (arr) {$\rightarrow$};
 \node[] at (0,0.5) (l) {\Large\textbf{(f)}};
 \node[] at (0,-0.1) (l) {\Large\textbf{SC}};

 \begin{scope}[shift={(1,0)}]
 
 \node[] at (0, 0) (p) {$p$};
 \node[] at (0, -1) (p_) {};
 \node[] at (1, 2.1) (x) {$x$};
 \node[] at (0, 1.5) (y) {$y$};
 \node[draw,isosceles triangle,isosceles triangle apex angle=60,rotate=270,anchor=apex,minimum size=20pt]
         at (-0.5, 3) (c1t) {};
 \node[circle, fill=white]
         at (-0.5, 3) (c1) {$c_1$};
 \node[draw,isosceles triangle,isosceles triangle apex angle=60,rotate=270,anchor=apex,minimum size=20pt]
         at (0.5, 3) (c2t) {};
 \node[circle, fill=white]
         at (0.5, 3) (c2) {$c_k$};
 \node[circle]
         at (0, 3) (cd) {$\dots$};
 
 \draw[decorate] (p) -- (p_);
 \draw[] (p) -- (y);
 \draw[] (x) -- (y);
 \draw[] (y) -- (c1);
 \draw[] (y) -- (c2);
 \draw[] (y) -- (cd);

 \end{scope}
 \end{scope}
 \label{fig:reg_sad_none}

 \begin{scope}[shift={(5.75,0)}]
 \node[draw=none,fill=none,circle] at (-2.5, -1) (dummy1) {};
 \node[draw=none,fill=none,circle] at (4, -1) (dummy2) {};
 \node[draw=none,fill=none,circle] at (-2.5, 4) (dummy3) {};
 \node[draw=none,fill=none,circle] at (4, 4) (dummy4) {};

 \begin{scope}[shift={(-1,0)}]
 
 \node[] at (0, 0) (p) {$p$};
 \node[] at (0, -1) (p_) {};
 \node[] at (0, 0.9) (x) {$x$};
 \node[] at (0, 1.5) (y) {$y$};
 \node[draw,isosceles triangle,isosceles triangle apex angle=60,rotate=270,anchor=apex,minimum size=20pt]
         at (-0.5, 3) (c1t) {};
 \node[circle, fill=white]
         at (-0.5, 3) (c1) {$c_1$};
 \node[draw,isosceles triangle,isosceles triangle apex angle=60,rotate=270,anchor=apex,minimum size=20pt]
         at (0.5, 3) (c2t) {};
 \node[circle, fill=white]
         at (0.5, 3) (c2) {$c_k$};
 \node[circle]
         at (0, 3) (cd) {$\dots$};
 
 \draw[decorate] (p) -- (p_);
 \draw[] (p) -- (x);
 \draw[] (x) -- (y);
 \draw[] (y) -- (c1);
 \draw[] (y) -- (c2);
 \draw[] (y) -- (cd);

 \end{scope}

 \node[] at (0,1.5) (arr) {$\rightarrow$};
 \node[] at (0,0.5) (l) {\Large\textbf{(g)}};
 \node[] at (0,-0.1) (l) {\Large\textbf{SC}};

 \begin{scope}[shift={(1.2,0)}]
 
 \node[] at (0, 0) (p) {$p$};
 \node[] at (0, -1) (p_) {};
 \node[] at (1.5, 2.1) (x) {$x$};
 \node[] at (0, 1.5) (y) {$y$};
 \node[draw,isosceles triangle,isosceles triangle apex angle=60,rotate=270,anchor=apex,minimum size=20pt]
         at (-0.7, 3) (c1t) {};
 \node[circle, fill=white]
         at (-0.7, 3) (c1) {$c_1$};
 \node[draw,isosceles triangle,isosceles triangle apex angle=60,rotate=270,anchor=apex,minimum size=20pt]
         at (0.7, 3) (c2t) {};
 \node[circle, fill=white]
         at (0.7, 3) (c2) {$c_{k-1}$};
 \node[draw,isosceles triangle,isosceles triangle apex angle=60,rotate=270,anchor=apex,minimum size=20pt]
         at (2, 3) (c3t) {};
 \node[circle, fill=white]
         at (2, 3) (c3) {$c_{k}$};
 \node[circle]
         at (0, 3) (cd) {$\dots$};
 
 \draw[decorate] (p) -- (p_);
 \draw[] (p) -- (y);
 \draw[] (x) -- (y);
 \draw[] (y) -- (c1);
 \draw[] (y) -- (c2);
 \draw[] (y) -- (cd);
 \draw[] (x) -- (c3);

 \end{scope}
 \end{scope}

 \begin{scope}[shift={(12.45,0)}]
 \node[draw=none,fill=none,circle] at (-2.5, -1) (dummy1) {};
 \node[draw=none,fill=none,circle] at (5.5, -1) (dummy2) {};
 \node[draw=none,fill=none,circle] at (-2.5, 4) (dummy3) {};
 \node[draw=none,fill=none,circle] at (5.5, 4) (dummy4) {};

 \begin{scope}[shift={(-1,0)}]
 
 \node[] at (0, 0) (p) {$p$};
 \node[] at (0, -1) (p_) {};
 \node[] at (0, 0.9) (x) {$x$};
 \node[] at (0, 1.5) (y) {$y$};
 \node[draw,isosceles triangle,isosceles triangle apex angle=60,rotate=270,anchor=apex,minimum size=20pt]
         at (-0.5, 3) (c1t) {};
 \node[circle, fill=white]
         at (-0.5, 3) (c1) {$c_1$};
 \node[draw,isosceles triangle,isosceles triangle apex angle=60,rotate=270,anchor=apex,minimum size=20pt]
         at (0.5, 3) (c2t) {};
 \node[circle, fill=white]
         at (0.5, 3) (c2) {$c_k$};
 \node[circle]
         at (0, 3) (cd) {$\dots$};
 
 \draw[decorate] (p) -- (p_);
 \draw[] (p) -- (x);
 \draw[] (x) -- (y);
 \draw[] (y) -- (c1);
 \draw[] (y) -- (c2);
 \draw[] (y) -- (cd);

 \end{scope}

 \node[] at (0,1.2) (arr) {$\rightarrow$};
 \node[] at (0,0.5) (l) {\Large\textbf{(h)}};
 \node[] at (0,-0.1) (l) {\Large\textbf{HS}};

 \begin{scope}[shift={(1.5,0)}]
 
 \node[] at (0, 0) (p) {$p$};
 \node[] at (0, -1) (p_) {};
 \node[] at (2.5, 2.1) (x) {$x$};
 \node[] at (0, 1.5) (y) {$y$};
 
 \node[draw,isosceles triangle,isosceles triangle apex angle=60,rotate=270,anchor=apex,minimum size=20pt]
         at (-0.7, 3) (c1t) {};
 \node[circle, fill=white]
         at (-0.7, 3) (c1) {$c_1$};
 \node[draw,isosceles triangle,isosceles triangle apex angle=60,rotate=270,anchor=apex,minimum size=20pt]
         at (0.7, 3) (cit) {};
 \node[circle, fill=white]
         at (0.7, 3) (ci) {$c_{i}$};
 \node[circle]
         at (0, 3) (cd) {$\dots$};
         
 \node[draw,isosceles triangle,isosceles triangle apex angle=60,rotate=270,anchor=apex,minimum size=20pt]
         at (1.8, 3) (ciit) {};
 \node[circle, fill=white]
         at (1.8, 3) (cii) {$c_{i+1}$};
 \node[draw,isosceles triangle,isosceles triangle apex angle=60,rotate=270,anchor=apex,minimum size=20pt]
         at (3.2, 3) (c2t) {};
 \node[circle, fill=white]
         at (3.2, 3) (c2) {$c_{k}$};
 \node[circle]
         at (2.5, 3) (cdd) {$\dots$};
 
 \draw[decorate] (p) -- (p_);
 \draw[] (p) -- (y);
 \draw[] (x) -- (y);
 \draw[] (y) -- (c1);
 \draw[] (y) -- (ci);
 \draw[] (y) -- (cd);
 \draw[] (x) -- (cii);
 \draw[] (x) -- (c2);
 \draw[] (x) -- (cdd);

 \end{scope}
 \end{scope}
 \label{fig:reg_sad_some}

 \draw[thin, gray!50] (-2.3,4) -- (18.55,4);
 
 \draw[thin, gray!50] (2.6,4.2) -- (2.6,8.8);
 \draw[thin, gray!50] (7,4.2) -- (7,8.8);
 \draw[thin, gray!50] (10.9,4.2) -- (10.9,8.8);
 \draw[thin, gray!50] (14.4,4.2) -- (14.4,8.8);
 
 \draw[thin, gray!50] (3.35,3.8) -- (3.35,-0.8);
 \draw[thin, gray!50] (10,3.8) -- (10,-0.8);

 \end{tikzpicture}
 }
 \caption{Cases for shifting a maximum past a maximum (a) and a regular vertex past a regular~(b-c), maximum~(d) or saddle vertex~(e-h), following Observation~\ref{obs:no_edge}-\ref{obs:children_y}. Change types are annotated.}
 \label{fig:cases_max_reg}
\end{figure}

\begin{figure}[t!]
 \centering 
 \resizebox{\linewidth}{!}{
 \begin{tikzpicture}
 \tikzset{decoration={snake}}

 \begin{scope}[shift={(2.25-0.2,-5)}]
 \node[draw=none,fill=none,circle] at (-3.5, -1) (dummy1) {};
 \node[draw=none,fill=none,circle] at (3.5, -1) (dummy2) {};
 \node[draw=none,fill=none,circle] at (-3.5, 4) (dummy3) {};
 \node[draw=none,fill=none,circle] at (3.5, 4) (dummy4) {};
 
 \begin{scope}[shift={(-1.5,0)}]
 
 \node[] at (0, 0) (p) {$p$};
 \node[] at (0, -1) (p_) {};
 \node[] at (0, 0.9) (x) {$x$};
 \node[] at (-0.5, 1.5) (y) {$y$};
 \node[draw,isosceles triangle,isosceles triangle apex angle=60,rotate=270,anchor=apex,minimum size=20pt]
         at (-1, 3) (c1t) {};
 \node[circle, fill=white]
         at (-1, 3) (c1) {$c_1$};
 \node[draw,isosceles triangle,isosceles triangle apex angle=60,rotate=270,anchor=apex,minimum size=20pt]
         at (0, 3) (c2t) {};
 \node[circle, fill=white]
         at (0, 3) (c2) {$c_2$};
 \node[draw,isosceles triangle,isosceles triangle apex angle=60,rotate=270,anchor=apex,minimum size=20pt]
         at (1, 3) (ckt) {};
 \node[circle, fill=white]
         at (1, 3) (ck) {$c_k$};
 \node[circle]
         at (0.5, 3) (cd) {$\dots$};
 
 \draw[decorate] (p) -- (p_);
 \draw[] (p) -- (x);
 \draw[] (x) -- (y);
 \draw[] (y) -- (c1);
 \draw[] (x) -- (c2);
 \draw[] (x) -- (cd);
 \draw[] (x) -- (ck);

 \end{scope}

 \node[] at (0,1.5) (arr) {$\rightarrow$};
 \node[] at (0,0.5) (l) {\Large\textbf{(a)}};
 \node[] at (0,-0.1) (l) {\Large\textbf{HS}};

 \begin{scope}[shift={(1.5,0)}]
 
 \node[] at (0, 0) (p) {$p$};
 \node[] at (0, -1) (p_) {};
 \node[] at (0.5, 2.1) (x) {$x$};
 \node[] at (0, 1.5) (y) {$y$};
 \node[draw,isosceles triangle,isosceles triangle apex angle=60,rotate=270,anchor=apex,minimum size=20pt]
         at (-1, 3) (c1t) {};
 \node[circle, fill=white]
         at (-1, 3) (c1) {$c_1$};
 \node[draw,isosceles triangle,isosceles triangle apex angle=60,rotate=270,anchor=apex,minimum size=20pt]
         at (0, 3) (c2t) {};
 \node[circle, fill=white]
         at (0, 3) (c2) {$c_2$};
 \node[draw,isosceles triangle,isosceles triangle apex angle=60,rotate=270,anchor=apex,minimum size=20pt]
         at (1, 3) (ckt) {};
 \node[circle, fill=white]
         at (1, 3) (ck) {$c_k$};
 \node[circle]
         at (0.5, 3) (cd) {$\dots$};
 
 \draw[decorate] (p) -- (p_);
 \draw[] (p) -- (y);
 \draw[] (x) -- (y);
 \draw[] (y) -- (c1);
 \draw[] (x) -- (c2);
 \draw[] (x) -- (cd);
 \draw[] (x) -- (ck);

 \end{scope}
 \end{scope}
 \label{fig:sad_reg_one_many}

 \begin{scope}[shift={(8.25,-5)}]
 \node[draw=none,fill=none,circle] at (-2.5, -1) (dummy1) {};
 \node[draw=none,fill=none,circle] at (2.5, -1) (dummy2) {};
 \node[draw=none,fill=none,circle] at (-2.5, 4) (dummy3) {};
 \node[draw=none,fill=none,circle] at (2.5, 4) (dummy4) {};
 
 \begin{scope}[shift={(-1,0)}]
 
 \node[] at (0, 0) (p) {$p$};
 \node[] at (0, -1) (p_) {};
 \node[] at (0, 0.9) (x) {$x$};
 \node[] at (-0.5, 1.5) (y) {$y$};
 \node[draw,isosceles triangle,isosceles triangle apex angle=60,rotate=270,anchor=apex,minimum size=20pt]
         at (-0.5, 3) (c1t) {};
 \node[circle, fill=white]
         at (-0.5, 3) (c1) {$c_1$};
 \node[draw,isosceles triangle,isosceles triangle apex angle=60,rotate=270,anchor=apex,minimum size=20pt]
         at (0.5, 3) (c2t) {};
 \node[circle, fill=white]
         at (0.5, 3) (c2) {$c_2$};
 
 \draw[decorate] (p) -- (p_);
 \draw[] (p) -- (x);
 \draw[] (x) -- (y);
 \draw[] (y) -- (c1);
 \draw[] (x) -- (c2);

 \end{scope}

 \node[] at (0,1) (arr) {$\rightarrow$};
 \node[] at (0,0.5) (l) {\Large\textbf{(b)}};
 \node[] at (0,-0.1) (l) {\Large\textbf{SC}};

 \begin{scope}[shift={(1,0)}]
 
 \node[] at (0, 0) (p) {$p$};
 \node[] at (0, -1) (p_) {};
 \node[] at (0.5, 2.1) (x) {$x$};
 \node[] at (0, 1.5) (y) {$y$};
 \node[draw,isosceles triangle,isosceles triangle apex angle=60,rotate=270,anchor=apex,minimum size=20pt]
         at (-0.5, 3) (c1t) {};
 \node[circle, fill=white]
         at (-0.5, 3) (c1) {$c_1$};
 \node[draw,isosceles triangle,isosceles triangle apex angle=60,rotate=270,anchor=apex,minimum size=20pt]
         at (0.5, 3) (c2t) {};
 \node[circle, fill=white]
         at (0.5, 3) (c2) {$c_2$};
 
 \draw[decorate] (p) -- (p_);
 \draw[] (p) -- (y);
 \draw[] (x) -- (y);
 \draw[] (y) -- (c1);
 \draw[] (x) -- (c2);

 \end{scope}
 \end{scope}
 \label{fig:sad_reg_one_one}

 \begin{scope}[shift={(14.25+0.2,-5)}]
 \node[draw=none,fill=none,circle] at (-3.5, -1) (dummy1) {};
 \node[draw=none,fill=none,circle] at (2.5, -1) (dummy2) {};
 \node[draw=none,fill=none,circle] at (-3.5, 4) (dummy3) {};
 \node[draw=none,fill=none,circle] at (2.5, 4) (dummy4) {};
 
 \begin{scope}[shift={(-1.5,0)}]
 
 \node[] at (0, 0) (p) {$p$};
 \node[] at (0, -1) (p_) {};
 \node[] at (0, 0.9) (x) {$x$};
 \node[] at (-0.5, 1.5) (y) {$y$};
 \node[draw,isosceles triangle,isosceles triangle apex angle=60,rotate=270,anchor=apex,minimum size=20pt]
         at (-1, 3) (c1t) {};
 \node[circle, fill=white]
         at (-1, 3) (c1) {$c_1$};
 \node[draw,isosceles triangle,isosceles triangle apex angle=60,rotate=270,anchor=apex,minimum size=20pt]
         at (0, 3) (c2t) {};
 \node[circle, fill=white]
         at (0, 3) (c2) {$c_2$};
 \node[draw,isosceles triangle,isosceles triangle apex angle=60,rotate=270,anchor=apex,minimum size=20pt]
         at (1, 3) (ckt) {};
 \node[circle, fill=white]
         at (1, 3) (ck) {$c_k$};
 \node[circle]
         at (0.5, 3) (cd) {$\dots$};
 
 \draw[decorate] (p) -- (p_);
 \draw[] (p) -- (x);
 \draw[] (x) -- (y);
 \draw[] (y) -- (c1);
 \draw[] (x) -- (c2);
 \draw[] (x) -- (cd);
 \draw[] (x) -- (ck);

 \end{scope}

 \node[] at (0,1.5) (arr) {$\rightarrow$};
 \node[] at (0,0.5) (l) {\Large\textbf{(c)}};
 \node[] at (0,-0.1) (l) {\Large\textbf{SC}};

 \begin{scope}[shift={(1,0)}]
 
 \node[] at (0, 0) (p) {$p$};
 \node[] at (0, -1) (p_) {};
 \node[] at (0, 2.1) (x) {$x$};
 \node[] at (0, 1.5) (y) {$y$};
 \node[draw,isosceles triangle,isosceles triangle apex angle=60,rotate=270,anchor=apex,minimum size=20pt]
         at (-0.5, 3) (c1t) {};
 \node[circle, fill=white]
         at (-0.5, 3) (c1) {$c_1$};
 \node[draw,isosceles triangle,isosceles triangle apex angle=60,rotate=270,anchor=apex,minimum size=20pt]
         at (0.5, 3) (ckt) {};
 \node[circle, fill=white]
         at (0.5, 3) (ck) {$c_k$};
 \node[circle]
         at (0, 3) (cd) {$\dots$};
 
 \draw[decorate] (p) -- (p_);
 \draw[] (p) -- (y);
 \draw[] (x) -- (y);
 \draw[] (x) -- (c1);
 \draw[] (x) -- (cd);
 \draw[] (x) -- (ck);

 \end{scope}
 \end{scope}
 \label{fig:sad_reg_none}

 \begin{scope}[shift={(1.5-0.2,-10.2)}]
 \node[draw=none,fill=none,circle] at (-2.5, -1) (dummy1) {};
 \node[draw=none,fill=none,circle] at (2.5, -1) (dummy2) {};
 \node[draw=none,fill=none,circle] at (-2.5, 4) (dummy3) {};
 \node[draw=none,fill=none,circle] at (2.5, 4) (dummy4) {};
 
 \begin{scope}[shift={(-1.5,0)}]
 
 \node[] at (0, 0) (p) {$p$};
 \node[] at (0, -1) (p_) {};
 \node[] at (0, 0.9) (x) {$x$};
 \node[] at (-0.5, 1.5) (y) {$y$};
 \node[draw,isosceles triangle,isosceles triangle apex angle=60,rotate=270,anchor=apex,minimum size=20pt]
         at (0, 3) (c1t) {};
 \node[circle, fill=white]
         at (0, 3) (c1) {$c_1$};
 \node[draw,isosceles triangle,isosceles triangle apex angle=60,rotate=270,anchor=apex,minimum size=20pt]
         at (1, 3) (ckt) {};
 \node[circle, fill=white]
         at (1, 3) (ck) {$c_k$};
 \node[circle]
         at (0.5, 3) (cd) {$\dots$};
 
 \draw[decorate] (p) -- (p_);
 \draw[] (p) -- (x);
 \draw[] (x) -- (y);
 \draw[] (x) -- (c1);
 \draw[] (x) -- (cd);
 \draw[] (x) -- (ck);

 \end{scope}

 \node[] at (0,1.5) (arr) {$\rightarrow$};
 \node[] at (0,0.5) (l) {\Large\textbf{(d)}};
 \node[] at (0,-0.1) (l) {\Large\textbf{SC}};

 \begin{scope}[shift={(1,0)}]
 
 \node[] at (0, 0) (p) {$p$};
 \node[] at (0, -1) (p_) {};
 \node[] at (0, 2.1) (x) {$x$};
 \node[] at (0, 1.5) (y) {$y$};
 \node[draw,isosceles triangle,isosceles triangle apex angle=60,rotate=270,anchor=apex,minimum size=20pt]
         at (-0.5, 3) (c1t) {};
 \node[circle, fill=white]
         at (-0.5, 3) (c1) {$c_1$};
 \node[draw,isosceles triangle,isosceles triangle apex angle=60,rotate=270,anchor=apex,minimum size=20pt]
         at (0.5, 3) (ckt) {};
 \node[circle, fill=white]
         at (0.5, 3) (ck) {$c_k$};
 \node[circle]
         at (0, 3) (cd) {$\dots$};
 
 \draw[decorate] (p) -- (p_);
 \draw[] (p) -- (y);
 \draw[] (x) -- (y);
 \draw[] (x) -- (c1);
 \draw[] (x) -- (cd);
 \draw[] (x) -- (ck);

 \end{scope}
 \end{scope}
 \label{fig:sad_max_many}

 \begin{scope}[shift={(6,-10.2)}]
 \node[draw=none,fill=none,circle] at (-2, -1) (dummy1) {};
 \node[draw=none,fill=none,circle] at (2, -1) (dummy2) {};
 \node[draw=none,fill=none,circle] at (-2, 4) (dummy3) {};
 \node[draw=none,fill=none,circle] at (2, 4) (dummy4) {};
 
 \begin{scope}[shift={(-1,0)}]
 
 \node[] at (0, 0) (p) {$p$};
 \node[] at (0, -1) (p_) {};
 \node[] at (0, 0.9) (x) {$x$};
 \node[] at (-0.5, 1.5) (y) {$y$};
 \node[draw,isosceles triangle,isosceles triangle apex angle=60,rotate=270,anchor=apex,minimum size=20pt]
         at (0.5, 3) (c1t) {};
 \node[circle, fill=white]
         at (0.5, 3) (c1) {$c_1$};
 
 \draw[decorate] (p) -- (p_);
 \draw[] (p) -- (x);
 \draw[] (x) -- (y);
 \draw[] (x) -- (c1);

 \end{scope}

 \node[] at (0,1.5) (arr) {$\rightarrow$};
 \node[] at (0,0.5) (l) {\Large\textbf{(e)}};
 \node[] at (0,-0.1) (l) {\Large\textbf{ES}};

 \begin{scope}[shift={(1,0)}]
 
 \node[] at (0, 0) (p) {$p$};
 \node[] at (0, -1) (p_) {};
 \node[] at (0, 2.1) (x) {$x$};
 \node[] at (0, 1.5) (y) {$y$};
 \node[draw,isosceles triangle,isosceles triangle apex angle=60,rotate=270,anchor=apex,minimum size=20pt]
         at (0, 3) (c1t) {};
 \node[circle, fill=white]
         at (0, 3) (c1) {$c_1$};
 
 \draw[decorate] (p) -- (p_);
 \draw[] (p) -- (y);
 \draw[] (x) -- (y);
 \draw[] (x) -- (c1);

 \end{scope}
 \end{scope}
 \label{fig:sad_max_one}

 \begin{scope}[shift={(13+0.2,-10.2)}]
 \node[draw=none,fill=none,circle] at (-5, -1) (dummy1) {};
 \node[draw=none,fill=none,circle] at (3.5, -1) (dummy2) {};
 \node[draw=none,fill=none,circle] at (-5, 4) (dummy3) {};
 \node[draw=none,fill=none,circle] at (3.5, 4) (dummy4) {};
 
 \begin{scope}[shift={(-2,0)}]
 
 \node[] at (0, 0) (p) {$p$};
 \node[] at (0, -1) (p_) {};
 \node[] at (0, 1) (x) {$x$};
 \node[] at (-1.5, 2) (y) {$y$};
 
 \node[draw,isosceles triangle,isosceles triangle apex angle=60,rotate=270,anchor=apex,minimum size=20pt]
         at (-2.2, 3) (c1t) {};
 \node[circle, fill=white]
         at (-2.2, 3) (c1) {$c_1$};
 \node[draw,isosceles triangle,isosceles triangle apex angle=60,rotate=270,anchor=apex,minimum size=20pt]
         at (-0.8, 3) (cit) {};
 \node[circle, fill=white]
         at (-0.8, 3) (ci) {$c_{i}$};
 \node[circle]
         at (-1.5, 3) (cd) {$\dots$};
         
 \node[draw,isosceles triangle,isosceles triangle apex angle=60,rotate=270,anchor=apex,minimum size=20pt]
         at (0.8, 3) (ciit) {};
 \node[circle, fill=white]
         at (0.8, 3) (cii) {$c_{i+1}$};
 \node[draw,isosceles triangle,isosceles triangle apex angle=60,rotate=270,anchor=apex,minimum size=20pt]
         at (2.2, 3) (c2t) {};
 \node[circle, fill=white]
         at (2.2, 3) (c2) {$c_{k}$};
 \node[circle]
         at (1.5, 3) (cdd) {$\dots$};
 
 \draw[decorate] (p) -- (p_);
 \draw[] (p) -- (x);
 \draw[] (x) -- (y);
 \draw[] (y) -- (c1);
 \draw[] (y) -- (ci);
 \draw[] (y) -- (cd);
 \draw[] (x) -- (cii);
 \draw[] (x) -- (c2);
 \draw[] (x) -- (cdd);

 \end{scope}

 \node[] at (0,1.5) (arr) {$\rightarrow$};
 \node[] at (0,0.5) (l) {\Large\textbf{(f)}};
 \node[] at (0,-0.1) (l) {\Large\textbf{HS}};

 \begin{scope}[shift={(2,0)}]
 
 \node[] at (0, 0) (p) {$p$};
 \node[] at (0, -1) (p_) {};
 \node[] at (0, 2) (x) {$x$};
 \node[] at (0, 1) (y) {$y$};
 \node[draw,isosceles triangle,isosceles triangle apex angle=60,rotate=270,anchor=apex,minimum size=20pt]
         at (-0.5, 3) (c1t) {};
 \node[circle, fill=white]
         at (-0.5, 3) (c1) {$c_1$};
 \node[draw,isosceles triangle,isosceles triangle apex angle=60,rotate=270,anchor=apex,minimum size=20pt]
         at (0.5, 3) (ckt) {};
 \node[circle, fill=white]
         at (0.5, 3) (ck) {$c_k$};
 \node[circle]
         at (0, 3) (cd) {$\dots$};
 
 \draw[decorate] (p) -- (p_);
 \draw[] (p) -- (y);
 \draw[] (x) -- (y);
 \draw[] (x) -- (c1);
 \draw[] (x) -- (cd);
 \draw[] (x) -- (ck);

 \end{scope}
 \end{scope}
 \label{fig:sad_sad_all}

 \begin{scope}[shift={(3.5-0.1,-15.4)}]
 \node[draw=none,fill=none,circle] at (-6, -1) (dummy1) {};
 \node[draw=none,fill=none,circle] at (6, -1) (dummy2) {};
 \node[draw=none,fill=none,circle] at (-6, 4) (dummy3) {};
 \node[draw=none,fill=none,circle] at (6, 4) (dummy4) {};
 
 \begin{scope}[shift={(-3,0)}]
 
 \node[] at (0, 0) (p) {$p$};
 \node[] at (0, -1) (p_) {};
 \node[] at (0, 1) (x) {$x$};
 \node[] at (-1.5, 2) (y) {$y$};
 
 \node[draw,isosceles triangle,isosceles triangle apex angle=60,rotate=270,anchor=apex,minimum size=20pt]
         at (-2.2, 3) (c1t) {};
 \node[circle, fill=white]
         at (-2.2, 3) (c1) {$c_1$};
 \node[draw,isosceles triangle,isosceles triangle apex angle=60,rotate=270,anchor=apex,minimum size=20pt]
         at (-0.8, 3) (cit) {};
 \node[circle, fill=white]
         at (-0.8, 3) (ci) {$c_{i}$};
 \node[circle]
         at (-1.5, 3) (cd) {$\dots$};
         
 \node[draw,isosceles triangle,isosceles triangle apex angle=60,rotate=270,anchor=apex,minimum size=20pt]
         at (0.8, 3) (ciit) {};
 \node[circle, fill=white]
         at (0.8, 3) (cii) {$c_{i+1}$};
 \node[draw,isosceles triangle,isosceles triangle apex angle=60,rotate=270,anchor=apex,minimum size=20pt]
         at (2.2, 3) (c2t) {};
 \node[circle, fill=white]
         at (2.2, 3) (c2) {$c_{k}$};
 \node[circle]
         at (1.5, 3) (cdd) {$\dots$};
 
 \draw[decorate] (p) -- (p_);
 \draw[] (p) -- (x);
 \draw[] (x) -- (y);
 \draw[] (y) -- (c1);
 \draw[] (y) -- (ci);
 \draw[] (y) -- (cd);
 \draw[] (x) -- (cii);
 \draw[] (x) -- (c2);
 \draw[] (x) -- (cdd);

 \end{scope}

 \node[] at (0,1.5) (arr) {$\rightarrow$};
 \node[] at (0,0.5) (l) {\Large\textbf{(g)}};
 \node[] at (0,-0.1) (l) {\Large\textbf{HS}};

 \begin{scope}[shift={(3,0)}]
 
 \node[] at (0, 0) (p) {$p$};
 \node[] at (0, -1) (p_) {};
 \node[] at (1.5, 2) (x) {$x$};
 \node[] at (0, 1) (y) {$y$};
 
 \node[draw,isosceles triangle,isosceles triangle apex angle=60,rotate=270,anchor=apex,minimum size=20pt]
         at (-2.2, 3) (c1t) {};
 \node[circle, fill=white]
         at (-2.2, 3) (c1) {$c_1$};
 \node[draw,isosceles triangle,isosceles triangle apex angle=60,rotate=270,anchor=apex,minimum size=20pt]
         at (-0.8, 3) (cit) {};
 \node[circle, fill=white]
         at (-0.8, 3) (ci) {$c_{j}$};
 \node[circle]
         at (-1.5, 3) (cd) {$\dots$};
         
 \node[draw,isosceles triangle,isosceles triangle apex angle=60,rotate=270,anchor=apex,minimum size=20pt]
         at (0.8, 3) (ciit) {};
 \node[circle, fill=white]
         at (0.8, 3) (cii) {$c_{j+1}$};
 \node[draw,isosceles triangle,isosceles triangle apex angle=60,rotate=270,anchor=apex,minimum size=20pt]
         at (2.2, 3) (c2t) {};
 \node[circle, fill=white]
         at (2.2, 3) (c2) {$c_{k}$};
 \node[circle]
         at (1.5, 3) (cdd) {$\dots$};
 
 \draw[decorate] (p) -- (p_);
 \draw[] (p) -- (y);
 \draw[] (x) -- (y);
 \draw[] (y) -- (c1);
 \draw[] (y) -- (ci);
 \draw[] (y) -- (cd);
 \draw[] (x) -- (cii);
 \draw[] (x) -- (c2);
 \draw[] (x) -- (cdd);

 \end{scope}
 \end{scope}
 \label{fig:sad_sad_many_many}

 \begin{scope}[shift={(14.5+0.1,-15.4)}]
 \node[draw=none,fill=none,circle] at (-5, -1) (dummy1) {};
 \node[draw=none,fill=none,circle] at (3.5, -1) (dummy2) {};
 \node[draw=none,fill=none,circle] at (-5, 4) (dummy3) {};
 \node[draw=none,fill=none,circle] at (3.5, 4) (dummy4) {};
 
 \begin{scope}[shift={(-2,0)}]
 
 \node[] at (-0.5, 0) (p) {$p$};
 \node[] at (-0.5, -1) (p_) {};
 \node[] at (-0.5, 1) (x) {$x$};
 \node[] at (-1.5, 2) (y) {$y$};
 
 \node[draw,isosceles triangle,isosceles triangle apex angle=60,rotate=270,anchor=apex,minimum size=20pt]
         at (-2.2, 3) (c1t) {};
 \node[circle, fill=white]
         at (-2.2, 3) (c1) {$c_1$};
 \node[draw,isosceles triangle,isosceles triangle apex angle=60,rotate=270,anchor=apex,minimum size=20pt]
         at (-0.8, 3) (cit) {};
 \node[circle, fill=white]
         at (-0.8, 3) (ci) {$c_{k-1}$};
 \node[circle]
         at (-1.5, 3) (cd) {$\dots$};
         
 \node[draw,isosceles triangle,isosceles triangle apex angle=60,rotate=270,anchor=apex,minimum size=20pt]
         at (0.5, 3) (ckt) {};
 \node[circle, fill=white]
         at (0.5, 3) (ck) {$c_{k}$};
 
 \draw[decorate] (p) -- (p_);
 \draw[] (p) -- (x);
 \draw[] (x) -- (y);
 \draw[] (y) -- (c1);
 \draw[] (y) -- (ci);
 \draw[] (y) -- (cd);
 \draw[] (x) -- (ck);

 \end{scope}

 \node[] at (-0.5,1.5) (arr) {$\rightarrow$};
 \node[] at (-0.5,0.5) (l) {\Large\textbf{(h)}};
 \node[] at (-0.5,-0.1) (l) {\Large\textbf{HS}};

 \begin{scope}[shift={(2,0)}]
 
 \node[] at (-0.5, 0) (p) {$p$};
 \node[] at (-0.5, -1) (p_) {};
 \node[] at (0, 2) (x) {$x$};
 \node[] at (-0.5, 1) (y) {$y$};
 
 \node[draw,isosceles triangle,isosceles triangle apex angle=60,rotate=270,anchor=apex,minimum size=20pt]
         at (-2.2, 3) (c1t) {};
 \node[circle, fill=white]
         at (-2.2, 3) (c1) {$c_1$};
 \node[draw,isosceles triangle,isosceles triangle apex angle=60,rotate=270,anchor=apex,minimum size=20pt]
         at (-0.8, 3) (cit) {};
 \node[circle, fill=white]
         at (-0.8, 3) (ci) {$c_{k-1}$};
 \node[circle]
         at (-1.5, 3) (cd) {$\dots$};
         
 \node[draw,isosceles triangle,isosceles triangle apex angle=60,rotate=270,anchor=apex,minimum size=20pt]
         at (0.5, 3) (ckt) {};
 \node[circle, fill=white]
         at (0.5, 3) (ck) {$c_{k}$};
 
 \draw[decorate] (p) -- (p_);
 \draw[] (p) -- (y);
 \draw[] (x) -- (y);
 \draw[] (y) -- (c1);
 \draw[] (y) -- (ci);
 \draw[] (y) -- (cd);
 \draw[] (x) -- (ck);

 \end{scope}
 \end{scope}
 \label{fig:sad_sad_none_one}

 \draw[thin, gray!50] (-1.1,-6.2) -- (16.5,-6.2);
 \draw[thin, gray!50] (-1.5,-11.4) -- (17,-11.4);
 
 \draw[thin, gray!50] (5.7,-6) -- (5.7,-1.2);
 \draw[thin, gray!50] (10.85,-6) -- (10.85,-1.2);
 
 \draw[thin, gray!50] (3.8,-6.4) -- (3.8,-11.2);
 \draw[thin, gray!50] (8.1,-6.4) -- (8.1,-11.2);
 
 \draw[thin, gray!50] (9.5,-11.6) -- (9.5,-16.2);
 
 \end{tikzpicture}
 }
 \caption{Cases for shifting a saddle vertex past a regular~(a-c), maximum~(d-e) or saddle vertex~(f-h), following Observation~\ref{obs:no_edge}-\ref{obs:children_y}. Change types are annotated below markers.}
 \label{fig:cases_sad}

\end{figure}
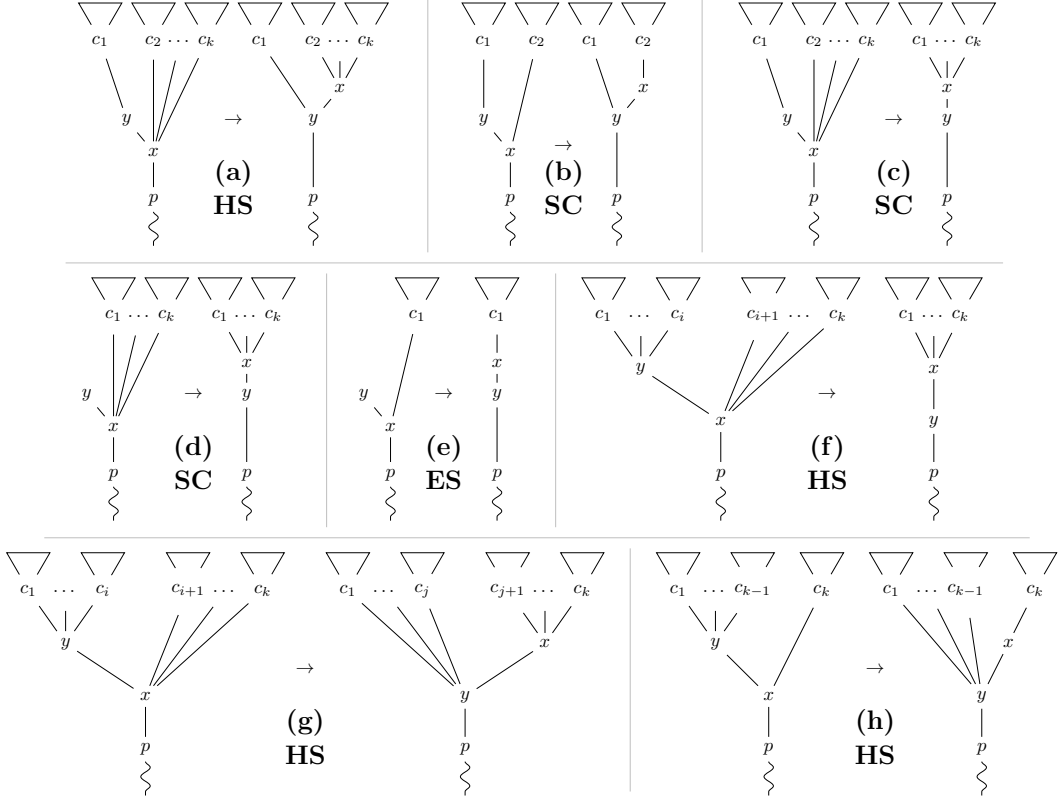

\subsection*{Structural Changes in the Unaugmented Merge Tree}

The changes in $\mathcal{T}_f$ induce similar changes in the unaugmented merge tree $T_f$.
In Appendix~\ref{sec:case_distinction}, we provide a complete case distinction considering all critical types of the involved vertices and the number of redirected edges (following Observations~\ref{obs:no_edge}-\ref{obs:children_y}).
The important cases are shown in Figures~\ref{fig:cases_max_reg} (perturbation of maximum or regular vertex) and~\ref{fig:cases_sad} (perturbation of a saddle).
We observe that all types of merge tree changes only appear in very specific forms for minimal perturbations, fitting exactly the intuitions of their definition:
\begin{itemize}
    \item Vertical swaps only come in the form of two maxima swapping position, whose branches are in a direct parent-child relationship (Figure~\ref{fig:cases_max_reg}a).
    \item Edge splits only occur in pure form, i.e., a single edge is split by a new saddle node, which also has exactly one new maximum node as its second child (Figure~\ref{fig:cases_max_reg}b). The inverse is, of course, also possible (Figure~\ref{fig:cases_sad}e).
    \item Horizontal swaps only appear in two forms (with their corresponding inverse). Either a saddle is split into two connected saddles with the original children being partitioned between the two new nodes (Figures~\ref{fig:cases_max_reg}h, \ref{fig:cases_sad}a, \ref{fig:cases_sad}f, \ref{fig:cases_sad}h). Or two saddles are swapped, potentially with some children moving from one to the other (Figure~\ref{fig:cases_sad}g).
    \item Simple changes are either no changes at all (Figures~\ref{fig:cases_max_reg}c, \ref{fig:cases_max_reg}g) a single node moving in scalar value (Figures~\ref{fig:cases_max_reg}d, \ref{fig:cases_max_reg}e, \ref{fig:cases_sad}b, \ref{fig:cases_sad}c), a new maximum node being appended to an existing saddle node (Figure~\ref{fig:cases_max_reg}f), or the inverse of the latter (Figure~\ref{fig:cases_sad}d).
\end{itemize}

\noindent
We summarize the results in the following lemma.
\begin{lemma}
\label{lemma:changes}
    Let $\mathbb{X},f$ and $\mathbb{X},f'$ be two piecewise linear scalar fields that differ by a minimal vertex perturbation of vertex $x$ passing $y$, such that $f'(x) > f(x)$.
    \begin{itemize}
    \item If $T_f$ and $T_{f'}$ differ by a \emph{vertical swap}, then the two merge trees themselves remain unchanged, we have $T_f = T_{f'}$ and exactly one label differs. 
    The swapped nodes $x,y$ are maximum nodes in both trees with $\mathcal{B}_f(y)$ a parent of $\mathcal{B}_f(x)$ and $\mathcal{B}_{f'}(x)$ a parent of $\mathcal{B}_{f'}(y)$.

    \item If $T_f$ and $T_{f'}$ differ by an \emph{edge split} with $\bdt(T_f) \subseteq \bdt(T_{f'})$, then we have $V(T_{f'}) = V(T_f) \cupdot \{x,y\}$,
    $E(T_{f'}) = (E(T_f) \setminus \{(c,p)\}) \cup \{(x,y),(c,y),(y,p)\}$ for some $(c,p) \in E(T_f)$ and $V(\bdt(T_{f'})) = V(\bdt(T_f)) \cup \{yx\}$.

    \item If $T_f$ and $T_{f'}$ differ by a \emph{horizontal swap}, then we either have $(y,x) \in E(T_f)$ and $(x,y) \in E(T_{f'})$ with $x,y$ both saddles in both $T_f,T_{f'}$; or we have (ignoring renaming and the inverse change) $x \notin V(T_f)$ and $(x,y) \in E(T_{f'})$ with $x,y$ both saddles in $T_{f'}$.

    \item If $T_f$ and $T_{f'}$ differ by a \emph{simple change}, we either have $T_f = T_{f'}$ and $\bdt(T_f) = \bdt(T_{f'})$ (up to renaming); or $\big||V(T_f)| - |V(T_{f'})|\big| = 1$, $\big||E(T_f)| - |E(T_{f'})|\big| = 1$, $\big||V(\bdt(T_f))| - |V(\bdt(T_{f'}))|\big| = 1$ and $\big||E(\bdt(T_f))| - |E(\bdt(T_{f'}))|\big| = 1$.
    \end{itemize}
\end{lemma}
\begin{proof}
    See Appendix~\ref{sec:case_distinction}.
\end{proof}

\section{Edit Distance Stability}
\label{sec:perturbation_stability}

Using the observations in Lemma~\ref{lemma:changes}, we now go through the different classes of perturbations and show which edit distances exhibit stable behavior and which do not.
Afterwards, we summarize the results in \mbox{Theorems \ref{theorem:stability_deform}-\ref{theorem:stability_genclas}.}
Due to space limitations, we skip the discussion for the Sridharamurthy distance and branch mapping distance, as they form outliers in terms of definition.
The corresponding arguments can be found in Appendix~\ref{sec:perturbation_stability_app}.
The Sridharamurthy distance behaves very similar to the BDT-based distances, the branch mapping distance very similar to the path mapping distance.

\subsection*{Simple Changes}

Let $f,g$ be two scalar fields with a minimal vertex perturbation of vertex $x$ by $\epsilon$ between them, such that $T_f \subseteq T_g$ and $\bdt(T_f) \subseteq \bdt(T_g)$.
It is easy to see that all distances can be bound by $\deg(T_f)\cdot\epsilon$. 
If we move a saddle, possibly multiple incident edges or branches need to be relabeled ($\deg(T_f)$ many), otherwise only a single one or none at all.
See Appendix~\ref{sec:perturbation_stability_app} for the full arguments.

\subsection*{Edge splits}
Let $f,g$ be two scalar fields with a minimal vertex perturbation of $x$ by $\epsilon$ between them, such that $\obdt(T_f) \subset \obdt(T_g)$ (inclusion strict by Lemma~\ref{lemma:changes}) but neither $T_f \subseteq T_g$ nor $T_f \subseteq T_g$ hold.

\medskip
\noindent
\textbf{BDT-based distances.}
By Lemma~\ref{lemma:changes} we append a single leaf branch of length at most~$\epsilon$, which also bounds its distance (in birth-death space) to its corresponding diagonal point~\cite{DBLP:journals/tvcg/PontVDT22}. 
Both BDT-based distances can perform a single insertion with costs of exactly this distance.

\medskip
\noindent
\textbf{Classic edit distances.}
Classic constrained edit distances cannot handle such changes.
By Lemma~\ref{lemma:changes}, a single edge is split by a new saddle with a new child.
Figure~\ref{fig:edge_split} shows a simple counter example.
The perturbation results in a new edge of length $\epsilon$.
However, any classic edit distance, even the most general one will result in a distance larger than $x$: the left tree has three edges of length $x$ and one of length $2x$ whereas the right one has one of length $x$ and two of length $2x$, which obviously results in mappings where one edge of length $x$ remains unmapped, i.e.\ deleted or inserted.
Thus, with $x$ arbitrary, there is no constant $c$ such that the costs are bound by $c \cdot \epsilon$.

\begin{figure}
    \centering
    \resizebox{0.5\linewidth}{!}{
    \Large
    \bfseries
    \begin{tikzpicture}[xscale=1,yscale=0.8]
    
    \node[draw,circle,fill=gray!100,minimum width=1cm,label={[left=10pt]$0$}] at (0, 0) (root_1) {A};
    \node[draw,circle,fill=gray!100,minimum width=1cm,label={[left=10pt]$x$}] at (0, 3) (s1_1) {B};
    \node[draw,circle,fill=red!80,minimum width=1cm,label={[above]$3x$}] at (-2, 9) (m1_1) {C};
    \node[draw,circle,fill=red!80,minimum width=1cm,label={[above]$3x$}] at (2, 9) (m2_1) {D};
    \node[draw,circle,fill=gray!100,minimum width=1cm,label={[right=10pt]$2x$}] at (1, 6) (s2_1) {E};
    \node[draw,circle,fill=red!80,minimum width=1cm,label={[above]$2x+\epsilon$}] at (0.2, 7.2) (m3_1) {F};
    \draw[gray,very thick] (root_1) -- (s1_1);
    \draw[gray,very thick] (s1_1) -- (m1_1);
    \draw[gray,very thick] (s1_1) -- (s2_1);
    \draw[gray,very thick] (s2_1) -- (m2_1);
    \draw[gray,very thick] (s2_1) -- (m3_1);
    
    \node[draw,circle,fill=gray!100,minimum width=1cm,label={[right=10pt]$0$}] at (0+7, 0) (root_2) {A};
    \node[draw,circle,fill=gray!100,minimum width=1cm,label={[right=10pt]$x$}] at (0+7, 3) (s1_2) {B};
    \node[draw,circle,fill=red!80,minimum width=1cm,label={[above]$3x$}] at (-2+7, 9) (m1_2) {C};
    \node[draw,circle,fill=red!80,minimum width=1cm,label={[above]$3x$}] at (2+7, 9) (m2_2) {D};
    \draw[gray,very thick] (root_2) -- (s1_2);
    \draw[gray,very thick] (s1_2) -- (m1_2);
    \draw[gray,very thick] (s1_2) -- (m2_2);
    
    \draw[ultra thick,-,dotted,green] (root_1) -- (root_2);
    \draw[ultra thick,-,dotted,green] (s1_1) -- (s1_2);
    \draw[ultra thick,-,dotted,green] (m2_1) to[bend left=20] (m1_2);
    \draw[ultra thick,-,dotted,green] (m1_1) to[bend left=35] (m2_2);
    
    \end{tikzpicture}
    }
    \caption{Two merge trees that can be transformed into each other by a minimal vertex perturbation of vertex $F$ (cf.\ Figure~\ref{fig:cases_max_reg}b). The optimal edit mapping (for large $x$) is shown in green.}
    \label{fig:edge_split}
\end{figure}
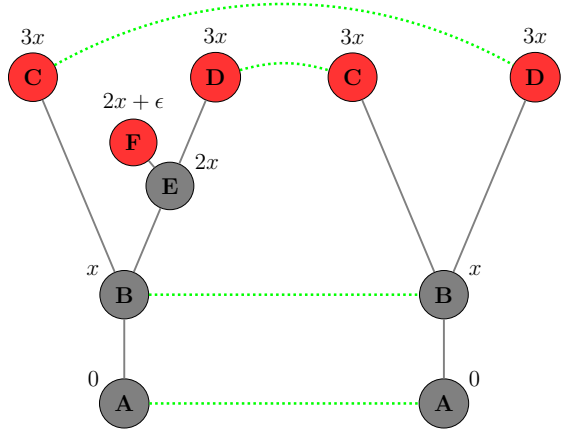

\medskip
\noindent
\textbf{Deformation-based distances.}
By Lemma~\ref{lemma:changes}, only one type of change can happen: a single edge is split by a new saddle $y$ with a new child $x$.
This is exactly covered by the deformation-based insertion of $(x,y)$.
Since for these distances, it is a simple leaf insertion, it is also covered by the one-degree variant, i.e.\ the path mapping distance.
The cost of the operation is bounded by the length of the new edge, i.e.~$\epsilon$.

\subsection*{Vertical branch swaps}
Let $f,g$ be two scalar fields with a minimal vertex perturbation of $x$ by $\epsilon$ between them, such that $T_f \subseteq T_g$ but neither $\obdt(T_f) \subseteq \obdt(T_g)$ nor $\obdt(T_f) \subseteq \obdt(T_g)$ hold.

\medskip
\noindent
\textbf{Classic and deformation-based edit distances.}
By Lemma~\ref{lemma:changes}, we know that $T_f = T_g$ and only a single edge changes its length. 
All classic and deformation-based edit distances can do this by a single relabel and are bounded by~$\epsilon$.

\medskip
\noindent
\textbf{BDT-based distances.}
For the BDT-based distances, we can construct a simple counter example such that there is no constant $c$ for which the costs are bound by $c \cdot \epsilon$, see several works by Wetzels et al.~\cite{wetzels2022branch,wetzels2022path}.
We provide the example in Appendix~\ref{sec:perturbation_stability_app}.

\subsection*{Horizontal Swaps}
Let $f,g$ be two scalar fields with a minimal vertex perturbation of $x$ by $\epsilon$ between them, such that neither $\obdt(T_f) \subseteq \obdt(T_g)$ nor $\obdt(T_g) \subseteq \obdt(T_f)$ and also neither $T_f \subseteq T_g$ nor $T_f \subseteq T_g$ hold.

\medskip
\noindent
\textbf{BDT-based distances.}
All BDT-based distances fail to capture horizontal swaps properly.
This can be shown through a simple counter example.
By Lemma~\ref{lemma:changes}, two saddles can swap between $T_f$ and $T_g$ as shown in Figure~\ref{fig:horizontal_instability}, together with the corresponding BDTs.
To transform one BDT into the other, we have to change the parent-child relationship between the branches $(d,b)$ and $(f,e)$.
No edit distance can achieve this without insertion or deletion of one of the branches, inducing costs of at least $x-\epsilon$  (length of the branch).
Thus, there is no constant $c$ such that the costs are bound by $c \cdot \epsilon$.

\medskip
\noindent
\textbf{Constrained edit distances.}
On the same instance as above, $\delta_L$ and $\delta_P$ fail for similar though different reasons.
Mapping the three leaf edges identically onto each other is necessary (any relabeling of the leafs induces costs of at least $x-\epsilon$), but needs the inner edge between $b$ and $e$ to be deleted and re-inserted.
This is not possible using constrained edit distances.
Therefore, one of the leaf edges has to be deleted and re-inserted, inducing costs of at least $x-\epsilon$.
Thus, there is no constant $c$ such that the costs are bound by $c \cdot \epsilon$.

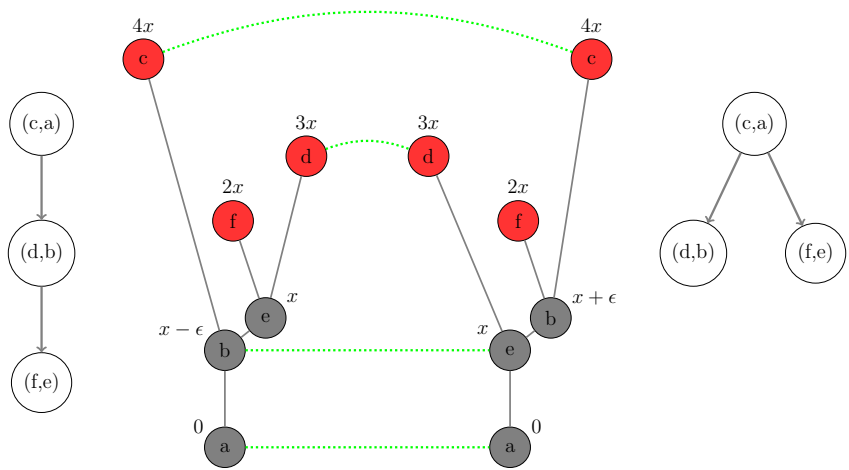
\begin{figure}
    \centering
    \resizebox{0.75\linewidth}{!}{
    \Large
    \begin{tikzpicture}[xscale=1,yscale=0.8]
    
    \begin{scope}[shift={(-3.5,0)}]
        \node[draw,circle,fill=gray!100,minimum width=1cm,label={[left=10pt]$0$}] at (0, 0) (root_1) {a};
        \node[draw,circle,fill=gray!100,minimum width=1cm,label={[left=10pt]$x-\epsilon$}] at (0, 3) (s1_1) {b};
        \node[draw,circle,fill=red!80,minimum width=1cm,label={[above]$4x$}] at (-2, 12) (m1_1) {c};
        \node[draw,circle,fill=red!80,minimum width=1cm,label={[above]$3x$}] at (2, 9) (m2_1) {d};
        \node[draw,circle,fill=gray!100,minimum width=1cm,label={[right=10pt]$x$}] at (1, 4) (s2_1) {e};
        \node[draw,circle,fill=red!80,minimum width=1cm,label={[above]$2x$}] at (0.2, 7) (m3_1) {f};
        \draw[gray,very thick] (root_1) -- (s1_1);
        \draw[gray,very thick] (s1_1) -- (m1_1);
        \draw[gray,very thick] (s1_1) -- (s2_1);
        \draw[gray,very thick] (s2_1) -- (m2_1);
        \draw[gray,very thick] (s2_1) -- (m3_1);
    \end{scope}
    
    \begin{scope}[shift={(3.5,0)}]
        \node[draw,circle,fill=gray!100,minimum width=1cm,label={[right=10pt]$0$}] at (0, 0) (root_2) {a};
        \node[draw,circle,fill=gray!100,minimum width=1cm,label={[left=10pt]$x$}] at (0, 3) (s1_2) {e};
        \node[draw,circle,fill=red!80,minimum width=1cm,label={[above]$3x$}] at (-2, 9) (m1_2) {d};
        \node[draw,circle,fill=red!80,minimum width=1cm,label={[above]$4x$}] at (2, 12) (m2_2) {c};
        \node[draw,circle,fill=gray!100,minimum width=1cm,label={[right=10pt]$x+\epsilon$}] at (1, 4) (s2_2) {b};
        \node[draw,circle,fill=red!80,minimum width=1cm,label={[above]$2x$}] at (0.2, 7) (m3_2) {f};
        \draw[gray,very thick] (root_2) -- (s1_2);
        \draw[gray,very thick] (s1_2) -- (s2_2);
        \draw[gray,very thick] (s1_2) -- (m1_2);
        \draw[gray,very thick] (s2_2) -- (m2_2);
        \draw[gray,very thick] (s2_2) -- (m3_2);
    \end{scope}
    
    \draw[ultra thick,-,dotted,green] (root_1) -- (root_2);
    \draw[ultra thick,-,dotted,green] (s1_1) -- (s1_2);
    \draw[ultra thick,-,dotted,green] (m2_1) to[bend left=25] (m1_2);
    \draw[ultra thick,-,dotted,green] (m1_1) to[bend left=25] (m2_2);

    \begin{scope}[shift={(-8,-2)}]
        \node[draw,circle] at (0, 12) (n1) {(c,a)};
        \node[draw,circle] at (0, 8) (n2) {(d,b)};
        \node[draw,circle] at (0, 4) (n3) {(f,e)};
        \draw[->,gray,ultra thick] (n1) -- (n2);
        \draw[->,gray,ultra thick] (n2) -- (n3);
    \end{scope}

    \begin{scope}[shift={(8,-2)}]
        \node[draw,circle] at (1.5, 12) (n1) {(c,a)};
        \node[draw,circle] at (0, 8) (n2) {(d,b)};
        \node[draw,circle] at (3, 8) (n3) {(f,e)};
        \draw[->,gray,ultra thick] (n1) -- (n2);
        \draw[->,gray,ultra thick] (n1) -- (n3);
    \end{scope}
    
    \end{tikzpicture}
    }
    \caption{Two merge trees (inner trees) together with their BDTs (outer trees) and an optimal mapping (green). They can be transformed into each other by a minimal vertex perturbation of $b$.}
    \label{fig:horizontal_instability}
\end{figure}

\medskip
\noindent
\textbf{Unconstrained edit distances.}
Lemma~\ref{lemma:changes} gives us two types of horizontal swaps possible through a minimal vertex perturbation.
Either we swap two saddle vertices (Figure~\ref{fig:horizontal_swap_types}a) or we split a saddle to create a new one, partitioning its children (Figure~\ref{fig:horizontal_swap_types}b).
We skip the inverse transformation due to symmetry.

The first case, the actual swap, can be expressed by a simple sequence of edit operations, for both the unconstrained deformation-based and classic edit distance. 
First, we delete the edge $(y,x)$, then we insert a new edge $(x,y)$. Lastly, we  have to fix the lengths of the surrounding edges (those connected to $x$ after the swap) by stretching or shortening operations.
None of the operations changes an edge length by more than $\epsilon$, which again gives us a bound of~$\deg(T_f)\cdot\epsilon$.
Since no edges are split, these are all classic operations, too.
The second case can be expressed in the same way, simply dropping the first operation.

\begin{figure}

\centering
\resizebox{\linewidth}{!}{
\begin{tikzpicture}

\begin{scope}[shift={(0,0)}]
 \node[draw=none,fill=none,circle] at (-6, -1) (dummy1) {};
 \node[draw=none,fill=none,circle] at (6, -1) (dummy2) {};
 \node[draw=none,fill=none,circle] at (-6, 4) (dummy3) {};
 \node[draw=none,fill=none,circle] at (6, 4) (dummy4) {};
 
 \begin{scope}[shift={(-3,0)}]
 
 \node[] at (0, 0) (p) {$p$};
 \node[] at (0, -1) (p_) {};
 \node[] at (0, 1) (x) {$x$};
 \node[] at (-1.5, 2) (y) {$y$};
 
 \node[draw,isosceles triangle,isosceles triangle apex angle=60,rotate=270,anchor=apex,minimum size=20pt]
         at (-2.2, 3) (c1t) {};
 \node[circle, fill=white]
         at (-2.2, 3) (c1) {$c_1$};
 \node[draw,isosceles triangle,isosceles triangle apex angle=60,rotate=270,anchor=apex,minimum size=20pt]
         at (-0.8, 3) (cit) {};
 \node[circle, fill=white]
         at (-0.8, 3) (ci) {$c_{i}$};
 \node[circle]
         at (-1.5, 3) (cd) {$\dots$};
         
 \node[draw,isosceles triangle,isosceles triangle apex angle=60,rotate=270,anchor=apex,minimum size=20pt]
         at (0.8, 3) (ciit) {};
 \node[circle, fill=white]
         at (0.8, 3) (cii) {$c_{i+1}$};
 \node[draw,isosceles triangle,isosceles triangle apex angle=60,rotate=270,anchor=apex,minimum size=20pt]
         at (2.2, 3) (c2t) {};
 \node[circle, fill=white]
         at (2.2, 3) (c2) {$c_{k}$};
 \node[circle]
         at (1.5, 3) (cdd) {$\dots$};
 
 \draw[decorate] (p) -- (p_);
 \draw[] (p) -- (x);
 \draw[] (x) -- (y);
 \draw[] (y) -- (c1);
 \draw[] (y) -- (ci);
 \draw[] (y) -- (cd);
 \draw[] (x) -- (cii);
 \draw[] (x) -- (c2);
 \draw[] (x) -- (cdd);

 \end{scope}

 \node[] at (0,1.5) (arr) {$\rightarrow$};
 \node[] at (0,0.5) (l) {\Large\textbf{(a)}};

 \begin{scope}[shift={(3,0)}]
 
 \node[] at (0, 0) (p) {$p$};
 \node[] at (0, -1) (p_) {};
 \node[] at (1.5, 2) (x) {$x$};
 \node[] at (0, 1) (y) {$y$};
 
 \node[draw,isosceles triangle,isosceles triangle apex angle=60,rotate=270,anchor=apex,minimum size=20pt]
         at (-2.2, 3) (c1t) {};
 \node[circle, fill=white]
         at (-2.2, 3) (c1) {$c_1$};
 \node[draw,isosceles triangle,isosceles triangle apex angle=60,rotate=270,anchor=apex,minimum size=20pt]
         at (-0.8, 3) (cit) {};
 \node[circle, fill=white]
         at (-0.8, 3) (ci) {$c_{j}$};
 \node[circle]
         at (-1.5, 3) (cd) {$\dots$};
         
 \node[draw,isosceles triangle,isosceles triangle apex angle=60,rotate=270,anchor=apex,minimum size=20pt]
         at (0.8, 3) (ciit) {};
 \node[circle, fill=white]
         at (0.8, 3) (cii) {$c_{j+1}$};
 \node[draw,isosceles triangle,isosceles triangle apex angle=60,rotate=270,anchor=apex,minimum size=20pt]
         at (2.2, 3) (c2t) {};
 \node[circle, fill=white]
         at (2.2, 3) (c2) {$c_{k}$};
 \node[circle]
         at (1.5, 3) (cdd) {$\dots$};
 
 \draw[decorate] (p) -- (p_);
 \draw[] (p) -- (y);
 \draw[] (x) -- (y);
 \draw[] (y) -- (c1);
 \draw[] (y) -- (ci);
 \draw[] (y) -- (cd);
 \draw[] (x) -- (cii);
 \draw[] (x) -- (c2);
 \draw[] (x) -- (cdd);

 \end{scope}
 \end{scope}

 \begin{scope}[shift={(10,0)}]
 \node[draw=none,fill=none,circle] at (-2.5, -1) (dummy1) {};
 \node[draw=none,fill=none,circle] at (5.5, -1) (dummy2) {};
 \node[draw=none,fill=none,circle] at (-2.5, 4) (dummy3) {};
 \node[draw=none,fill=none,circle] at (5.5, 4) (dummy4) {};

 \begin{scope}[shift={(-1,0)}]
 
 \node[] at (0, 0) (p) {$p$};
 \node[] at (0, -1) (p_) {};
 \node[] at (0, 1.5) (y) {$y$};
 \node[draw,isosceles triangle,isosceles triangle apex angle=60,rotate=270,anchor=apex,minimum size=20pt]
         at (-0.5, 3) (c1t) {};
 \node[circle, fill=white]
         at (-0.5, 3) (c1) {$c_1$};
 \node[draw,isosceles triangle,isosceles triangle apex angle=60,rotate=270,anchor=apex,minimum size=20pt]
         at (0.5, 3) (c2t) {};
 \node[circle, fill=white]
         at (0.5, 3) (c2) {$c_k$};
 \node[circle]
         at (0, 3) (cd) {$\dots$};
 
 \draw[decorate] (p) -- (p_);
 \draw[] (p) -- (y);
 \draw[] (y) -- (c1);
 \draw[] (y) -- (c2);
 \draw[] (y) -- (cd);

 \end{scope}

 \node[] at (0,1.2) (arr) {$\leftrightarrow$};
 \node[] at (0,0.5) (l) {\Large\textbf{(b)}};

 \begin{scope}[shift={(1.5,0)}]
 
 \node[] at (0, 0) (p) {$p$};
 \node[] at (0, -1) (p_) {};
 \node[] at (2.5, 2.1) (x) {$x$};
 \node[] at (0, 1.5) (y) {$y$};
 
 \node[draw,isosceles triangle,isosceles triangle apex angle=60,rotate=270,anchor=apex,minimum size=20pt]
         at (-0.7, 3) (c1t) {};
 \node[circle, fill=white]
         at (-0.7, 3) (c1) {$c_1$};
 \node[draw,isosceles triangle,isosceles triangle apex angle=60,rotate=270,anchor=apex,minimum size=20pt]
         at (0.7, 3) (cit) {};
 \node[circle, fill=white]
         at (0.7, 3) (ci) {$c_{i}$};
 \node[circle]
         at (0, 3) (cd) {$\dots$};
         
 \node[draw,isosceles triangle,isosceles triangle apex angle=60,rotate=270,anchor=apex,minimum size=20pt]
         at (1.8, 3) (ciit) {};
 \node[circle, fill=white]
         at (1.8, 3) (cii) {$c_{i+1}$};
 \node[draw,isosceles triangle,isosceles triangle apex angle=60,rotate=270,anchor=apex,minimum size=20pt]
         at (3.2, 3) (c2t) {};
 \node[circle, fill=white]
         at (3.2, 3) (c2) {$c_{k}$};
 \node[circle]
         at (2.5, 3) (cdd) {$\dots$};
 
 \draw[decorate] (p) -- (p_);
 \draw[] (p) -- (y);
 \draw[] (x) -- (y);
 \draw[] (y) -- (c1);
 \draw[] (y) -- (ci);
 \draw[] (y) -- (cd);
 \draw[] (x) -- (cii);
 \draw[] (x) -- (c2);
 \draw[] (x) -- (cdd);

 \end{scope}
 \end{scope}
    
\end{tikzpicture}
}
\caption{The two types of horizontal swaps. The second type (b) can happen in both directions.}
\label{fig:horizontal_swap_types}

\end{figure}
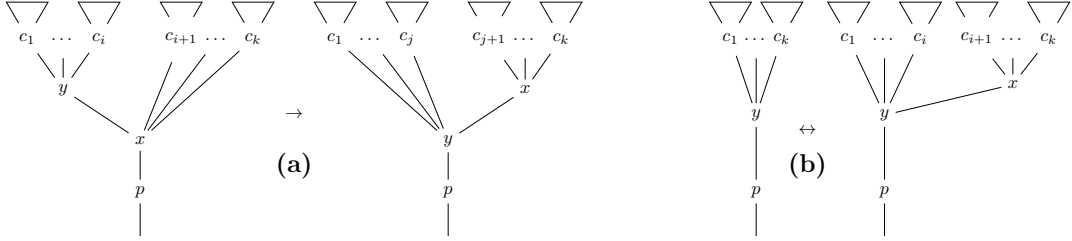

\subsection*{Stability Results}

Summing up, we can conclude the following results.
The unconstrained deformation-based edit distance is stable against any kind of minimal vertex perturbation.
\begin{theorem}
\label{theorem:stability_deform}
    Let $f,f'$ be two scalar fields that differ by a minimal vertex perturbation of extent $\epsilon$.
    Then the unconstrained deformation-based edit distance is bounded: $\delta_E \leq \deg(T_f)\cdot\epsilon$.
    In particular, there is an unconstrained deformation-based edit sequence $s = s_1 \dots s_k$ with $T_f \xrightarrow{\scriptscriptstyle s} T_{f'}$ such that $k \leq \deg(T_f)$ and $c(s_i) \leq \epsilon$ for all $1 \leq i \leq k$.
\end{theorem}
The path and branch mapping distance are stable against any minimal vertex perturbation that is not a horizontal swap.
\begin{theorem}
\label{theorem:stability_bdied}
    Let $f,f'$ be two scalar fields that differ by a minimal vertex perturbation of extent $\epsilon$ and of type simple change, edge split, or vertical swap.
    Then it holds that $\delta_P \leq \deg(T_f)\cdot\epsilon, \delta_B \leq \deg(T_f)\cdot\epsilon$.
    In particular, there is a one-degree deformation-based edit sequence $s = s_1 \dots s_k$ with $T_f \xrightarrow{\scriptscriptstyle s} T_{f'}$ such that $k \leq \deg(T_f)$ and $c(s_i) \leq \epsilon$ for all $1 \leq i \leq k$.
\end{theorem}
All BDT-based distances are stable against edge splits and simple changes.
\begin{theorem}
\label{theorem:stability_bdt}
    Let $f,f'$ be two scalar fields that differ by a minimal vertex perturbation of extent $\epsilon$ and of type simple change or edge split.  
    Then it holds that $\delta \leq \deg(T_f)\cdot\epsilon$ for all $\delta \in \{\delta_W,\delta_S,\delta_X\}$.
    In particular, for all three distances, there is an edit sequence $s = s_1 \dots s_k$ with $T_f \xrightarrow{\scriptscriptstyle s} T_{f'}$ such that $k \leq \deg(T_f)$ and $c(s_i) \leq \epsilon$ for all $1 \leq i \leq k$.
\end{theorem}
The classic one-degree edit distance is stable against vertical swaps and simple changes.
The unconstrained variant is also stable against horizontal swaps.
\begin{theorem}
\label{theorem:stability_constrclas}
    Let $f,f'$ be two scalar fields that differ by a minimal vertex perturbation of extent $\epsilon$ and of type simple change or edge split.
    Then it holds that $\delta_L \leq \deg(T_f)\cdot\epsilon$.
    In particular, there is a classic one-degree edit sequence $s = s_1 \dots s_k$ with $T_f \xrightarrow{\scriptscriptstyle s} T_{f'}$ such that $k \leq \deg(T_f)$ and $c(s_i) \leq \epsilon$ for all $1 \leq i \leq k$.
\end{theorem}
\begin{theorem}
\label{theorem:stability_genclas}
    Let $f,f'$ be two scalar fields that differ by a minimal vertex perturbation of extent $\epsilon$ and of type simple change, vertical swap or horizontal swap.
    Then it holds that $\delta_G \leq \deg(T_f)\cdot\epsilon$.
    In particular, there is a classic edit sequence $s = s_1 \dots s_k$ with $T_f \xrightarrow{\scriptscriptstyle s} T_{f'}$ such that $k \leq \deg(T_f)$ and $c(s_i) \leq \epsilon$ for all $1 \leq i \leq k$.
\end{theorem}
\begin{proof}[Proof of Theorem~\ref{theorem:stability_deform}-\ref{theorem:stability_genclas}]
    Arguments above and in Appendix~\ref{sec:perturbation_stability_app}.
\end{proof}

\noindent
Interestingly, these stability properties precisely reflect the performance of the current best algorithms.
Classic constrained tree edit distances, on bounded-degree trees (a reasonable assumption for merge trees in practice), can be computed in $\mathcal{O}(n^2)$ through a naive dynamic programming algorithm~\cite{treeEditSurvey} and have a tight lower bound~\cite{DBLP:journals/siamcomp/BackursI18} (under SETH) derived from a conditional lower bound for the string edit distance.
The one-degree edit distance can be computed in $\mathcal{O}(n^2)$ for arbitrary trees, using the fact that nodes need to be mapped to nodes of the same depth~\cite{DBLP:journals/ipl/Selkow77,DBLP:journals/tvcg/PontVDT22}.
Thus, for $\delta_W$ and $\delta_L$, we get $\Theta(n^2)$ in general (with $n$ the number of critical points), and for $\delta_S$, too, on most practical instances.
In terms of stability, they are orthogonal to each other.
The path and branch mapping distances have a naive dynamic programming algorithm in time $\mathcal{O}(n^4)$ for bounded-degree trees~\cite{wetzels2022branch,wetzels2022path} and the quadratic lower bound for strings applies here, too.
Thus, it remains open whether the naive approach is optimal for those distances as well.
These distances are strictly more stable than the classic tree edit distances with quadratic running time (both on merge trees themselves or BDTs).
The unconstrained deformation-based edit distance is the only distance stable against all kinds of perturbations and strictly more stable than all other distances. It is NP-complete and can be computed through a mixed integer linear program~\cite{taming}.
Figure~\ref{fig:hierarchy} illustrates these relationships.

\section{Finite/Local Stability of Edit Sequences}
\label{sec:sequence_stability}

Having shown stability against minimal vertex perturbation in various forms, we now consider arbitrary perturbations on the scalar field.
It turns out that the stability results from before lift to perturbations of multiple vertices.
Consider a scalar field $f: \mathbb{X} \rightarrow \mathbb{R}$ and a perturbed field $f': \mathbb{X} \rightarrow \mathbb{R}$ such that $||f-f'||_\infty \leq \epsilon$.
We show that the deformation-based edit distance can be bounded by the sizes of the trees (length of the edit sequence) and $\epsilon$ (maximum cost of the edit operations). 
Our main theorem in this section (Theorem~\ref{theorem:local_deform}) has already been shown in a more general form for a distance measure equivalent to the unconstrained deformation-based distance (on abstract merge trees) by Pegoraro~\cite{pegoraro2024finitelystableeditdistance,pegoraro2024functionaldatarepresentationmerge}.
There, this kind of stability was labeled \emph{finite} or \emph{local} stability.
We use our results on minimal vertex perturbations to derive finite stability of the unconstrained deformation-based distance and then show that such results can also be derived from the weaker stability results.

Instead of comparing $f$ and $f'$ with $||f-f'||_\infty \leq \epsilon$ directly,
we can split up the transformation from $f$ to $f'$ into minimal perturbations:
we construct a sequence of scalar functions $f_1,f_2,\dots, f_k$ with orderings $F_1,F_2,\dots,F_k$ such that $f=f_1$ and $f'=f_k$ and $f_i,f_{i+1}$ differ only by minimal vertex perturbations for all $1 \leq i < k$.
We also require the sequence to not perform unnecessary transpositions of vertices:
\begin{itemize}
    \item If $f(x)<f(y)$ and $f'(x)<f'(y)$, then $f_i(x)<f_i(y)$ for all $1 \leq i < k$.
    \item If $f(x)<f(y)$ and $f'(x)>f'(y)$, then there is an $i$ with $1 \leq i < k$ such that $f_j(x)<f_j(y)$ for all $1 \leq j < i$ and $f_j(x)>f_j(y)$ for all $i \leq j < k$.
\end{itemize}
Note that we can always derive such a sequence by simply perturbing the vertices one after another in a top-down fashion and splitting each single-vertex perturbation at the transpositions.
Furthermore, we can assume w.l.o.g.\ that $f'(x)>f(x)$ for all $x \in \mathbb{X}$, since for $f''(x) = f'(x)+\epsilon$ we have $\delta_E(T_{f'},T_{f''}) = 0$ by definition. 
In that case, the largest offset per vertex becomes $2\epsilon$.

By Theorem~\ref{theorem:stability_deform} we know that for each minimal vertex perturbation from $f_{i-1}$ to $f_{i}$ ($1 < i \leq k$) by $\epsilon_i$, there is a sequence of edit operations $s_i$ that transforms $T_{f_{i-1}}$ into $T_{f_{i}}$ such that the most expensive edit operation has cost $\epsilon_i$.
Since for all $i$ we have $\epsilon_i \leq \epsilon$, we can conclude that there is an edit sequence transforming $T_f$ into $T_{f'}$ for which every single operation has cost at most $\epsilon$.
However, this specific sequence can be exceptionally long, depending on the amount of minimal vertex perturbations (or swaps in the vertex ordering) necessary to perform the full perturbation.
In particular, the number of necessary minimal vertex perturbations can grow quadratically with the number of vertices (inverting a sequence of vertices needs a quadratic amount of swaps).
However, we can improve the corresponding sequence of edit operations and give an upper bound on its number of operations.

Note that the edit operations between two trees give us a clear node correspondence which is transitive.
Thus, we have a clear node mapping from $T_f$ to $T_{f'}$ induced by the edit sequence $s_2 \dots s_k$ derived from the minimal vertex perturbations.

For a node $n_1$ mapped to nodes $n_2,\dots,n_k$ ($n_i \in T_{f_i}$ for all $1 \leq i \leq k$), let $v_1,\dots,v_k \in \mathbb{X}$ be the corresponding vertices.
When performing a minimal vertex perturbation, the vertex corresponding to a merge tree node only changes if it is passed by a neighbor.
\begin{lemma}
\label{lemma:vertex_switch}
    We have $v_i \neq v_{i+1}$ if and only if $f_i(v_i) < f_i(v_{i+1})$ and $f_{i+1}(v_i) > f_{i+1}(v_{i+1})$ or $f_i(v_{i+1}) < f_i(v_{i})$ and $f_{i+1}(v_{i+1}) > f_{i+1}(v_{i})$.
\end{lemma}
\begin{proof}
    See case distinction in Appendix~\ref{sec:case_distinction}.
\end{proof}
This implies that the scalar value of the mapped vertices can only increase.
\begin{corollary}
\label{cor:increasing}
    We have $f_i(v_i) > f(v_1)$ for all $1 \leq i \leq k$.
\end{corollary}
From the Lemma~\ref{lemma:vertex_switch} and the construction of the sequence, we can also conclude that all nodes corresponding to $n_1$ have a bounded distance to $n_1$.
\begin{lemma}
\label{lemma:local_stability_2}
    Let $u_i$ be largest in $\{v_1 \dots v_i\}$ w.r.t.\ $f_i$, i.e.\ $u_i \in \{v_1 \dots v_i\}$ with $f_i(u_i) \geq f_i(v_j)$ for all $1 \leq j \leq i$. Then we have for all $1 \leq i \leq k$: $f'(u_i) \leq f(v_1) + 2\epsilon$ and if $v_i = u_i$ then $f(v_i) \leq f(v_1)$.
\end{lemma}
\begin{proof}
    Initially, the condition holds with $v_i = v_1 = u_i$.
    If $v_i = v_{i+1}$, then the condition is maintained trivially.
    Now if $v_i \neq v_{i+1}$, then either $v_i$ passes $v_{i+1}$ or $v_{i+1}$ passes $v_i$ from $f_i$ to $f_{i+1}$ by Lemma~\ref{lemma:vertex_switch}.
    
    If $f_{i+1}(v_i) > f_{i+1}(v_{i+1})$ ($v_i$ passes $v_{i+1}$) and $u_i=v_i$, then $u_{i+1} = u_i = v_i \neq v_{i+1}$, so both conditions are maintained.
    
    If $f_{i+1}(v_i) > f_{i+1}(v_{i+1})$ ($v_i$ passes $v_{i+1}$) and $u_i \neq v_i$, then $f_{i+1}(v_{i+1}) < f_{i+1}(v_i) < f_{i+1}(u_i)$, since $v_{i}$ only passes a single vertex in a minimal perturbation (and thereby only $v_{i+1}$, not $u_i$).
    Hence, $v_{i+1} \neq u_{i+1} = u_i$ and  both conditions are maintained.

    If $f_i(v_{i+1}) < f_i(v_i)$ ($v_{i+1}$ passes $v_i$) and $u_i=v_i$, then $f(v_{i+1}) < f(v_i)$ by construction of the sequence (no unnecessary swaps) and thus $f(v_{i+1}) < f(v_1)$ by induction hypothesis.
    With $f(v_{i+1}) < f(v_1)$, we can also conclude that $f'(u_{i+1}) = f'(v_{i+1}) < f(v_1)+2\epsilon$, so both conditions are maintained.
    
    If $f_i(v_{i+1}) < f_i(v_i)$ ($v_{i+1}$ passes $v_i$) and $u_i \neq v_i$, then $f_{i+1}(v_{i}) < f_{i+1}(v_{i+1}) < f_{i+1}(u_i)$, since $v_{i+1}$ only passes a single vertex in a minimal perturbation (and thereby only $v_{i}$, not $u_i$). Thus, we have $v_{i+1} \neq u_{i+1} = u_i$ and both conditions are maintained.
\end{proof}
By Lemma~\ref{lemma:local_stability_2}, we can conclude for the whole sequence that $|f(v_1)-f'(v_k)| < 2\epsilon$.
\begin{corollary}
\label{cor:local_stability_3}
    We have $|f(n_1) - f'(n_k)| \leq 2\epsilon$ for any nodes $n_1 \in V(T_f)$ mapped to $n_k \in V(T_{f'})$ in the mapping induced by the edit sequence $s_2 \dots s_k$.
\end{corollary}

Next, we consider unmapped nodes.
If a node $n_1 \in V(T_f)$ with parent $p_1 \in V(T_f)$ is not mapped in the induced mapping, then it has to be deleted between $f_i$ and $f_{i+1}$ for some $1 \leq i < k$.
For the sequence $s_2 \dots s_{i}$, let $n_2,\dots,n_{i}$ be the nodes mapped to $n_1$ and $p_2,\dots,p_{i}$ the  corresponding parents in $T_{f_2},\dots,T_{f_i}$.
Either $n_1$ is removed through the deletion of $(n_i,p_i)$, or it is pruned due to a deletion of a child edge.
We only have to consider the former case to cover all unmapped nodes (the pruned nodes are covered as parents of other unmapped nodes).
The case distinction in Appendix~\ref{sec:case_distinction} shows that the $i$-th perturbation can remove $n_i$ through a simple change (see Figure~\ref{fig:cases_sad}d), an edge split (see Figure~\ref{fig:cases_sad}e), or a horizontal swap (see Figures~\ref{fig:cases_sad}f, \ref{fig:cases_sad}g, and~\ref{fig:cases_sad}h), where $p_i$ passes $n_i$.
Since $p_i$ increases by at most $2\epsilon$, we have an upper bound on the length of $(n_i,p_i)$.
Next, we derive a bound on the edge $(n_1,p_1)$.
Since the vertices corresponding to $n_1$ only move up (Corollary~\ref{cor:increasing}), the perturbations on them can only make the edge longer, keeping the upper bound.
Thus, we show a lower bound on the vertices corresponding to $p_1,\dots,p_i$.
Note that they are not necessarily mapped to each other like $n_1,\dots,n_i$, since the parents of $n_1,\dots,n_i$ can be deleted or inserted.
\begin{lemma}
    Let $q_j $ be lowest in $\{p_j \dots p_i\}$, i.e.\ $q_j \in \{p_j \dots p_i\}$ with $f_j(q_j) \leq f_j(p_{j'})$ for all $j \leq j' \leq i$.
    Then we have for all $1 \leq j \leq i$: $f'(q_j)  \geq f'(n_i)$ and if $p_j=q_j$ then $f(p_j) \leq f(n_i)$.
\end{lemma}
\begin{proof}
    As explained above, we know that $f_{i}(p_i) \leq f_{i}(n_i)$ and $f_{i+1}(n_i) \leq f_{i+1}(p_i)$.
    By construction of the sequence (no unnecessary swaps), we also have $f'(p_i) \geq f'(n_i)$ and $f(p_i) \leq f(n_i)$.
    With $q_i = p_i$, the property to show holds initially.
    
    Going from $j$ to $j-1$, we consider three cases regarding the perturbation.
    Either $p_{j-1}$ and $p_j$ are the same node, i.e.\ they are mapped onto each other between $f_{j-1}$ and $f_j$ (not necessarily meaning the same vertex), or $p_{j-1}$ is deleted, or $p_j$ is inserted.
    Furthermore, for each main case, we either have $q_i = p_i$ or not.

    Note that if $p_{j-1}$ is deleted, then, as explained above, the vertex corresponding to $p_j$ passes the one corresponding to $p_{j-1}$. 
    If $p_j$ is inserted, its corresponding vertex needs to pass the one corresponding to $p_{j-1}$ (compare case distinction in Appendix~\ref{sec:case_distinction}, specifically the cases in Figures~\ref{fig:cases_max_reg}h and~\ref{fig:cases_sad}g) as well.
    If $p_{j-i}$ and $p_j$ are mapped, then either they are the same vertex, meaning the property is maintained, or one passes they other.
    Overall, we only have to consider two cases: either $p_j$ passes $p_{j-1}$ or $p_{j-1}$ passes $p_j$.
    
    First, assume that $q_j = p_j$.
    If the vertex corresponding to $p_{j-1}$ passes the one corresponding to $p_j$, then by construction of the sequence we have $f'(p_{j-1}) > f'(p_j) \geq f'(n_i)$ and $f(p_{j-1}) < f(p_j) \leq f(n_i)$.
    The property is maintained with $q_{j-1} = p_{j-1}$.
    If $p_j$ passes $p_{j-1}$, then the property is maintained with $q_{j-1} = p_j$.
    
    Second, assume that $q_j \neq p_j$.
    Since $f_{j-1}$ and $f_j$ differ by a minimal vertex perturbation swapping $p_{j-1}$ and $p_j$, both unequal to $q_j$, we have $f_{j-1}(q_j) < f_{j-1}(p_{j-1})$ and $f_{j-1}(q_j) < f_{j-1}(p_{j})$.
    Thus, with $q_{j-1} = q_j$, the property is maintained and $q_{j-1} \neq p_{j-1}$.
\end{proof}

For $j=1$ in particular, we get that $f(p_1) \geq f(q_1) \geq f_i(q_1) - 2\epsilon \geq f_i(n_i) - 2\epsilon$.
With Corollary~\ref{cor:increasing}, we get that $f(p_1) \geq f_i(n_i) - 2\epsilon \geq f(n_1) - 2\epsilon$.
Thus, we can conclude that the edge $(n_1,p_1)$ has length of at most $2\epsilon$ and can be removed with cost of at most $2\epsilon$.
The same arguments apply analogously for unmapped nodes $n_1 \in V(T_{f'})$ which are inserted.
\begin{lemma}
\label{lemma:local_stability_4}
  Let $n \in V(T_f)$ be a node in $T_f$ with parent $p \in T_f$ such that $n$ is mapped to $n' \in V(T_{f_i})$ by $s_2 \dots s_i$ and the unique edge $(n',p')\in E(T_{f_i})$ is deleted by $s_{i+1}$.
  Then we have $f(n)-f(p) \leq 2\epsilon$.
\end{lemma}
\begin{lemma}
\label{lemma:local_stability_5}
  Let $n \in V(T_{f'})$ be a node in $T_{f'}$ with parent $p \in T_{f'}$ such that $n$ is mapped to $n' \in V(T_{f_i})$ by $s_i \dots s_k$ and the unique edge $(n',p')\in E(T_{f_i})$ is deleted by $s_{i-1}$.
  Then we have $f(n)-f(p) \leq 2\epsilon$.
\end{lemma}

We can now construct a new edit sequence of shorter length using the previous lemmas.
We first delete all nodes of $T_f$ that are not present in the mapping through a contraction of their unique parent edge.
By Lemma~\ref{lemma:local_stability_4}, each deletion can be bound by $2\epsilon$.
We then relabel all edges in the remaining tree.
From the limit on the vertex shifts (Lemma~\ref{cor:local_stability_3}), we can also derive a limit of $2\epsilon$ for each relabel operation.
Lastly, we insert all nodes of $T_{f'}$ that are not present in the mapping with a limit of $2\epsilon$ per insertion (Lemma~\ref{lemma:local_stability_5}).

This sequence has a maximum length of $|V(T_f)|+|V(T_{f'})|$, since every vertex is edited at most once.
For intuition, note that we reduced the length of the sequence by removing redundant inserts and deletes for saddle swaps.
If a saddle vertex gets passed repeatedly by several other saddles, doing the full swaps after each other leads to the corresponding node being deleted in re-inserted each time.
By collapsing the tree structure to the common subtree once, we circumvent such a behavior.

\begin{theorem}
\label{theorem:local_deform}
    Let $f,f': \mathbb{X} \rightarrow \mathbb{R}$ be two scalar fields with $||f-f'||_\infty \leq \epsilon$.
    There is an unconstrained deformation-based edit sequence $s$ with $T_f \xrightarrow{\scriptscriptstyle s} T_{f'}$ such that $c(s) = \delta_E(T_f,T_{f'}) \leq (|V(T_f)|+|V(T_{f'})|)\cdot2\epsilon$.
    In particular, it holds that $c(o) \leq 2\epsilon$ for all $o \in s$.
\end{theorem}

The same generalization holds for the path mapping distance.
Assume that the sequence $f_1,\dots,f_k$ does not contain any horizontal swaps.
By Theorem~\ref{theorem:stability_bdied}, the sequence $s_2 \dots s_k$ is a valid one-degree sequence and the corresponding edit mapping a valid path mapping.
Thus, we know for any inner unmapped node $v$ with children $c_1,\dots,c_l$ that (reordering w.l.o.g.) there is at most one $1\leq i \leq l$ such that $T[c_i]$ is not completely unmapped.
The complete subtrees $T[c_2],\dots,T[c_{i-1}],T[c_{i+1}],\dots,T[c_l]$ can be deleted by one-degree operations, meaning the sequence of deletions at the beginning of the constructed sequence is a valid one-degree edit sequence.
The same holds for the insertions at the end of the constructed sequence.

\begin{theorem}
    Let $f,f': \mathbb{X} \rightarrow \mathbb{R}$ be two scalar fields with $||f-f'||_\infty \leq \epsilon$ such that they can be transformed into each other by a sequence of minimal vertex perturbations of type simple change, edge split or vertical swap.
    There is a one-degree deformation-based edit sequence $s$ with $T_f \xrightarrow{\scriptscriptstyle s} T_{f'}$ such that $c(s) = \delta_P(T_f,T_{f'}) \leq (|V(T_f)|+|V(T_{f'})|)\cdot2\epsilon$.
    In particular, it holds that $c(o) \leq 2\epsilon$ for all $o \in s$.
\end{theorem}

The same is true for the classic edit distances on merge trees, $\delta_L$ and $\delta_G$.
If no edge splits happen, then all necessary deletions can be done through classic edit operations.
If also no horizontal swaps happen, they only require one-degree edit operations, since the mapping induces a strongly connected subgraph.

\begin{theorem}
    Let $f,f': \mathbb{X} \rightarrow \mathbb{R}$ be two scalar fields with $||f-f'||_\infty \leq \epsilon$ such that they can be transformed into each other by a sequence of minimal vertex perturbations of type simple change, horizontal swap or vertical swap.
    There is a classic unconstrained edit sequence $s$ with $T_f \xrightarrow{\scriptscriptstyle s} T_{f'}$ such that $c(s) = \delta_G(T_f,T_{f'}) \leq (|V(T_f)|+|V(T_{f'})|)\cdot 2\epsilon$.
    In particular, it holds that $c(o) \leq 2\epsilon$ for all $o \in s$.
\end{theorem}

\begin{theorem}
    Let $f,f': \mathbb{X} \rightarrow \mathbb{R}$ be two scalar fields with $||f-f'||_\infty \leq \epsilon$ such that they can be transformed into each other by a sequence of minimal vertex perturbations of type simple change or vertical swap.
    There is a classic one-degree edit sequence $s$ with $T_f \xrightarrow{\scriptscriptstyle s} T_{f'}$ such that $c(s) = \delta_L(T_f,T_{f'}) \leq (|V(T_f)|+|V(T_{f'})|)\cdot 2\epsilon)$.
    In particular, it holds that $c(o) \leq 2\epsilon$ for all $o \in s$.
\end{theorem}

For the BDT-based distances $\delta_X$ and $\delta_W$, we can also conclude finite stability against edge splits.
We are given a sequence $s_2 \dots s_k$ of edit operations between $\bdt(T_f)$ and $\bdt(T_{f'})$.
Since it is a one-degree sequence, the induced mapping represents a strongly connected common subgraph.
Thus, we can again construct a sequence of deletions before relabels before insertions that is a valid one-degree sequence.
Furthermore, the mapping between the BDTs induces a mapping between the vertices of the merge trees.
Thus, Corollary~\ref{cor:local_stability_3} also holds for the start and end vertex of any mapped branch, limiting the costs of each relabel.
Complete subtrees in the BDT correspond to complete subtrees in the original merge tree.
Since deleting each edge individually can only be more expensive than deleting a whole branch in one, we can also apply the bound on deletions from Lemma~\ref{lemma:local_stability_4}.

\begin{theorem}
    Let $f,f': \mathbb{X} \rightarrow \mathbb{R}$ be two scalar fields with $||f-f'||_\infty \leq \epsilon$ such that they can be transformed into each other by a sequence of minimal vertex perturbations of type simple change or edge split.
    There is a classic one-degree edit sequence $s$ with $\bdt(T_f) \xrightarrow{\scriptscriptstyle s} \bdt(T_{f'})$ such that $c(s) = \delta_W(\bdt(T_f),\bdt(T_{f'})) \leq (|V(\bdt(T_f))|+|V(\bdt(T_{f'}))|)\cdot 2\epsilon)$.
    In particular, it holds that $c(o) \leq 2\epsilon$ for all $o \in s$.
\end{theorem}

\begin{theorem}
    Let $f,f': \mathbb{X} \rightarrow \mathbb{R}$ be two scalar fields with $||f-f'||_\infty \leq \epsilon$ such that they can be transformed into each other by a sequence of minimal vertex perturbations of type simple change or edge split.
    There is a classic one-degree edit sequence $s$ with $\bdt(T_f) \xrightarrow{\scriptscriptstyle s} \bdt(T_{f'})$ such that $c(s) = \delta_X(\bdt(T_f),\bdt(T_{f'})) \leq (|V(\bdt(T_f))|+|V(\bdt(T_{f'}))|)\cdot 2\epsilon)$.
    In particular, it holds that $c(o) \leq 2\epsilon$ for all $o \in s$.
\end{theorem}

\section{Implications and Future Work}
\label{sec:conclusion}

The results of the previous sections provide a clean hierarchy of merge tree edit distances by stability.
We now present future work and open questions induced by these findings.
Furthermore, we discuss slight oversimplifications in the hierarchy presented in Figure~\ref{fig:hierarchy} and give a more detailed view.

\begin{figure}[t]
    \centering
    \resizebox{0.7\linewidth}{!}{
    \Large
    \begin{tikzpicture}[xscale=1,yscale=1,
    box1/.style = {draw,red,ultra thick,inner sep=0.2pt,rounded corners=12.5pt},
    box2/.style = {draw,red,ultra thick,inner sep=0.2pt,rounded corners=12.5pt},
    box3/.style = {draw,red,ultra thick,inner sep=8pt,rounded corners=16pt}]
    
    \node[circle] at (0, 0) (sc) {SC};
    \node[circle] at (-3, 3) (vs) {VS};
    \node[circle] at (3, 3) (es) {ES};
    \node[circle] at (0, 6) (hs) {\large UHS};
    \node[circle] at (1.9, 4.1) (ohs) {\large OHS};
    
    \node[cyan] at (1.5, 1.5) (bdt) {\small $\mathbff{\theta(n^2)}$};
    \node[blue] at (1.6, 3) (bdt) {\small $\mathbff{\theta(n^2)}$};
    \node[green!70!black] at (-1.5, 1.5) (clas) {\small $\mathbff{\theta(n^2)}$};
    \node[red] at (-0.8, 2.6) (bdi) {\small $\mathbff{\mathcal{O}(n^4)}$};
    \node[gray] at (-1, 4.5) (bdi) {\small\bfseries NP-c};
    \node[black] at (1.3, 5.2) (bdi) {\small\bfseries NP-c};

    \node[fill=cyan,label={[right,yshift=-4pt,xshift=3pt]OBDT-based}] at (4.5, 4.1+0.2) (obdt_l) {};
    \node[fill=blue!80,label={[right,yshift=-4pt,xshift=3pt]BDT-based}] at (4.5, 3.6+0.2) (bdt_l) {};
    \node[fill=green!80!black,label={[right,yshift=-4pt,xshift=3pt]Classic constrained}] at (4.5, 3.1+0.2) (clas_l) {};
    \node[fill=gray,label={[right,yshift=-4pt,xshift=3pt]Classic unconstrained}] at (4.5, 2.6+0.2) (uclas_l) {};
    \node[fill=red,label={[right,yshift=-4pt,xshift=3pt]Branch/Path Mapping}] at (4.5, 2.1+0.2) (bdi_l) {};\node[fill=black,label={[right,yshift=-4pt,xshift=3pt]Deformation-based}] at (4.5, 1.6+0.2) (bdt_l) {};

    \node[box1,green!80!black,rotate fit=45,fit=(sc)(vs)] {};
    \node[box2,cyan,rotate fit=45,fit=(sc)(es)] {};
    \draw [blue!80,ultra thick] ($ (sc) + (-0.51,0.19) $) arc(160:315:0.55) -- ($ (es) + (0.39,-0.39) $) arc(0
    -45:45:0.55) -- ($ (ohs) + (0.39,0.39) $) arc(45:155:0.55) -- cycle;
    \node[box3,black,rotate fit=45,fit=(sc)(es)(hs)] {};
    \draw [gray,ultra thick] ($ (hs) + (0.59,0) $) arc(0:135:0.59) -- ($ (vs) + (-0.42,0.42) $) arc(135:225:0.59) -- ($ (sc) + (-0.42,-0.42) $) arc(225:360:0.59) -- cycle;
    \draw [red,ultra thick] ($ (sc) + (-0.42,-0.42) $) arc(225:315:0.59) --($ (es) + (0.42,-0.42) $) arc(-45:90:0.59) -- ($ (vs) + (0,0.59) $) arc(90:225:0.59) -- cycle;
    
    \end{tikzpicture}
    }
    \caption{Illustrative visualization of the stability-based hierarchy of the different edit distances annotated with their complexity.}
    \label{fig:hierarchy2}
\end{figure}
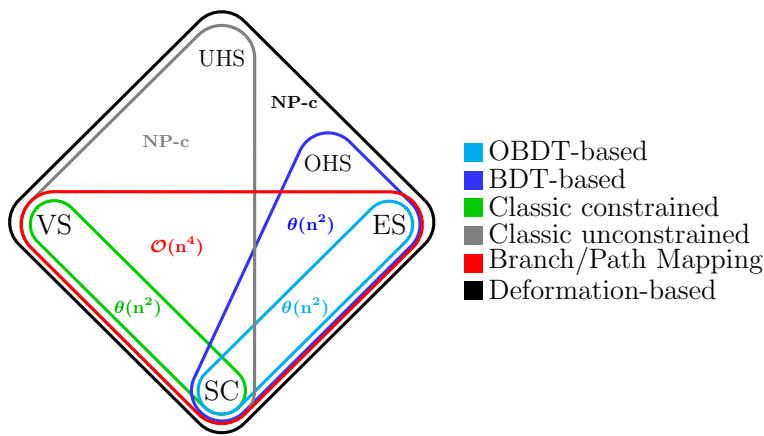

\textbf{Types of Horizontal Swaps.}
As explained in the definition of perturbation types (Section~\ref{sec:transformations}), horizontal swaps come in two forms: ordered and unordered ones.
For a constrained edit distance defined/computed on unordered BDTs (like the merge tree Wasserstein distance), ordered horizontal swaps fall into the category of changes that are easy to handle.
In fact, the two unordered BDTs are equal in the case of an ordered horizontal swap.
First of all, this makes the merge tree Wasserstein distance strictly more powerful than the Saikia distance.
In particular, the handling of ordered horizontal swaps is precisely the difference between the two distances, which are otherwise equivalent.
Furthermore, this also allows the merge tree Wasserstein distance to handle some instabilities which the path or branch mapping distance can not (and also none of the classic constrained edit distances): a constrained edit distance defined on the merge tree itself will never be able to handle saddle swaps.
Thus, the two classes of distances are actually orthogonal to each other, which has already been stated in~\cite{wetzels2022branch}.
This is visualized in Figure~\ref{fig:hierarchy2}.
It remains open whether there are techniques that allow for the definition of an edit distance that is stable against ordered horizontal swaps in addition to being stable against edge splits and vertical swaps.

\textbf{Asymmetry of Hierarchy.}
We considered five different distances based on the classic constrained edit distance: three on BDTs, two on merge trees directly.
On merge trees themselves, we looked at two adapted forms of the contour tree alignment distance, one constrained and one unconstrained.
In Figure~\ref{fig:hierarchy}, the constrained one looks symmetric to the constrained BDT-based distances in terms of stability.
As shown above, this is already a misrepresentation, since distances on unordered BDTs can, in fact, handle some horizontal swaps, which constrained distances on merge trees can not.
It remains open whether there are techniques that allow for the definition of an edit distance that is stable against ordered horizontal swaps and vertical swaps, but not edge splits.

Furthermore, we considered an unconstrained classic edit distance on merge trees directly ($\delta_G$), which is then able to handle all forms of horizontal swaps.
This is also shown in Figure~\ref{fig:hierarchy2}.
In contrast, such an unconstrained edit distance on BDTs seems impossible to define in a reasonable way, due to the nature of the branch hierarchy:
how to interpret the deletion of an inner branch (specifically the root) is highly unintuitive.
It remains open whether there are techniques that allow for the definition of an edit distance that is stable against both kinds of horizontal swaps and edge splits, but not against vertical swaps.

\textbf{Other Future Work.}
Apart from the extended hierarchy (Figure~\ref{fig:hierarchy2}), the initial hierarchy already gives rise to some interesting open questions that should be considered in future work.
Based on the quadratic gap between current lower and upper bounds for the path and branch mapping distances, we can ask the question of whether there are more efficient algorithms or if a quartic lower bound can be proven.
The same can be asked for heuristics, approximations or other distances of the same expressiveness.
More generally, one could study whether the increased complexity is inherent for certain stability properties (oblivious to a specific distance used).
A possible approach for such a result could be to prove a (conditional) lower bound for certain stability properties under mild assumptions about the structure of the distance.

Furthermore, it remains open whether more general stability properties can be shown, e.g.\ typical stability against the $L_\infty$-norm under mild assumptions (in the vein of the stability proof for the Wasserstein distance~\cite{edelsbrunner09}).

\begin{appendices}

\section{Case Distinction for Lemma~\ref{lemma:changes}}
\label{sec:case_distinction}

In this section, we provide the proof for Lemma~\ref{lemma:changes} through a case distinction on the possible changes in an augmented merge tree induced by a minimal vertex perturbation.
If we have a minimal vertex perturbations increasing the scalar value of vertex $x$ between $f$ and $f'$, we have to consider single transpositions as well as a perturbation without a change in the vertex ordering.
Both can induce meaningful changes in the augmented or unaugmented merge tree: while changes without transpositions do only change labels, transpositions can also change the actual tree structure.
Observations~\ref{obs:no_edge}-\ref{obs:children_y} define the way in which the tree structure can change: If vertex $x$ passes vertex $y$, we only have to consider which children of $y$ become children of $x$.
Furthermore, the tree structure only changes if the two vertices are connected (Observation~\ref{obs:no_edge}).
Hence, we do a case distinction over whether there is an edge between the involved vertices, their node types (defining the numbers of children in $T_f$), and the number of edges that are redirected.
We begin with the non-transposing perturbations, for which we also consider the critical types of $x$ separately.
Next, we consider perturbations with transpositions of non-neighboring vertices per type of shifted vertex.
Lastly, we consider transpositions of neighbors for all types of node pairs.

\paragraph*{Regular vertex moves.}
If we perturb a regular vertex without changes in the vertex order, the two augmented merge trees $\mathcal{T}_f$ and $\mathcal{T}_{f'}$ are fully isomorphic and likewise are $T_f$ and $T_{f'}$.
The branch decomposition cannot change either.
So we have $T_f = T_{f'}$ and $\bdt(T_f) = \bdt(T_{f'})$, a simple change.
More specifically, no critical values change and the labeling of $T_f$ and $T_{f'}$ are equal, too, so any meaningful distance should be $0$.

\paragraph*{Minimum vertex moves.}
If we perturb the global minimum without changes in the vertex order, the two augmented merge trees $\mathcal{T}_f$ and $\mathcal{T}_{f'}$ are fully isomorphic and likewise are $T_f$ and $T_{f'}$.
The branch decomposition cannot change either.
So we have $T_f = T_{f'}$ and $\bdt(T_f) = \bdt(T_{f'})$ and only a single edge or branch (the one containing the root) changes its label.
This is a simple change.

\paragraph*{Saddle vertex moves.}
If we perturb a saddle vertex without changes in the vertex order, we again have $T_f = T_{f'}$ and $\bdt(T_f) = \bdt(T_{f'})$.
Only one node in the tree changes it scalar value, but multiple edge or branch labels can change.
We have a simple change.

\paragraph*{Maximum vertex moves.}
If we perturb a maximum vertex without changes in the vertex order, we again have $T_f = T_{f'}$ and $\bdt(T_f) = \bdt(T_{f'})$.
Only a single edge or branch (containing the moved maximum) changes its label.
We again get a simple change.

\paragraph*{Regular vertex passes non-neighbor vertex.}
If we perturb a regular vertex $x$ such that it passes another vertex $y$ with $(y,x) \notin E(\mathcal{T}_{f})$, $\mathcal{T}_f$ and $\mathcal{T}_{f'}$ are fully isomorphic (Observation~\ref{obs:no_edge}), as are $T_f$ and $T_{f'}$.
Since no critical values change, the labeling of $T_f$ and $T_{f'}$ are equal, too.
Therefore, the branch decomposition cannot change either.
So we have a simple change with $T_f = T_{f'}$ and $\bdt(T_f) = \bdt(T_{f'})$ and additionally any meaningful distance should be $0$.

\paragraph*{Minimum vertex passes non-neighbor vertex.}
This case can never happen in a $d$-manifold with $d>1$, since the global minimum is always connected to the next vertex in the ordering.

\paragraph*{Saddle vertex passes non-neighbor vertex.}
If we perturb a saddle vertex $x$ such that it passes another vertex $y$ with $(y,x) \notin E(\mathcal{T}_{f})$, the two merge trees $T_f$ and $T_{f'}$ remain equal structurally (Observation~\ref{obs:no_edge}). 
Since all maxima stay equal the persistence-based branch decompositions $\bdt(T_f),\bdt(T_{f'})$ also do not change.
We get a simple change with $T_f = T_{f'}$ and $\bdt(T_f) = \bdt(T_{f'})$.

\paragraph*{Maximum vertex passes non-neighbor vertex.}
If we perturb a maximum vertex $x$ such that it passes another vertex $y$ with $(y,x) \notin E(\mathcal{T}_{f})$, the two merge trees $T_f$ and $T_{f'}$ remain equal structurally (Observation~\ref{obs:no_edge}).
However, this does not necessarily hold for the BDTs.

If $y$ is a maximum node with $\mathcal{B}_f(y)$ being an ancestor of $\mathcal{B}_f(x)$ in $\bdt(T_f)$, then $\mathcal{B}_f(y)$ is the parent of $\mathcal{B}_f(x)$. Suppose otherwise, then there is another maximum $z$ with $\mathcal{B}_f(z)$ being an ancestor of $\mathcal{B}_f(x)$ and a descendant of $\mathcal{B}_f(y)$, i.e.\ 
by construction of the persistence-based branch decomposition
$f(x) < f(z) < f(y)$, which contradicts the assumption that $x$ and $y$ are neighbors in the total ordering. 
The two branches swap, i.e.\ $\mathcal{B}_{f'}(x)$ is the parent of $\mathcal{B}_{f'}(y)$ in $\bdt(\mathcal{T}_{f'})$.
This case is illustrated in Figure~\ref{fig:cases_max_reg}a.

If $y$ is a maximum with $\mathcal{B}_f(x)$ and $\mathcal{B}_f(y)$ not in an ancestor relation, then neither $E(BDT(T_f))$ changes (the two branches do not originate in each other), nor the sibling ordering (only depends on the start vertices of the branches).
In this case, we get a simple change with $T_f=T_{f'}$ and $BDT(T_f) = BDT(T_{f'})$.

If $y$ is not a maximum, the nesting of branches stays exactly the same, so again we get a simple change with $T_f=T_{f'}$ and $BDT(T_f) = BDT(T_{f'})$.

\paragraph*{Minimum vertex passes neighbor.}
Consider a perturbation of the global minimum $x$ such that it passes another vertex $y$ with $(y,x) \in E(\mathcal{T}_{f})$.
Since we restrict to $d$-manifolds with $d>1$, we can conclude that all children of $y$ become children of $x$ (otherwise, the new global minimum would become a saddle).
In this case, we simply swap the two vertices, which gives us a simple change with $T_f=T_{f'}$ and $BDT(T_f) = BDT(T_{f'})$, up to a renaming of a single vertex ($y$ to $x$).

\paragraph*{Maximum vertex passes neighbor.}
This case can never happen, since a maximum does not have neighbors of higher scalar value.

\paragraph*{Regular vertex passes regular neighbor.}
If we perturb a regular vertex $x$ such that it passes another regular vertex $y$ with $(y,x) \in E(\mathcal{T}_{f})$, we have to consider two cases.
Either the unique child $c$ of $y$ stays a child of $y$ or it becomes a child of $x$ (Observation~\ref{obs:children_y}).

If $(c,y) \in E(\mathcal{T}_{f'})$, then $x$ becomes a maximum node and a child of $y$, see Figure~\ref{fig:cases_max_reg}b.
In this case, y becomes a saddle node.
Thus, the merge tree changes structurally: we have $x,y \notin V(T_{f})$ but $x,y \in V(T_{f'})$.
The edge in $T_f$ that contains $x$ and $y$ is split by the new saddle $y$ in $T_{f'}$ and a new edge $(x,y)$ is created.
Since $x$ is the closest neighbor of $y$ in $\mathcal{T}_{f'}$, the branch $yx$ becomes a new leaf in $\bdt(T_{f'})$.
The relation of all other branches does not change, meaning that $\bdt(T_f) \subset \bdt(T_{f'})$ with $V(\bdt(T_f)) = V(\bdt(T_{f'})) \cup \{yx\}$.
This is an edge split.

In the second case that $(c,x) \in E(\mathcal{T}_{f'})$ the only change is that $x$ and $y$ got swapped in the augmented merge trees, see Figure~\ref{fig:cases_max_reg}c.
The unaugmented merge tree remains unchanged.
$T_f$ and $T_{f'}$ are fully isomorphic, even in terms of scalar values, as are $\bdt(T_f)$ and $\bdt(T_{f'})$.
Again, any meaningful distance should be $0$ and we have a simple change.

\paragraph*{Regular vertex passes maximum neighbor.}
If we perturb a regular vertex $x$ such that it passes a maximum vertex $y$ with $(y,x) \in E(\mathcal{T}_{f})$, i.e.\ its unique child, then $x$ becomes a new maximum while $y$ becomes regular, see Figure~\ref{fig:cases_max_reg}d.
In essence, the maximum corresponding to $y$ is simply shifted, ignoring renaming to $x$.
The two merge trees $T_f$ and $T_{f'}$ remain equal structurally.
Since the maximum does not pass another critical point, the branch decompositions do not change either.
We get a simple change with $T_f = T_{f'}$ and $\bdt(T_f) = \bdt(T_{f'})$.
Note that this only holds up to isomorphism, since we rename the maximum node ($y$ to $x$).

\paragraph*{Regular vertex passes saddle neighbor.}
If we perturb a regular vertex $x$ such that it passes a saddle vertex $y$ with $(y,x) \in E(\mathcal{T}_{f})$, we have to consider four cases.
We make a case distinction over the number of children of $y$ that become children of $x$ (Observation~\ref{obs:children_y}).

If all children of $y$ become children of $x$, we again get a simple swap of the two vertices in the augmented merge tree: $x$ becomes a saddle, $y$ becomes a regular vertex, see Figure~\ref{fig:cases_max_reg}e.
In the unaugmented merge tree, no structural changes happen.
Only the saddle node corresponding to $y$ (and now to $x$) is shifted, also the branch decomposition remains unchanged.
So we get a simple change with $T_f = T_{f'}$ and $\bdt(T_f) = \bdt(T_{f'})$ (up to renaming of $x$ and $y$).

If no children of $y$ become children of $x$, $x$ becomes a new maximum node and thus a child of $y$ in $T_{f'}$.
The edge $(x,y)$ is created, otherwise the merge tree stays the same, see Figure~\ref{fig:cases_max_reg}f.
Note that the branch decomposition is also not changed other than a new persistence pair being added at the lowest level, as $x$ becomes the closest child of $y$.
Since in both the merge tree and the BDT, we only add a leaf, we have a simple change.

Next, we consider the case that exactly one child of $y$ becomes a child of $x$.
In this case, $x$ stays a regular vertex but moves into one of the superlevel set components of $y$, see Figure~\ref{fig:cases_max_reg}g.
The unaugmented merge tree remains unchanged, since no critical point is moved.
$T_f$ and $T_{f'}$ are fully isomorphic, even in terms of scalar values.
The same holds for the BDTs, so we have a simple change with $T_f = T_{f'}$ and $\bdt(T_f) = \bdt(T_{f'})$.
Again, any meaningful distance should be $0$.

In the last case, at least two but not all children of $y$ become children of $x$.
Here, $x$ becomes a new saddle node in $T_{f'}$ and $y$ stays a saddle with $(x,y) \in E(T_{f'})$, see Figure~\ref{fig:cases_max_reg}h.
Thus, neither $T_{f'} \subseteq T_f$ nor $T_{f} \subseteq T_{f'}$ holds.
Since $x\notin V(T_f)$, we have exactly the second case for horizontal swaps from Lemma~\ref{lemma:changes}.
Depending on the branch priorities, the branch decomposition can change only in terms of the ordering (if the main branch goes along $x$ in $T_{f'}$, the previous incomparable children of $y$ become ordered in $<_B$) or structurally (if the main branch of $T_{f'}$ does not go along $x$, some branches along $x$ are no longer children of the main branch).
The ordered branch decomposition always changes.
Thus, in both cases, we have a horizontal branch swap.

\paragraph*{Saddle vertex passes regular neighbor.}
If we perturb a saddle vertex $x$ such that it passes a regular vertex $y$ with $(y,x) \in E(\mathcal{T}_{f})$, we get three cases.
This time, we make a case distinction over whether the unique child $c_1$ of $y$ becomes a child of $x$ (Observation~\ref{obs:children_y}) and the number of children of $x$.
Let $c_2 \dots c_k$ be the other children of $x$ in $\mathcal{T}_f$.

If we have $(c_1,y) \in E(\mathcal{T}_{f'})$ and $x$ has more than one child in $\mathcal{T}_{f'}$ ($k>2$), then $x$ stays a saddle and $y$ becomes a saddle, see Figure~\ref{fig:cases_sad}a.
Since $(c_1',x) \in E(T_f)$ but $(c_1',y) \in T_{f'}$ for some $c_1'$ in the component of $c_1$, neither $T_{f'} \subseteq T_f$ nor $T_{f} \subseteq T_{f'}$ holds.
We get the second case for horizontal swaps from Lemma~\ref{lemma:changes}, up to renaming the remaining saddle.
Depending on the branch priorities, the branch decomposition can change structurally (if the main branch goes along $c_1$, all branches along $c_2\dots c_k$ are children of the main branch in $T_f$, but not in $T_{f'}$) or only in terms of ordering (if the main branch goes along $c_3$, the branches along $c_1$ and $c_2$ are incomparable in $T_f$ but strictly ordered in $T_{f'}$).
In both cases, we have a horizontal branch swap.

If we have $(c_1,y) \in E(\mathcal{T}_{f'})$ but $x$ has only one child in $\mathcal{T}_{f'}$ ($k=2$), then $x$ becomes regular and $y$ becomes a saddle, see Figure~\ref{fig:cases_sad}b.
In the unaugmented merge tree, the saddle corresponding to $x$ (and now to $y$) is shifted.
Everything else stays the same, including the BDTs.
Thus, we get a simple change with $T_f = T_{f'}$ and $\bdt(T_f) = \bdt(T_{f'})$ (up to renaming $x$ to $y$).

If we have $(c_1,x) \in E(\mathcal{T}_{f'})$, then $x$ stays a saddle and $y$ stays regular.
We simply move $y$ out of one of the superlevel set components of $x$, but their connectivity stays the same otherwise, see Figure~\ref{fig:cases_sad}c.
Thus, the two unaugmented merge trees are structurally equal and the branch decomposition stays the same, too.
We get a simple change with $T_f = T_{f'}$ and $\bdt(T_f) = \bdt(T_{f'})$.

\paragraph*{Saddle vertex passes maximum neighbor.}
If we perturb a saddle vertex $x$ such that it passes a maximum vertex $y$ with $(y,x) \in E(\mathcal{T}_{f})$, we consider two cases: either $x$ has exactly one other child or at least two.
In both cases, $y$ becomes a regular vertex below $x$ (Observations~\ref{obs:children_y} and~\ref{obs:flipped_edge}).

If at least two children of $x$ remain, $x$ stays a saddle vertex and only one leaf node is removed from the unaugmented merge tree (Figure~\ref{fig:cases_sad}d), i.e.\ $T_{f'} \subset T_f$.
Since $y$ was the closest child of $x$, the branch decomposition tree stays the same as well, other than removing the leaf corresponding to $y$, i.e.\ $\bdt(T_{f'}) \subset \bdt(T_f)$.
This is a simple change.

If only one child $c$ of $x$ remains, $x$ becomes a regular vertex, too (Figure~\ref{fig:cases_sad}e).
In this case, $y$ and $x$ are both removed from the unaugmented merge tree.
The edges $(c,x)$ and $(x,p)$ are merged into the edge $(c,p)$.
Thus, we have neither $T_{f} \subseteq T_{f'}$ nor $T_{f'} \subseteq T_f$.
Since $y$ was the closest child to $x$, it was a leaf in $\bdt(T_f)$, so we get $\bdt(T_{f'}) \subset \bdt(T_f)$ and thereby an edge split.

\paragraph*{Saddle vertex passes saddle neighbor.}
If we perturb a saddle vertex $x$ such that it passes a saddle vertex $y$ with $(y,x) \in E(\mathcal{T}_{f})$, we get three cases.
We do a case distinction over the number of children of $y$ (Observation~\ref{obs:children_y}) that become children of $x$ and the number of children of $x$.

If all children of $y$ become children of $x$, $y$ becomes a regular vertex below $x$ (Observations~\ref{obs:flipped_edge} and~\ref{obs:children_y}, see Figure~\ref{fig:cases_sad}f) and is removed from $T_{f'}$.
We have the inverse case of the second one from Lemma~\ref{lemma:changes}, up to renaming of $x$ and $y$.
Again, either the partial ordering of the branches changes (if the main branch goes along $y$ in $T_f$, the children of $x$ and $y$ become incomparable), or the actual branch nesting (if the main branch of $T_f$ does not go along $y$, the branches along children of $y$ become children of the main branch). 
We get at least an ordered horizontal branch swap.

If at least one child $c$ of $y$ stays a child of $y$ and $x$ has more than one child in $\mathcal{T}_ {f'}$, then both $x$ and $y$ stay saddles and swap their position, see Figure~\ref{fig:cases_sad}g.
We get exactly the first case for horizontal swaps from Lemma~\ref{lemma:changes}.
The tree structures are not in a subgraph relation.
Again, depending on the main branch, either the partial ordering changes (the children of $x$ pass the children of $y$) or the actual nesting of branches.
So we get a horizontal branch swap.
The same holds for the case where all children of $y$ stay children of $y$ and $x$ has more than one child in $\mathcal{T}_ {f}$/$\mathcal{T}_ {f'}$.

If all children of $y$ stay children of $y$ and $x$ has exactly one child in $\mathcal{T}_ {f'}$, then $y$ becomes a saddle and $x$ becomes regular, see Figure~\ref{fig:cases_sad}h.
This means that in $T_{f'}$ only one saddle remains and all subtrees of $x$ and $y$ are now connected to this saddle.
Ignoring renaming, we get the exact same situation as in Figure~\ref{fig:cases_sad}f, so a horizontal branch swap.

\section{Details on Further Distances}
\label{sec:other_dists}

In this section, we describe further details on distances that have been skipped in the main paper.
In particular, we provide remaining definitions and argue why certain adaptations have been made in specific cases.

\subsection{Saikia Distance}

We now consider $\delta_X$ in more detail.
The original paper by Saikia, Seidel and Weinkauf introduced a distance function on so-called extended branch decomposition graphs.
They are generalizations of BDTs with additional nodes representing the main branches of subtrees.
The resulting structure is a directed acyclic graph as opposed to an actual tree.
However, the defined distance measure can be applied to both types of structures equivalently, due to the nature of algorithms for constrained edit distances.

When applying their method to trees, it becomes the one-degree edit distance.
Extended branch decomposition graphs can be reduced to \emph{ordered} BDTs by removing additional nodes and keeping the orderings on neighbors.
Thus, we define $\delta_X$ as the classic one-degree edit distance on ordered BDTs.
The label functions considered in the original paper are persistence of branches or volume of branches, together with the absolute difference as cost function and $0$ as the empty symbol.
Of course, one could also use the Wasserstein metric on branches to obtain the ordered variant of $\delta_W$.

When using the Wasserstein metric or persistence, all arguments for $\delta_W$ from Section~\ref{sec:perturbation_stability} hold for $\delta_X$ as well, since they only argue about ordered BDTs.
We ignore the volume labeling, as it breaks the purely topological framework.
Differences between $\delta_X$ and $\delta_W$, even when using the same labels and cost function, are discussed in Section~\ref{sec:conclusion}.

\subsection{Branch Mapping Distance}

Next, we define the branch mapping distance.
Its stability is discussed in Section~\ref{sec:perturbation_stability_app}.
The branch mapping distance $\delta_B$ is the first real outlier in the considered distances.
While closely related to the path mapping distance algorithmically, it does not fit the framework of edit distances described in Section~\ref{sec:preliminaries}.
In fact, it does not even define a metric.
However, we will show that it behaves similar to the path mapping distance in terms of stability, too.

Formally, $\delta_B$ is defined as the cost of an optimal branch mapping between two abstract merge trees.
A mapping between two (arbitrary, not persistence-based) branch decompositions is branch mapping if
\begin{itemize}
    \item it is a one-to-one mapping,
    \item it maps the two root branches onto each other,
    \item if $a$ is mapped to $b$ then the parents $a'$ and $b'$ are mapped onto each other,
    \item the ancestor relation and sibling ordering in the corresponding BDTs are preserved.
\end{itemize}
For more details, see~\cite{wetzels2022branch}.
Note that the definition does not require a specific branch decomposition.
Thus, the algorithm in~\cite{wetzels2022branch} computes the two optimal branch decompositions as well as the optimal mapping to obtain minimal costs.
The costs are defined as the sum of costs for each mapped pair of branches.
We use the Wasserstein metric as the costs for a mapped pair, though the original paper defines several alternatives as well.
We say that two branches mapped onto each other are relabeled, whereas branches not present in the mapping are inserted/deleted.

\subsection{Sridharamurthy Distance}

Lastly, we give a definition of $\delta_S$, the original merge tree edit distance by Sridharamurthy et al.~\cite{DBLP:journals/tvcg/SridharamurthyM20}.
In contrast to all other distances, it works on merge trees directly while using labels derived from the persistence-based branch decomposition.
An edit sequence between trees $T_1,T_2$ is a constrained edit sequence if the induced mapping $M \subset V(T_1)\times V(T_2)$ between the nodes of the trees fulfills the following properties:
\begin{itemize}
    \item $M$ is a one-to-one mapping
    \item $M$ is ancestor-preserving
    \item For any triple $(v_1,u_1),(v_2,u_2),(v_3,u_3)$ in $M$, the lowest common ancestor of $v_1,v_2$ is not an ancestor or a descendant of $v_3$ if and only if the lowest common ancestor of $u_1,u_2$ is not an ancestor or a descendant of $u_3$.
\end{itemize}
$\delta_S$ is defined as the constrained edit distance~\cite{DBLP:journals/algorithmica/Zhang96} using the following labels and cost function.
Each node is labeled with the birth and death values of its corresponding branch.
For a maximum node $v$, this is the branch ending in $v$.
For a saddle or the global minimum $v$, it is the longest branch starting in $v$.
The authors of~\cite{DBLP:journals/tvcg/SridharamurthyM20} describe multiple cost functions.
For better comparability, we use the same one as for $\delta_W$, which was also proposed in the original paper. Note that in the case of no multi-saddles, for each branch there are two representing nodes.
If there are multi-saddles, then there are also branches only represented by one node.

\section{Remaining Stability Arguments}
\label{sec:perturbation_stability_app}

\subsection{Simple Changes}
\label{sec:stability_simple_app}

We now provide the arguments about stability against simple changes in more detail.
Let $f,g$ be two scalar fields with a minimal vertex perturbation of vertex $x$ by $\epsilon$ between them, such that $T_f \subseteq T_g$ and $\bdt(T_f) \subseteq \bdt(T_g)$.
If $x$ is not a node in one of the merge trees, all distances are $0$ and therefore stable.
Otherwise, we consider different types of edit distances separately.
We have two cases to consider: either $T_f = T_g$ or $T_f \neq T_g$.

\medskip
\noindent
\textbf{BDT-based distances.}
We begin with the BDT-based distances: the Saikia distance and the Wasserstein distance.
Let us first consider the case where $E(T_f) = E(T_g)$ and $E(\bdt(T_f)) = E(\bdt(T_g))$ (up to renaming of a single node), but the scalar value of a node changes.
If the changed node is a maximum or the global minimum, exactly one persistence pair or branch is changed: we can express the change through a single relabel operation.
For both distances, the cost of changing this branch is bounded by $\epsilon$: either the birth or death value is shifted by at most $\epsilon$, meaning that the point in the birth-death space is also shifted by at most $\epsilon$.
If the changed node is a saddle, multiple branches are changed.
However, for each such branch, only the death value is changed by at most $\epsilon$.
Thus, we perform a relabel operation for each changed branch, each cost is clearly bounded by $\epsilon$, yielding a total cost bound of $\deg(T_f)\cdot\epsilon$.

Next, we consider the case where $E(T_f) \neq E(T_g)$.
By Lemma~\ref{lemma:changes}, we know that a single branch is added to or removed from the BDTs.
All other branches remain unchanged.
Since the added or removed node in the BDT is a leaf (by definition of the subgraph relation), this change can be expressed as a single insertion or deletion of a leaf, which is a valid one-degree operation.
The length of the new branch and thereby the cost of the operation is clearly bounded by $\epsilon$.

\medskip
\noindent
\textbf{Classic edit distances.}
The next distances are the classic edit distances: the contour tree alignments and their generalization.
We first consider a simple shift of a node.
Depending on the critical type of the perturbed vertex, either one or several edges change their length.
In all cases, for each edge, only one node per edge is moved by $\epsilon$.
Thus, for each edge, it suffices to apply a relabeling of cost $\epsilon$.
If a node is added or removed, its distance to its parent is bounded by $\epsilon$, and so are the cost of the corresponding insertion or deletion.
Again, by definition of the subgraph relation or Lemma~\ref{lemma:changes}, the node is a leaf and the edit operation is valid in all constrained variants.
We get a total cost bound of $\deg(T_f)\cdot\epsilon$.

\medskip
\noindent
\textbf{Deformation-based distances.}
Here, all arguments for the classic edit distances hold as well, since for the considered cases, the same operations apply.

\medskip
\noindent
\textbf{Branch mapping distance.}
As for classic edit distances, we first consider a simple shift of a node.
Here, the tree and the branch persistence-based branch decomposition remain equal structurally.
Thus, we use the persistence-based branch decomposition for both trees and the identity mapping, which fulfills all branch mapping properties.
Depending on the critical type of the perturbed vertex, either one or several branches change their length, for any branch decomposition.
In all cases, for the amount of branches that are changed is bounded by $\deg(T_f)$ and each relabeling cost is bounded by $\epsilon$ (same argument as for the BDT-based distances).
If a node is added or removed, its distance to its parent is bounded by $\epsilon$, and so are the cost of the corresponding insertion or deletion.
Again, we choose the persistence-based branch decomposition where the inserted branch is a leaf branch.
The corresponding mapping induced by the identity and one leaf branch insertion is a valid branch mapping.
We get a total cost bound of $\deg(T_f)\cdot\epsilon$.

\medskip
\noindent
\textbf{Sridharamurthy distance.}
We again consider the two cases $T_f = T_g$ or $T_f \neq T_g$ separately.
If $T_f = T_g$ and a maximum or the global minimum changes its scalar value, only one single branch is changed.
Thus, at most two nodes change their label and this change can again be bound by $\epsilon$.
If $T_f = T_g$ and a saddle node $v$ changes its scalar value, exactly the $\deg(T_f)-1$ branches originating in $v$ are changed.
Thus, the corresponding $\deg(T_f)-1$ maxima and $v$ itself need to be relabeled.
Each relabel cost is bounded by $\epsilon$.
In both cases, the corresponding edit mapping is the identity, so clearly a valid constrained edit mapping.
If $T_f \neq T_g$, a single new leaf is added as a child of an existing saddle.
This also adds a new branch of length bounded by $\epsilon$ to the branch decomposition.
The new node is labeled by the birth and death values of the new branch, so the insertion cost is also bounded by $\epsilon$.
Overall, the costs can be bounded by $\deg(T_f)\cdot\epsilon$.
The corresponding mapping is the identity on all nodes except the new one, which clearly fulfills the first two properties of a constrained edit mapping.
Since the new node cannot be the lowest common ancestor of any two nodes, the third property is also trivially fulfilled.

\subsection{Edge Splits}

Next, we consider the remaining arguments for stability against edge splits.
Let $f,g$ be two scalar fields with a minimal vertex perturbation of $x$ by $\epsilon$ between them, such that $\obdt(T_f) \subset \obdt(T_g)$ (inclusion strict by Lemma~\ref{lemma:changes}) but neither $T_f \subseteq T_g$ nor $T_f \subseteq T_g$ hold.

\medskip
\noindent
\textbf{Branch mapping distance.}
By Lemma~\ref{lemma:changes}, only one type of change can happen: a single edge is split by a new saddle with a new child.
We can use the persistence-based branch decomposition, in which a leaf branch is inserted.
For all remaining branches, we use the identity mapping and the new branch is unmapped.
We obtain a valid branch mapping and the costs are bound by the insertion cost, in particular also bound by~$\epsilon$.

\medskip
\noindent
\textbf{Sridharamurthy distance.}
By Lemma~\ref{lemma:changes}, only one type of change can happen: a single edge is split by a new saddle with a new child.
In the persistence-based branch decomposition a single leaf branch is inserted.
In the constrained edit sequence, we insert the new saddle node $y$, then the new leaf $x$.
This gives us the identity mapping on $V(T_f)\setminus\{x,y\}$ while $x$ and $y$ are unmapped.
Since $x$ is not in the mapping and $y$ only has one other child, $y$ cannot be the lowest common ancestor of any two nodes in the mapping.
Hence, all three requirements for a constrained edit mapping are fulfilled.

\subsection{Vertical Swaps}

Next, we consider the remaining arguments for stability against vertical swaps.
Let $f,g$ be two scalar fields with a minimal vertex perturbation of $x$ by $\epsilon$ between them, such that $T_f \subseteq T_g$ but neither $\obdt(T_f) \subseteq \obdt(T_g)$ nor $\obdt(T_f) \subseteq \obdt(T_g)$ hold.

\medskip
\noindent
\textbf{BDT-based distances.}
For the BDT-based distances, we show instability.
By Lemma~\ref{lemma:changes}, we know that a branch $a$ is swapped with its unique parent branch $b$.
Constrained edit distances on BDTs cannot handle such changes, i.e.\ they will always match parent to parent and child to child, as the matching has to be ancestor-preserving.
This induces increased costs in case the two branches have other descendants.
As has been illustrated in~\cite{wetzels2022branch}, there are simple counter examples.
Figure~\ref{fig:vertical_instability} shows such an example in accordance with the case distinction in Appendix~\ref{sec:case_distinction} and Lemma~\ref{lemma:changes}.
Both BDT-based distances have to map the two roots, inducing cost of $\epsilon$.
To preserve the ancestor relation of an edit mapping, they cannot map the remaining branches onto each other.
Thus, at least one such branch has to be inserted and deleted, inducing costs of $2x$.
Since $x$ can be chosen freely and $\epsilon$ is fixed, there is no constant $c$ such that the costs are bound by $c \cdot \epsilon$.

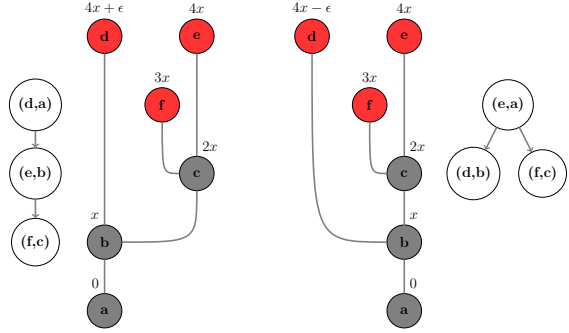
\begin{figure}
    \centering
    \resizebox{0.5\linewidth}{!}{
    \Large
    \bfseries
    \begin{tikzpicture}[xscale=0.8,yscale=0.6]
    
    \node[draw,circle,fill=gray!100,minimum width=1.2cm,label={[above left]$0$}] at (-2, 0) (root_1) {\textbf{a}};
    \node[draw,circle,fill=gray!100,minimum width=1.2cm,label={[above left]$x$}] at (-2, 4) (s1_1) {\textbf{b}};
    \node[draw,circle,fill=red!80,minimum width=1.2cm,label={$4x+\epsilon$}] at (-2, 16) (m1_1) {\textbf{d}};
    \node[draw,circle,fill=red!80,minimum width=1.2cm,label={$4x$}] at (2, 16) (m2_1) {\textbf{e}};
    \node[draw,circle,fill=gray!100,minimum width=1.2cm,label={[above right]$2x$}] at (2, 8) (s2_1) {\textbf{c}};
    \node[draw,circle,fill=red!80,minimum width=1.2cm,label={$3x$}] at (0.5, 12) (m3_1) {\textbf{f}};
    
    \draw[gray,ultra thick] (root_1) -- (s1_1);
    \draw[gray,ultra thick] (s1_1) -- (m1_1);
    \draw[gray,ultra thick] (s2_1) -- (m2_1);
    \draw[gray,ultra thick] (s1_1) .. controls (2,4) .. (s2_1);
    \draw[gray,ultra thick] (s2_1) .. controls (0.5,8) .. (m3_1);

    \begin{scope}[shift={(9,0)}]
    
    \node[draw,circle,fill=gray!100,minimum width=1.2cm,label={[above right]$0$}] at (2, 0) (root_1) {\textbf{a}};
    \node[draw,circle,fill=gray!100,minimum width=1.2cm,label={[above right]$x$}] at (2, 4) (s1_1) {\textbf{b}};
    \node[draw,circle,fill=red!80,minimum width=1.2cm,label={$4x-\epsilon$}] at (-2, 16) (m1_1) {\textbf{d}};
    \node[draw,circle,fill=red!80,minimum width=1.2cm,label={$4x$}] at (2, 16) (m2_1) {\textbf{e}};
    \node[draw,circle,fill=gray!100,minimum width=1.2cm,label={[above right]$2x$}] at (2, 8) (s2_1) {\textbf{c}};
    \node[draw,circle,fill=red!80,minimum width=1.2cm,label={$3x$}] at (0.5, 12) (m3_1) {\textbf{f}};
    
    \draw[gray,ultra thick] (root_1) -- (s1_1);
    \draw[gray,ultra thick] (s1_1) .. controls (-2,4) .. (m1_1);
    \draw[gray,ultra thick] (s2_1) -- (m2_1);
    \draw[gray,ultra thick] (s1_1) .. controls (2,4) .. (s2_1);
    \draw[gray,ultra thick] (s2_1) .. controls (0.5,8) .. (m3_1);

    \end{scope}

    \begin{scope}[shift={(-5,0)}]
        \node[draw,circle] at (0, 12) (n1) {(d,a)};
        \node[draw,circle] at (0, 8) (n2) {(e,b)};
        \node[draw,circle] at (0, 4) (n3) {(f,c)};
        \draw[->,gray,ultra thick] (n1) -- (n2);
        \draw[->,gray,ultra thick] (n2) -- (n3);
    \end{scope}

    \begin{scope}[shift={(14,0)}]
        \node[draw,circle] at (1.5, 12) (n1) {(e,a)};
        \node[draw,circle] at (0, 8) (n2) {(d,b)};
        \node[draw,circle] at (3, 8) (n3) {(f,c)};
        \draw[->,gray,ultra thick] (n1) -- (n2);
        \draw[->,gray,ultra thick] (n1) -- (n3);
    \end{scope}
    
    \end{tikzpicture}
    }
    \caption{Two merge trees (inner trees) together with their BDTs (outer trees). They can be transformed into each other by a minimal vertex perturbation of $d$.}
    \label{fig:vertical_instability}
\end{figure}

\medskip
\noindent
\textbf{Branch mapping distance.}
Since we can pick an arbitrary branch decomposition for both trees, we pick the same one for both trees and use the identity mapping, a valid branch mapping. Only a single branch is relabeled, meaning the costs are bounded by $\epsilon$.

\medskip
\noindent
\textbf{Sridharamurthy distance.}
Consider again the example in Figure~\ref{fig:vertical_instability}.
The labels used by $\delta_S$ are the following:
in $T_1$ we have $\ell_1(d)=\ell_1(a)=(4x+\epsilon,0), \ell_1(e)=\ell_1(b)=(4x,x), \ell_1(f)=\ell_1(c)=(3x,2x)$, in $T_2$ we have $\ell_2(d)=\ell_1(b)=(4x-\epsilon,x), \ell_2(e)=\ell_1(a)=(4x,0), \ell_2(f)=\ell_1(c)=(3x,2x)$.
While the label of $a,b,c,f$ remain almost unchanged, the labels of $d$ and $e$ essentially swap.
Any ancestor-preserving mapping has to map the three leaves onto each other as well as the three inner nodes.
Since the cost of $(d,d),(e,e),(d,f),(e,f)$ are at least $x$, there is exactly one such mapping that induces costs lower than $x$: $(d,e),(e,d)$ and $(v,v)$ $\forall v \notin\{d,e\}$.
However, this is not an edit mapping: mapping $c$ to $c$ but $e$ to $d$ breaks the ancestor-preservation.
Thus, any (constrained) edit mapping has costs of at least $x$ and with $x$ variable there is no constant $c$ such that the costs are bound by $c \cdot \epsilon$.

\subsection{Horizontal Swaps}
Let $f,g$ be two scalar fields with a minimal vertex perturbation of $x$ by $\epsilon$ between them, such that neither $\obdt(T_f) \subseteq \obdt(T_g)$ nor $\obdt(T_g) \subseteq \obdt(T_f)$ and also neither $T_f \subseteq T_g$ nor $T_f \subseteq T_g$ hold.
We now show that neither the branch mapping distance nor the Sridharamurthy distance are stable in this case.

\medskip
\noindent
\textbf{Branch mapping distance.}
We again consider the counter example in Figure~\ref{fig:horizontal_instability}.
Any branch mapping has to map the root branches onto each other.
To avoid costs larger than $x$, the maximum of the main branch has to be chosen equally in both trees.
If we choose $c$, then the ordering of $d$ and $f$ changes between the two trees (in the first tree, one has to branch from the other, while in the second tree the both branch from $(c,a)$), inducing costs of at least $x$.
The same holds for main branch $(d,a)$.
If we pick $f$, then in both trees the other two branches branch from $(f,a)$, but their orderings swaps, again inducing costs of at least $x$.
Thus, there is no constant $c$ such that the costs are bound by $c \cdot \epsilon$.

\medskip
\noindent
\textbf{Sridharamurthy distance.}
We again consider the counter example in Figure~\ref{fig:horizontal_instability}.
The labels used by $\delta_S$ are the following:
in $T_1$ we have $\ell_1(c)=\ell_1(a)=(4x,0), \ell_1(d)=\ell_1(b)=(3x,x-\epsilon), \ell_1(f)=\ell_1(e)=(2x,x)$, in $T_2$ we have $\ell_1(c)=\ell_1(a)=(4x,0), \ell_1(d)=\ell_1(e)=(3x,x), \ell_1(f)=\ell_1(b)=(2x,x+\epsilon)$.
Any bijective ancestor-preserving mapping has to map the three leaves onto each other as well as the three inner nodes.
Any non-bijective mapping induces costs of at least $x$.
Since the relabel costs between any non-equal branches are at least $x$, the only mapping of cost smaller than $x$ is the identity, which is not a valid edit mapping on this pair of trees.
Thus, any (constrained) edit mapping has costs of at least $x$ and with $x$ variable there is no constant $c$ such that the costs are bound by $c \cdot \epsilon$.

\end{appendices}

\bibliographystyle{plain}
\bibliography{main}

\end{document}